\documentclass[border=2px,11pt]{article}

\usepackage{indentfirst}
\usepackage[scale=0.8]{geometry}
\usepackage{authblk}

\usepackage{amsmath,amssymb,amsthm,amsfonts,mathdots}
\usepackage{latexsym,bm}
\usepackage{color}
\usepackage{arydshln}

\usepackage{graphicx}
\usepackage{float}

\usepackage{tikz}
\usetikzlibrary{decorations.pathreplacing}

\usepackage[ruled]{algorithm2e}

\usepackage{hyperref}

\newtheorem{theorem}{Theorem}[section]
\newtheorem{lemma}{Lemma}[section]

\newtheorem{definition}{Definition}[section]
\newtheorem{example}{\bf Example}[section]
\theoremstyle{remark}
\newtheorem{remark}{\bf Remark}[section]
\numberwithin{equation}{section}

\def \d {{\rm{d}}}
\def \e {{\rm{e}}}

\newcommand{\ket}[1]{| #1 \rangle} 
\newcommand{\bra}[1]{\langle #1 |} 

\DeclareMathOperator{\polylog}{polylog}

\author[1,2,3]{Shi Jin\thanks{shijin-m@sjtu.edu.cn}}
\author[1]{Shuyi Zhang\thanks{shuyi-zhang@sjtu.edu.cn}}
\affil[1]{\small School of Mathematical Sciences, Shanghai Jiao Tong University, Shanghai, 200240, China}
\affil[2]{\small Institute of Natural Sciences, Shanghai Jiao Tong University, Shanghai 200240, China}
\affil[3]{\small Ministry of Education Key Laboratory in Scientific and Engineering Computing, Shanghai Jiao Tong University, Shanghai 200240, China}

\date{}

\begin{document}

\title{Quantum simulation of the Liouville equation in classical mechanics with discontinuous potential via Schr\"odingerization}

\maketitle

\begin{abstract}
    We develop quantum simulation algorithms for the Liouville equation of classical mechanics with discontinuous potential. Such discontinuities represent potential barriers at which classical particles undergo energy preserving transmission or reflection, and the resulting interface conditions must be incorporated into the numerical flux. We combine Hamiltonian-preserving schemes in \cite{jinHamiltonianPreservingSchemesLiouville2005} with the Schr\"odingerization method, which embeds the resulting nonunitary semi-discrete dynamics into a unitary Schr\"odinger type system in one additional auxiliary variable \cite{jinQuantumSimulationPartial2023,jinQuantumSimulationPartial2024a}.
    For one-, two-, and \(n\)-dimensional problems with grid aligned interfaces, we construct sparse matrix representations of the transmission and reflection fluxes using step and hat functions, derive the corresponding Hamiltonians of the Schr\"odingerized systems, and analyze their sparse-access query complexity. In the sparse-access oracle model, the resulting algorithms have a {\it polynomial} dependence on the inverse accuracy and avoid the {\it exponential} dependence on the phase-space dimension suffered by classical grid based Hamiltonian-preserving schemes, up to the cost of implementing the oracles and the postselection overhead. We also describe the postselected recovery of the physical solution state and the quantum readout of macroscopic observables such as density and averaged velocity through overlap estimation. Numerical experiments based on classical simulation of the Schr\"odingerized dynamics validate the proposed formulation and illustrate the correct transmission/reflection behavior at potential barriers.
\end{abstract}



\section{Introduction}

The Liouville equation is a fundamental  equation in classical mechanics describing the evolution of phase-space densities under Hamiltonian dynamics.
This equation expresses the conservation of probability along characteristics generated by the Hamiltonian flow and reflects the incompressibility of phase-space transport. The Liouville equation has widespread applications in various scientific and engineering disciplines. In optics, it is employed to model light propagation through complex refractive media, lens systems, and waveguides \cite{kurtbernardowolfGeometricOpticsPhase2004}. In plasma physics, it describes electromagnetic wave transport in inhomogeneous plasmas \cite{ryzhikTransportEquationsElastic1996,engquistComputationalHighFrequency2003}. In geophysics and atmospheric science, analogous equations are used for the propagation of seismic waves, acoustic waves, and radiative transfer \cite{ryzhikTransportEquationsElastic1996,engquistComputationalHighFrequency2003}. Moreover, the Liouville equation serves as a bridge between classical and quantum descriptions, closely related to the Wigner transform and semiclassical limits of the Schr\"odinger equation \cite{JMS-Acta}.

The $d$-dimensional Liouville equation in classical mechanics can be expressed as a first-order linear hyperbolic partial differential equation \cite{jinAsymptoticpreservingSchemesMultiscale2022}, formulated as 
\begin{equation}\label{equ:ddimLiouville}
    f_{t}+\boldsymbol{v}\cdot\nabla_{\boldsymbol{x}}f-\nabla_{\boldsymbol{x}}V\cdot\nabla_{\boldsymbol{v}}f=0,\quad t>0,\quad\boldsymbol{x},\boldsymbol{v}\in R^{d},
\end{equation} where $f(t, \boldsymbol{x}, \boldsymbol{v})$ is the density distribution of a classical particle at position $\boldsymbol{x}$, time $t$ and travelling with velocity $\boldsymbol{v}$. $V(\boldsymbol{x})$ is the potential. The Liouville equation is the Eulerian formulation of Newton's second law: 
\begin{equation}\label{equ:ddimNewtonLaw}
    {\frac{d\boldsymbol{x}}{dt}}=\boldsymbol{v},\quad{\frac{d\boldsymbol{v}}{dt}}=-\nabla_{\boldsymbol{x}}V,
\end{equation}
which is a Hamiltonian system with the Hamiltonian
\begin{equation}\label{equ:LiouvilleHamilClassical}
    H_{\rm cla}=\frac{1}{2}|\boldsymbol{v}|^2+V(\boldsymbol{x}).
\end{equation}
The Newton's law is the Lagrangian formulation of classical particles, and provides the bi-characteristics of the Liouville equation along which $f$ is a constant. 

 If $V(\boldsymbol{x})$ is smooth, then the initial value problem to \eqref{equ:ddimNewtonLaw} is well-posed, and a standard numerical method (for example, the upwind scheme and its higher order extensions) for linear wave equations give satisfactory results.
However, if $V(\boldsymbol{x})$ is discontinuous corresponding to a potential barrier then the characteristic speed of the Liouville equation given by \eqref{equ:ddimNewtonLaw} is infinity at the discontinuous point.

Potential barriers appear in many important physical problems, such as  quantum tunnelling (and quantum dots) in semiconductor device modelling, plasmas, and geometrical optics through different materials. Liouville or Vlasov equations describe the density distribution of particles in such a heterogeneous medium. For some recent mathematical studies of discontinuous potential in high frequency waves see \cite{guillaumebalTransportTheoryAcoustic1999,millerRefractionHighfrequencyWaves2000,ryzhikTransportEquationsElastic1996}

Numerical solutions to the Liouville  equation with discontinuous potential face significant challenges: First, the high dimensionality of phase space (six dimensions for 3D problems) makes grid based numerical methods suffer from the curse of dimensionality, where computational cost grows exponentially with dimension.    Lagrangian or particle methods can be used to avoid the curse-of-dimensionality problem.   Examples include wavefront methods \cite{engquistComputationalHighFrequency2003} and symplectic integrators \cite{sanz-sernaRungekuttaSchemesHamiltonian1988,kangHamiltonianWayComputing1991,leimkuhlerSimulatingHamiltonianDynamics2005}. These approaches represent solutions as collections of moving particles, making them inherently suitable for high-dimensional problems. However, in regions where there are few particles, one has to add more particles and do interpolations to approximate the information for these new particles,  which is difficult to implement, not to mention the low-order accuracy of particle methods.  Second, through the potential barrier, one has to develop numerical strategy to suitably account for particle transmissions and reflection, which is not an issue for problems with smooth potential. Finally, relevant to this work, if one wants to solve it on a quantum computer, the nonlinearity of the Hamiltonian system poses major challenges, since quantum computer, designed using quantum mechanics element, can only efficiently handle linear problems. 

To deal with the potential barriers, 
the Hamiltonian-preserving schemes were introduced in \cite{jinHamiltonianPreservingSchemesLiouville2005,jinHamiltonianPreservingSchemeLiouville2006,jinHamiltonianpreservingSchemesLiouville2006}. By incorporating the particle transmission and reflection into the numerical flux, which is equivalent to Hamiltonian preservation across the barrier,  they provide a convenient way to handle the potential barrier in a physically consistent way, with  a typical hyperbolic stability condition $\Delta t= O(\Delta x,\Delta v)$. However, the curse-of-dimensionality remains a bottleneck for classical simulation. Quantum computations, due to their potential polynomial to even exponential advantages over their classical counterparts, offer a promising computational paradigm for handling the high-dimensional problems. 

In recent years, a novel universal method known as Schr\"odingerization has been proposed \cite{jinQuantumSimulationPartial2024a,jinQuantumSimulationPartial2023}. It first employs a warped phase transformation, which, by adding an extra dimension to the equation, converts any linear PDE with non-unitary dynamics into a Schr\"odinger-like system of equations with unitary dynamics, enabling subsequent Hamiltonian simulation. 
It has been applied to equations with physical boundary or interface conditions \cite{jinQuantumSimulationPartial2024}, quantum dynamics with artificial boundary conditions \cite{jinQuantumSimulationQuantum2024}, among other applications. 
By choosing smoother initial functions in auxiliary space, Schr\"odingerization can in fact achieve near optimal and even optimal scaling in matrix queries \cite{jinSchrodingerizationMethodLinear2025b}. Apart from implementation based on qubits, Schr\"odingerization method can also be extended to the qumode framework \cite{jinAnalogQuantumSimulation2024a,jinAnalogQuantumSimulation2024}, which is suitable for analog quantum computing.

In this paper, we combine Hamiltonian-preserving schemes in \cite{jinHamiltonianPreservingSchemesLiouville2005} with Schr\"odingerization to construct quantum simulation algorithms for Liouville equations with discontinuous potential. The main contributions are as follows: 
\begin{itemize}
    \item We derive a Schr\"odingerization compatible sparse matrix formulation for the Hamiltonian-preserving scheme of the Liouville equation in one, two, and \(n\) spatial dimensions with discontinuous potential by encoding the nonlinear transmission/reflection rule through step functions and compactly supported hat functions.
    \item We analyze the sparsity of the resulting semi-discrete matrices. In contrast to the geometrical optics case \cite{jinQuantumSimulationLiouville2026}, where the velocity mapping across an interface is essentially a rescaling determined by the wave speed, here the energy preserving map contains square root transformations and threshold indicators. We therefore give separate row and column sparsity estimates for the interface matrices, with detailed proofs in Appendix \ref{app:sparsity}.
    \item We derive query complexity estimates in one, two, and \(n\) spatial dimensions, and compare them with the corresponding classical grid based complexity. In the sparse-access oracle model, the resulting algorithms achieve {\it polynomial} scaling in the inverse accuracy and avoid the {\it exponential} dependence on the phase-space dimension suffered by classical grid based Hamiltonian-preserving schemes, up to the cost of implementing the oracles and the postselection overhead.
    \item We describe how the physical solution state is recovered by postselection and how macroscopic observables such as density, momentum, and averaged velocity can be estimated without reconstructing the full phase-space vector.
\end{itemize} 

A closely related problem was studied in \cite{jinQuantumSimulationLiouville2026}, where Schr\"odingerization was combined with Hamiltonian-preserving schemes for the Liouville equation in geometrical optics with {\it partial} transmission and reflection at material interfaces.  The present work addresses a different physical setting.  Here the Liouville equation comes from classical mechanics with Hamiltonian \(H_{\rm cla}=|\boldsymbol{v}|^2/2+V(\boldsymbol{x})\), and the interface rule is the specular energy transmission/reflection law generated by discontinuities of the potential \(V\).  This leads to a different flux construction and different sparse matrix representations.  Compared with   \cite{jinQuantumSimulationLiouville2026}, the distinction is twofold.
At the model level, the present interface law is not a partial splitting of rays with reflection and transmission coefficients. A classical particle encountering a potential jump is transmitted only when its kinetic energy exceeds the barrier, in which case the normal velocity is determined by the nonlinear square root map imposed by conservation of \(|\boldsymbol v|^2/2+V(\boldsymbol x)\); otherwise it is specularly reflected. This energy-threshold rule leads to discretization matrices with different nonzero patterns and column sparsity behavior from the velocity rescaling in geometrical optics. At the algorithmic and analytical level, the present paper also provides detailed sparsity estimates for the energy preserving velocity maps, extends the construction to grid aligned \(n\)-dimensional interfaces, and describes a postselected recovery and overlap based readout procedure for macroscopic observables.

The rest of this paper is organized as follows. Section \ref{sec:introLiouville} reviews the behavior of classical particles at a potential barrier and the corresponding interface condition for the density function. Section \ref{sec:Schr} recalls the Schr\"odingerization framework and the sparse-access Hamiltonian simulation model. Sections \ref{sec:interface1d} and \ref{sec:interface2d} construct the Schr\"odingerization based Hamiltonian-preserving algorithms \cite{jinHamiltonianPreservingSchemesLiouville2005} in one and two spatial dimensions, respectively. Section \ref{sec:interfacend} extends the construction to $n$ spatial dimensions for grid-aligned interfaces. Section \ref{sec:recovery} discusses the postselected recovery of the physical solution state and the quantum readout of macroscopic observables. Numerical examples are presented in Section \ref{sec:Numerical}. Section \ref{sec:conclusion} concludes the paper.

\section{Behavior of a classical particle at a potential barrier}\label{sec:introLiouville}

In classical mechanics, a particle will either cross a potential barrier with a changing momentum, or be reflected, depending on its momentum and on the strength of the potential barrier. The Hamiltonian \eqref{equ:LiouvilleHamilClassical} should be preserved across the potential barrier:
\begin{equation}\label{equ:ddimHamilPreserve}
    \frac{1}{2}(\boldsymbol{v}^+)^2+V^+=\frac{1}{2}(\boldsymbol{v}^-)^2+V^-,
\end{equation}
where the superscripts $\pm$ indicate the right and left limits of the quantity at the potential barrier.

For example, we consider the 1D case, then at a potential discontinuity \eqref{equ:ddimHamilPreserve} becomes 
\begin{equation}\label{equ:1dimHamilPreserve}
    \frac{1}{2}(\xi^+)^2+V^+=\frac{1}{2}(\xi^-)^2+V^-.
\end{equation}
Suppose the characteristic on the left of the potential discontinuity is given as a constant velocity $\xi^->0$. There are three possibilities (see Figure \ref{fig:potentialbarrier}) :
\begin{enumerate}
    \item $V^->V^+$. In this case, the potential decreases, so the particle will cross the potential barrier and gain momentum in order to maintain a constant Hamiltonian. \eqref{equ:1dimHamilPreserve} implies 
    \begin{equation}
        \xi^+=\sqrt{(\xi^-)^2+2(V^--V^+)}.
    \end{equation}
    \item $V^-<V^+$ and $\frac12(\xi^-)^2>V^+-V^-.$ If the kinetic energy of the particle is bigger than the potential jump then the particle will cross the barrier with a reduced momentum. \eqref{equ:1dimHamilPreserve} implies
    \begin{equation}
        \xi^+=\sqrt{(\xi^-)^2-2(V^+-V^-)}.
    \end{equation}
    \item $V^-<V^+$ and $\frac12(\xi^-)^2<V^+-V^-.$ In this case, the kinetic energy is not large enough for the particle to cross the potential barrier, so the particle will be reflected with a negative velocity $-\xi^-$.
\end{enumerate}

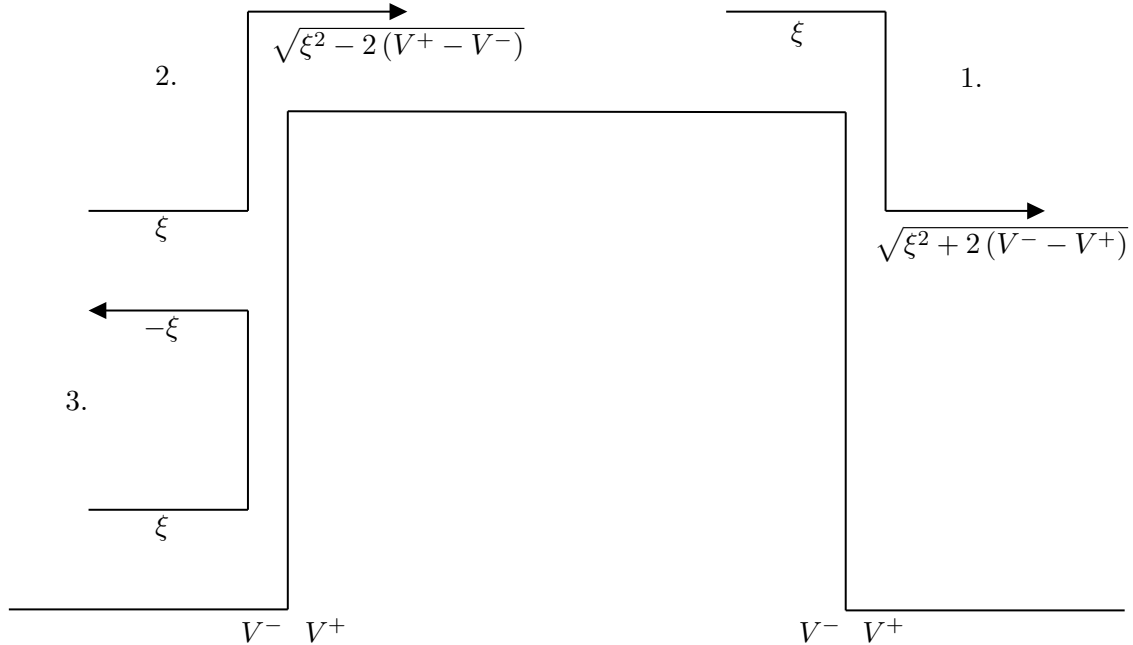
\begin{figure}[!htb]
    \centering

    \tikzset{every picture/.style={line width=0.75pt}} 

    \begin{tikzpicture}[x=0.75pt,y=0.75pt,yscale=-1,xscale=1]

    \draw    (190,100.35) -- (470,100.85) ;
    \draw    (50,350.35) -- (190,350.35) ;
    \draw    (470,350.85) -- (610,350.85) ;
    \draw    (190,100.35) -- (190,350.35) ;
    \draw    (470,100.85) -- (470,350.85) ;
    \draw    (90,300.35) -- (170,300.35) ;
    \draw    (170,200.35) -- (170,300.35) ;
    \draw    (170,200.35) -- (93,200.35) ;
    \draw [shift={(90,200.35)}, rotate = 360] [fill={rgb, 255:red, 0; green, 0; blue, 0 }  ][line width=0.08]  [draw opacity=0] (8.93,-4.29) -- (0,0) -- (8.93,4.29) -- cycle    ;
    \draw    (90,150.35) -- (170,150.35) ;
    \draw    (170,50.35) -- (170,150.35) ;
    \draw    (170,50.35) -- (247,50.35) ;
    \draw [shift={(250,50.35)}, rotate = 180] [fill={rgb, 255:red, 0; green, 0; blue, 0 }  ][line width=0.08]  [draw opacity=0] (8.93,-4.29) -- (0,0) -- (8.93,4.29) -- cycle    ;
    \draw    (410,50.35) -- (490,50.35) ;
    \draw    (490,50.35) -- (490,150.35) ;
    \draw    (490,150.35) -- (567,150.35) ;
    \draw [shift={(570,150.35)}, rotate = 180] [fill={rgb, 255:red, 0; green, 0; blue, 0 }  ][line width=0.08]  [draw opacity=0] (8.93,-4.29) -- (0,0) -- (8.93,4.29) -- cycle    ;

    \draw (165,353.4) node [anchor=north west][inner sep=0.75pt]  [font=\normalsize]  {$V^{-}$};
    \draw (445,353.4) node [anchor=north west][inner sep=0.75pt]  [font=\normalsize]  {$V^{-}$};
    \draw (197,353.4) node [anchor=north west][inner sep=0.75pt]  [font=\normalsize]  {$V^{+}$};
    \draw (477,353.4) node [anchor=north west][inner sep=0.75pt]  [font=\normalsize]  {$V^{+}$};
    \draw (526,77) node [anchor=north west][inner sep=0.75pt]   [align=left] {1.};
    \draw (122,76) node [anchor=north west][inner sep=0.75pt]   [align=left] {2.};
    \draw (77,239) node [anchor=north west][inner sep=0.75pt]   [align=left] {3.};
    \draw (122,151.4) node [anchor=north west][inner sep=0.75pt]  [font=\normalsize]  {$\xi $};
    \draw (441,53.4) node [anchor=north west][inner sep=0.75pt]  [font=\normalsize]  {$\xi $};
    \draw (122,301.4) node [anchor=north west][inner sep=0.75pt]  [font=\normalsize]  {$\xi $};
    \draw (116,201.4) node [anchor=north west][inner sep=0.75pt]  [font=\normalsize]  {$-\xi $};
    \draw (180,56.4) node [anchor=north west][inner sep=0.75pt]  [font=\normalsize]  {$\sqrt{\xi ^{2} -2\left( V^{+} -V^{-}\right)}$};
    \draw (483,156.4) node [anchor=north west][inner sep=0.75pt]  [font=\normalsize]  {$\sqrt{\xi ^{2} +2\left( V^{-} -V^{+}\right)}$};

    \end{tikzpicture}

    \caption{Change of particle momentum across a potential barrier for the case when $\xi^- > 0$.}
    \label{fig:potentialbarrier}
\end{figure}

If $\xi^- < 0$, similar behavior can also be analyzed using the constant Hamiltonian condition \eqref{equ:1dimHamilPreserve}.

For the 2D case, we suppose the discontinuity of $V$ aligns with the grids, then the particle behavior across the potential barrier is exactly the same as the 1D case, and the jump condition is applied in the normal direction to the barrier interface. The higher-dimensional case is similar to the two-dimensional case.

Note that the density distribution $f$ remains unchanged across the potential barrier, thus 
\begin{equation}\label{equ:ddimLiouvilleClassicalinterfacecondition}
    f(t,\boldsymbol{x}^+,\boldsymbol{v}^+)=f(t,\boldsymbol{x}^-,\boldsymbol{v}^-)
\end{equation}
at a discontinuous point $x$ of $V(x)$, where $\boldsymbol{v}^+$ and $\boldsymbol{v}^-$ are related by the constant Hamiltonian condition \eqref{equ:ddimHamilPreserve}. 

Hamiltonian preserving schemes use this mechanism for the numerical approximation to the Liouville equation with a discontinuous potential. This approximation, by its design, maintains a constant Hamiltonian up to approximation error across the potential barrier. It builds the particle transmission and reflection into the numerical flux, so it captures the correct physical behavior at the barrier very naturally. 

This interface condition differs from the one in geometrical optics
\cite{jinQuantumSimulationLiouville2026}, where an incident ray may split into transmitted and
reflected branches with reflection and transmission coefficients.  In the
classical mechanics setting considered here, the normal velocity is instead
selected by conservation of \( |\boldsymbol{v}|^2/2+V(\boldsymbol{x}) \), leading to either transmission
or reflection according to an energy threshold.

\section{Overview of the Schr\"odingerization method} \label{sec:Schr}
\subsection{The Schr\"odingerization method}

Consider the linear dynamical system
\begin{equation}\label{equ:ODElinear}
 \begin{cases}
 \displaystyle \frac{\mathrm{d}\boldsymbol{u}(t)}{\mathrm{d}t}
 = A\boldsymbol{u}(t)+\boldsymbol{b}(t),\\[2mm]
 \boldsymbol{u}(0)=\boldsymbol{u}_0,
 \end{cases}
\end{equation}
where $\boldsymbol{u}(t),\boldsymbol{b}(t)\in\mathbb{C}^n$ and
$A\in\mathbb{C}^{n\times n}$. In general, the matrix $A$ is not
Hermitian, namely $A^\dagger\neq A$, where $A^\dagger$ denotes the
conjugate transpose of $A$. For clarity of presentation, we first focus
on the homogeneous case $\boldsymbol{b}(t)=\boldsymbol{0}$. The
inhomogeneous term can be incorporated into an enlarged homogeneous
system by introducing auxiliary variables, as discussed in
\cite{jinQuantumSimulationMaxwells2024,jinSchrodingerizationBasedQuantumAlgorithms2025,huQuantumAlgorithmsMultiscale2024}.

The key observation is that any matrix $A$ admits a decomposition into
its Hermitian and anti-Hermitian components:
\begin{equation}\label{equ:Adecomposition}
    A = H_1+\mathrm{i}H_2,
    \qquad \mathrm{i}=\sqrt{-1},
\end{equation}
where
\begin{equation}\label{equ:H1H2}
    H_1=\frac{A+A^\dagger}{2},
    \qquad
    H_2=\frac{A-A^\dagger}{2\mathrm{i}}.
\end{equation}
Both $H_1$ and $H_2$ are Hermitian matrices, i.e.,
$H_1^\dagger=H_1$ and $H_2^\dagger=H_2$. 

For many semi-discrete systems arising from stable PDE discretizations,
it is natural to assume that the stability property is inherited at the
ODE level. In this setting, the Hermitian part $H_1$ is negative
semi-definite. Since $H_1$ is Hermitian, there exists a unitary matrix
$Q$ such that
\begin{equation}\label{equ:H1diagonal}
    Q^\dagger H_1 Q
    =
    \operatorname{diag}(\lambda_1,\ldots,\lambda_n)
    =:\Lambda,
    \qquad \lambda_j\leq 0,\quad j=1,\ldots,n .
\end{equation}
If $H_1$ contains positive eigenvalues, one may apply a suitable spectral
shift so that the shifted Hermitian part becomes negative semi-definite
\cite{guQuantumSimulationClass2025a}. 

The key step in the Schr\"odingerization procedure \cite{jinQuantumSimulationPartial2023,jinQuantumSimulationPartial2024a} is the warped phase transformation
\begin{equation}\label{equ:warped_phase_transform}
    \boldsymbol{v}(t,p)=g(p)\boldsymbol{u}(t),
    \qquad
    g(p)=\mathrm{e}^{-p},\quad p\geq 0 .
\end{equation}
Since $\partial_p \boldsymbol{v}(t,p)=-\boldsymbol{v}(t,p)$ for
$p>0$, the dissipative term can be rewritten as
\[
    H_1\boldsymbol{v}(t,p)
    =
    -H_1\partial_p\boldsymbol{v}(t,p).
\]
Therefore, the original system is transformed into the following
linear convection-type system in the enlarged $(t,p)$ space:
\begin{equation}\label{equ:u2v}
\begin{cases}
\displaystyle
    \partial_t \boldsymbol{v}(t,p)
    =
    -H_1\partial_p\boldsymbol{v}(t,p)
    +\mathrm{i}H_2\boldsymbol{v}(t,p),\\[1mm]
    \boldsymbol{v}(0,p)
    =
    \mathrm{e}^{-|p|}\boldsymbol{u}_0 .
\end{cases}
\end{equation}
For numerical purposes, one often uses a modified extension on the
negative half-line:
\begin{equation}\label{equ:u2valpha}
\begin{cases}
\displaystyle
    \partial_t \boldsymbol{v}(t,p)
    =
    -H_1\partial_p\boldsymbol{v}(t,p)
    +\mathrm{i}H_2\boldsymbol{v}(t,p),\\[1mm]
    \boldsymbol{v}(0,p)
    =
    \mathrm{e}^{-\alpha(p)|p|}\boldsymbol{u}_0 .
\end{cases}
\end{equation}
To preserve the exact relation
$\boldsymbol{v}(t,p)=\mathrm{e}^{-p}\boldsymbol{u}(t)$ in the recovery
region, one must take $\alpha(p)=1$ for $p>0$. On the negative side,
$\alpha(p)$ may be chosen sufficiently large so that the extended initial
data decays rapidly. 
The above extension gives a
first-order accurate approximation with respect to the mesh size
$\Delta p$, while higher-order extensions are available in
\cite{jinSchrodingerizationBasedQuantumAlgorithms2025,
jinSchrodingerizationBasedComputationally2025}.

If $H_1$ has both positive and negative eigenvalues, define
\begin{align}
    \lambda_{\max}^{+}(H_1)
    &=
    \max\left\{
    \sup_{0<t<T}
    \{|\lambda|:\lambda\in\sigma(H_1(t)),\lambda>0\},
    0
    \right\},\\
    \lambda_{\max}^{-}(H_1)
    &=
    \max\left\{
    \sup_{0<t<T}
    \{|\lambda|:\lambda\in\sigma(H_1(t)),\lambda<0\},
    0
    \right\}.
\end{align}
For a prescribed accuracy $\epsilon$, one may choose $L$ and $R$ such
that
\begin{equation}
    \mathrm{e}^{-R+\lambda_{\max}^{+}(H_1)T}\lesssim\epsilon,
    \qquad
    \mathrm{e}^{L+\lambda_{\max}^{-}(H_1)T}\lesssim\epsilon .
\end{equation}
These conditions ensure that the truncation error caused by the finite
auxiliary domain remains controlled
\cite{jinSchrodingerizationBasedQuantumAlgorithms2025}.

The original solution can be recovered from the transformed variable
through the identity
\[
    \boldsymbol{u}(t)=\mathrm{e}^{p}\boldsymbol{v}(t,p),
    \qquad p>0 .
\]
Thus, a simple recovery strategy is to choose a point
$p^\diamond>0$ and set
\begin{equation}\label{equ:restore1point}
    \boldsymbol{u}(t)
    =
    \mathrm{e}^{p^\diamond}\boldsymbol{v}(t,p^\diamond).
\end{equation}
When right-moving waves are present, the recovery point should be chosen
outside the affected transport region. In particular, one may take
$p^\diamond\geq \lambda_{\max}^{+}(H_1)T$, as discussed in
\cite{jinSchrodingerizationBasedQuantumAlgorithms2025}.

\begin{remark}\label{rem:recovery-point}
The condition $p^\diamond\ge \lambda^+_{\max}(H_1)T$ is a sufficient
condition used in the theoretical recovery estimate. It is generally
conservative, since $\lambda^+_{\max}(H_1)$ is a worst-case propagation speed
over all spectral components. In numerical implementations, taking
$p^\diamond$ too large may lead to an ill-conditioned recovery, because
$v(t,p)$ decays exponentially in the recovery region and the multiplication by
$e^{p^\diamond}$ can amplify discretization, truncation, and round-off errors.
Therefore, in the numerical experiments below, the recovery point is chosen in
a stable recovery region where the reconstructed solution is insensitive to
small changes of $p^\diamond$.
\end{remark}

It remains to discretize the auxiliary variable. Let $[L,R]$ be divided
uniformly with mesh size $\Delta p=\frac{R-L}{N_p}, N_p=2N,$
and grid points $L=p_0<p_1<\cdots<p_{N_p}=R.$
Collecting the values of $\boldsymbol{v}(t,p)$ at these grid points gives
a vector
\[
    \boldsymbol{w}(t)
    =
    \big[
    \boldsymbol{w}_1(t);
    \boldsymbol{w}_2(t);
    \cdots;
    \boldsymbol{w}_n(t)
    \big],
\]
where
\[
    \boldsymbol{w}_i(t)
    =
    \sum_k v_i(t,p_k)\ket{k},
\]
and $v_i$ denotes the $i$-th component of $\boldsymbol{v}$.

Applying a Fourier spectral approximation in the $p$ direction yields
\begin{equation}\label{equ:v2w}
    \frac{\mathrm{d}}{\mathrm{d}t}\boldsymbol{w}(t)
    =
    -\mathrm{i}(H_1\otimes P_\mu)\boldsymbol{w}(t)
    +
    \mathrm{i}(H_2\otimes I)\boldsymbol{w}(t),
\end{equation}
where $P_\mu$ is the discrete momentum operator associated with
$-\mathrm{i}\partial_p$. More precisely,
\[
    P_\mu=\Phi D_\mu \Phi^{-1},
    \qquad
    D_\mu=\operatorname{diag}(\mu_{-N},\ldots,\mu_{N-1}),
\]
with Fourier modes $\mu_l=\frac{2\pi l}{R-L},$
and discrete Fourier matrix
\[
    \Phi=(\phi_{jl})_{N_p\times N_p},
    \qquad
    \phi_{jl}=\phi_l(p_j),
    \qquad
    \phi_l(p)=\mathrm{e}^{\mathrm{i}\mu_l(p-L)}/N_p^{1/2} .
\]

Finally, by introducing $\tilde{\boldsymbol{w}}=(I\otimes\Phi^{-1})\boldsymbol{w},$
one obtains
\begin{equation}\label{equ:generalSchr}
    \frac{\mathrm{d}}{\mathrm{d}t}\tilde{\boldsymbol{w}}(t)
    =
    -\mathrm{i}
    \left(
    H_1\otimes D_\mu
    -
    H_2\otimes I
    \right)
    \tilde{\boldsymbol{w}}(t).
\end{equation}
This is a Schr\"odinger-type equation with Hamiltonian
$H=H_1\otimes D_\mu-H_2\otimes I$.
Thus, the original non-unitary dynamics has been embedded into a
unitary evolution in an enlarged Hilbert space. Moreover, if $H_1$ and
$H_2$ are sparse, then the resulting Hamiltonian inherits the sparsity
structure of the original system. In particular,
\[
    s(H)=\mathcal{O}(s(A)),
    \qquad
    \|H\|_{\max}
    \leq
    \frac{\|H_1\|_{\max}}{\Delta p}
    +
    \|H_2\|_{\max},
\]
where $s(H)$ denotes the maximum number of nonzero entries in each row
and column of $H$, and $\|\cdot\|_{\max}$ denotes the entrywise maximum
norm. Therefore, after Schr\"odingerization, the evolution can be treated
by standard Hamiltonian simulation techniques.

\subsection{Sparse-access input model and complexity analysis}

After Schr\"odingerization and the Fourier spectral discretization in the
auxiliary variable, we obtain the Hamiltonian system
\[
    \frac{d}{dt}\tilde{\boldsymbol{w}}(t)
    =-iH \tilde{\boldsymbol{w}}(t),
    \qquad
    H = H_1\otimes D_\mu-H_2\otimes I .
\]
We refer to the above matrix $H$ as the Hamiltonian of the
Schr\"odingerized system. It should be distinguished from the classical Hamiltonian of the Liouville equation $H_{\rm cla}$ in \eqref{equ:LiouvilleHamilClassical}.

In the complexity estimates below, we use the standard sparse-access oracle
model for sparse Hamiltonian simulation. This model is suitable for the
present problem because the Hamiltonian-preserving scheme leads to a
sparse semi-discrete matrix, and the Hamiltonian
$H$ inherits this sparsity. We use an oracle model of the
same form as that used in sparse Hamiltonian simulation by Low and
Chuang~\cite{lowOptimalHamiltonianSimulation2017}.

\begin{definition}[Sparse-access oracles]
Let $B\in \mathbb C^{N\times N}$ be an $s$-sparse matrix, namely, each row of
$B$ contains at most $s$ nonzero entries. For the Hermitian Hamiltonians
considered in sparse Hamiltonian simulation, row sparsity and column sparsity
coincide. We assume that $B$ is specified by
two oracles.

The value oracle $O_B$ returns the matrix entries:
\[
    O_B |j,k\rangle |0\rangle
    =
    |j,k\rangle |B_{jk}\rangle_m,
    \qquad j,k\in [N],
\]
where $|B_{jk}\rangle_m$ denotes an $m$-qubit register storing the matrix
element $B_{jk}$ to $m$ bits of precision.

The location oracle $O_F$ returns the positions of the nonzero entries:
\[
    O_F |j,\ell\rangle |0\rangle
    =
    |j,\ell\rangle |f_B(j,\ell)\rangle,
    \qquad j\in [N],\quad \ell\in [s].
\]
Here $f_B(j,\ell)$ denotes the column index of the $\ell$-th nonzero entry in
the $j$-th row of $B$. If the $j$-th row has fewer than $\ell+1$ nonzero
entries, $f_B(j,\ell)$ is assigned a fixed dummy value and the corresponding
matrix entry is treated as zero.
\end{definition}

\begin{remark}
This is a compact notation for the standard sparse-access input model. The
oracles can equivalently be written in reversible form by keeping the input
registers unchanged and adding the output to an auxiliary register.
\end{remark}

The computational cost of simulating the Schr\"odingerized system can therefore be
estimated by applying Theorem~3 of
\cite{lowOptimalHamiltonianSimulation2017}, which gives the query complexity
for sparse Hamiltonian simulation.

\begin{lemma}\label{lem:complexity}
    A $d$-sparse Hamiltonian $H$ on $n$ qubits with matrix elements specified to $m$ bits of precision can be simulated for time interval $[0,t]$, error $\epsilon$, and success probability at least $1 - 2\epsilon$  with 
    \[\mathcal{O}[td\|H\|_{\max}+\log{(1/\epsilon)}/\log{\log{(1/\epsilon)}}]\]
    queries to the sparse-access oracle model that provides a description of $H$ and a factor 
    \[\mathcal{O}[n+m\mathrm{polylog}(m)]\]
    additional quantum gates.
\end{lemma}

\section{Quantum simulation of the Liouville equation with interface condition \eqref{equ:ddimLiouvilleClassicalinterfacecondition} in one space dimension}\label{sec:interface1d}
We consider the numerical solution of the Liouville equation in one physical space dimension
\begin{equation}\label{equ:1dimLiouvilleClassical}
    f_{t}+\xi f_{x}-V_{x}f_{\xi}=0,
\end{equation}
with a discontinuous potential $V(x)$ and constant Hamiltonian condition 
\begin{equation}\label{equ:1dimLiouvilleClassicalPreserve}
    \frac{1}{2}(\xi^+)^2+V^+=\frac{1}{2}(\xi^-)^2+V^-.
\end{equation}
Then the interface condition \eqref{equ:ddimLiouvilleClassicalinterfacecondition} becomes 
\begin{equation}\label{equ:1dimLiouvilleClassicalinterfacecondition}
    f(t,x^+,\xi^+)=f(t,x^-,\xi^-).
\end{equation}

\subsection{The original algorithm}
We employ a uniform mesh with grid points at $x_{i+\frac{1}{2}}$, $i=0,\ldots,N_x$, in the $x$-direction and $\xi_{j+\frac{1}{2}}$, $j=0,\ldots,N_\xi$ in the $\xi$-direction. The cells are centered at $(x_i,\xi_j)$, $i=1,\ldots,N_x$, $j=1,\ldots,N_\xi$ with $x_i=\frac{1}{2}(x_{i+\frac{1}{2}}+x_{i-\frac{1}{2}})$ and $\xi_j=\frac{1}{2}(\xi_{j+\frac{1}{2}}+\xi_{j-\frac{1}{2}})$. The mesh size is denoted by $\Delta x= x_{i+\frac{1}{2}}-x_{i-\frac{1}{2}}$, $\Delta \xi=\xi_{j+\frac{1}{2}}-\xi_{j-\frac{1}{2}}$. We define the cell average of $f$ as 
\begin{equation*}
    f_{ij}=\frac{1}{\Delta x\Delta \xi}\int_{x_{i-\frac{1}{2}}}^{x_{i+\frac{1}{2}}}\int_{\xi_{j-\frac{1}{2}}}^{\xi_{j+\frac{1}{2}}}f(x,\xi,t)d\xi dx.
\end{equation*}

Assume that the discontinuous points of potential $V(x)$ are located at the grid points. Let the left and right limits of $V(x)$ at point $x_{i+1/2}$ be $V^{-}_{i+1/2}$ and $V^{+}_{i+1/2}$, respectively. Note that if $V$ is continuous at $x_{i+1/2}$, then $V^{+}_{i+1/2}=V^{-}_{i+1/2}$. We approximate $V$ by a piecewise linear function
\begin{equation*}
    V(x)\approx V_{i-1/2}^++\frac{V_{i+1/2}^--V_{i-1/2}^+}{\Delta x}(x-x_{i-1/2}).
\end{equation*}
 
The semi-discrete scheme (with continuous time) reads
\begin{equation}\label{equ:1dimLiouvilleClassicalSemidis}
    (f_{ij})_{t}+\frac{\xi_{j}}{\Delta x}\left(f_{i+\frac{1}{2},j}^{-}-f_{i-\frac{1}{2},j}^{+}\right)-\frac{V_{i+\frac{1}{2}}^{-}-V_{i-\frac{1}{2}}^{+}}{\Delta x\Delta\xi}\left(f_{i,j+\frac{1}{2}}-f_{i,j-\frac{1}{2}}\right)=0,
\end{equation}
where the numerical fluxes $f_{i,j+\frac{1}{2}}$ are defined using the upwind discretization.  Since the characteristics of the Liouville equation may be different on the two sides of a potential discontinuity, the corresponding numerical fluxes should also be different. The essential feature of the Hamiltonian-preserving schemes is to use \eqref{equ:1dimLiouvilleClassicalinterfacecondition} to define the split numerical fluxes $f^-_{i+\frac{1}{2},j},f^+_{i-\frac{1}{2},j}$ at each cell interface.

Assume $V$ is discontinuous at $x_{i+\frac{1}{2}}$. Consider the case $\xi_j > 0$. Using upwind scheme, $f_{i+\frac12,j}^- = f_{ij}$. However, $f_{i+1/2,j}^+=f(x_{i+1/2}^+,\xi_j^+)=f(x_{i+1/2}^-,\xi_j^-)$ while $\xi_j^-$ is obtained from $\xi_j^+ = \xi_j$ from \eqref{equ:1dimLiouvilleClassicalPreserve}. Since $\xi_j^-$ may not be a grid point, one has to define it approximately. One can first locate the two cell centers that bound $\xi_j^-$, then use a linear interpolation to evaluate the needed numerical flux. The case $\xi_j^- < 0$ is treated similarly. The detailed algorithm to generate the numerical flux is given in Algorithm \ref{alg:fluxClassical1}.

\begin{algorithm}[!htb]
    \footnotesize
  \label{alg:fluxClassical1}
  \caption{Computation of the numerical flux in \eqref{equ:1dimLiouvilleClassicalSemidis}}
  \SetAlgoLined
  \KwIn{$\xi_j$, $V^-_{i+\frac{1}{2}}$, $V^+_{i+\frac{1}{2}}$, $\{f_{ij},j=1,\ldots,N_\xi\}$ and $\{f_{i+1,j},j=1,\ldots,N_\xi\}$  }
  \KwOut{$f^-_{i+\frac{1}{2},j}$ and $f^+_{i+\frac{1}{2},j}$}
  \If{$\xi_j>0$}{
        $f^-_{i+\frac{1}{2},j}=f_{ij}$\;
        \uIf{$V_{i+\frac{1}{2}}^->V_{i+\frac{1}{2}}^+$}{
            \eIf{$\xi_j>\sqrt{2\left(V_{i+\frac{1}{2}}^--V_{i+\frac{1}{2}}^+\right)}$}
            {$\xi^-=\sqrt{\xi_j^2+2\left(V_{i+\frac{1}{2}}^+-V_{i+\frac{1}{2}}^-\right)}$\;
            \If{$\xi_k\le\xi^-<\xi_{k+1}$ for some $k$}
            {$f_{i+\frac{1}{2},j}^+=\frac{\xi_{k+1}-\xi^-}{\Delta\xi}f_{ik}+\frac{\xi^--\xi_k}{\Delta\xi}f_{i,k+1}$}}
            {$f_{i+\frac{1}{2},j}^+=f_{i+1,j'}$ where $\xi_{j'}=-\xi_j$}
    
        }
        \uElseIf{$V_{i+\frac{1}{2}}^-<V_{i+\frac{1}{2}}^+$}{
            $\xi^-=\sqrt{\xi_j^2+2\left(V_{i+\frac{1}{2}}^+-V_{i+\frac{1}{2}}^-\right)}$\;
            \If{$\xi_k\le\xi^-<\xi_{k+1}$ for some $k$}
            {$f_{i+\frac{1}{2},j}^+=\frac{\xi_{k+1}-\xi^-}{\Delta\xi}f_{ik}+\frac{\xi^--\xi_k}{\Delta\xi}f_{i,k+1}$}

        }
        \ElseIf{$V_{i+\frac{1}{2}}^-=V_{i+\frac{1}{2}}^+$}{
            $f_{i+\frac{1}{2},j}^+=f_{i+\frac{1}{2},j}^-$
        }
      }
  \If{$\xi_j<0$}{
        $f^+_{i+\frac{1}{2},j}=f_{i+1,j}$\;
        \uIf{$V_{i+\frac{1}{2}}^-<V_{i+\frac{1}{2}}^+$}{
            \eIf{$|\xi_j|>\sqrt{2\left(V_{i+\frac{1}{2}}^+-V_{i+\frac{1}{2}}^-\right)}$}
            {$\xi^+=-\sqrt{\xi_j^2+2\left(V_{i+\frac{1}{2}}^--V_{i+\frac{1}{2}}^+\right)}$\;
            \If{$\xi_k\le\xi^+<\xi_{k+1}$ for some $k$}
            {$f_{i+\frac{1}{2},j}^-=\frac{\xi_{k+1}-\xi^+}{\Delta\xi}f_{i+1,k}+\frac{\xi^+-\xi_k}{\Delta\xi}f_{i+1,k+1}$}}
            {$f_{i+\frac{1}{2},j}^-=f_{i+1,j'}$ where $\xi_{j'}=-\xi_j$}
    
        }
        \uElseIf{$V_{i+\frac{1}{2}}^->V_{i+\frac{1}{2}}^+$}{
            $\xi^+=-\sqrt{\xi_j^2+2\left(V_{i+\frac{1}{2}}^--V_{i+\frac{1}{2}}^+\right)}$\;
            \If{$\xi_k\le\xi^+<\xi_{k+1}$ for some $k$}
            {$f_{i+\frac{1}{2},j}^-=\frac{\xi_{k+1}-\xi^+}{\Delta\xi}f_{i+1,k}+\frac{\xi^+-\xi_k}{\Delta\xi}f_{i+1,k+1}$}

        }
        \ElseIf{$V_{i+\frac{1}{2}}^-=V_{i+\frac{1}{2}}^+$}{
            $f_{i+\frac{1}{2},j}^-=f_{i+\frac{1}{2},j}^+$
        }
      }
\end{algorithm}

\begin{remark}\label{rmk:convergencerate}
    Although  Algorithm \ref{alg:fluxClassical1} for evaluating numerical flux is of first order, the convergence rate of the numerical solution in $l^1$ norm may be less than first order. As it is well-known\cite{kuznetsov1977stable,tangSharpnessKuznetsovsOsqrtDelta1995}, when using a usual finite difference method for solving the discontinuous solution of a linear hyperbolic equation, the convergence rate is at most $1/2$ order. However, when the only discontinuity in the solutions is at the interface, the convergence rate of Hamiltonian-preserving scheme can still be first order \cite{jinHamiltonianPreservingSchemeLiouville2006}.
\end{remark}

\subsection{The ODE form}\label{sec:interface1ODEform}
In the following, we convert the semi-discrete scheme \eqref{equ:1dimLiouvilleClassicalSemidis} and Algorithm \ref{alg:fluxClassical1} into the ODE form \eqref{equ:ODElinear}. Assume that $N_x$ and $N_\xi$ are even numbers and the truncated interval in $x$ and $\xi$-direction are symmetric about $0$. (For asymmetrical case, the matrix $A$ in ODE \eqref{equ:ODElinear} will be slightly modified.) Then $\xi_j<0$ for $j=1,\ldots,\frac{N_\xi}{2}$ and $\xi_j>0$ for $j=\frac{N_\xi}{2}+1,\ldots,N_\xi$. We also assume that interfaces don't appear at $x_\frac{1}{2}$ and $x_{N_x+\frac12}$. Denote 
\[
\boldsymbol{f}(t)=[f_{11},\ldots,f_{1,N_\xi},f_{21},\ldots,f_{2,N_\xi},\ldots,f_{N_x,1},\ldots,f_{N_x,N_\xi}]^T=\sum_{i=1}^{N_x}\sum_{j=1}^{N_\xi}f_{ij}(t)\ket{i-1}\ket{j-1}.
\]
This definition specifies not only the order of the equations but also the order of the variables. We can also arrange $\boldsymbol{f}$ in different orders, then we only need to change the corresponding columns and rows of $A$ to get the new matrix.

If $\xi_j>0$, for the term $\frac{\xi_{j}}{\Delta x}\left(f_{i+\frac{1}{2},j}^{-}-f_{i-\frac{1}{2},j}^{+}\right)$ in \eqref{equ:1dimLiouvilleClassicalSemidis}, $f_{i+\frac{1}{2},j}^{-}=f_{ij}$ and $f^+_{i-\frac12,j}$ depends on the sign of $\xi_j^2+2\left(V_{i-\frac12}^+-V_{i-\frac12}^-\right)$. We define two step functions 
\begin{equation}\label{equ:stepfun}
    g^T(z) = \chi_{(0,+\infty)}(z),\quad g^R(z) = \chi_{(-\infty,0]}(z),
\end{equation}
where $\chi_{I}$ is the indicator function of set $I$
\begin{equation}
    \chi_{I}\left(z\right)=\left\{\begin{array}{ll}1,&\mathrm{if~}z\in I,\\0,&\mathrm{if~}z\not\in I.\end{array}\right.
\end{equation}
The above two step functions control the reflection and transmission at the interface.
Define
\begin{gather}
    a^T_{i+\frac{1}{2},j}=g^T\left(\xi_{j}^{2}+2\left(V_{i+\frac{1}{2}}^{+}-V_{i+\frac{1}{2}}^{-}\right)\right),\\
    a^R_{i+\frac{1}{2},j}=g^R\left(\xi_{j}^{2}+2\left(V_{i+\frac{1}{2}}^{+}-V_{i+\frac{1}{2}}^{-}\right)\right).
\end{gather}
Then $a^T_{i+\frac{1}{2},j}+a^R_{i+\frac{1}{2},j}=1$ and either $a^T_{i+\frac{1}{2},j}$ or $a^R_{i+\frac{1}{2},j}$ is $1$.

The search of which interval $\xi^-$ lies in involves a nonlinear operation. We need to use nonlinear functions to encode this nonlinearity into the matrix $A$. Define the hat function 
\begin{equation}\label{equ:hat}
    h(z)=\max\left(1-\frac{|z|}{\Delta\xi},0\right).
\end{equation}
The hat function can be regarded as a ``distance" function because it gives a non-negative value according to the distance from zero. It has a compact support of size $2\Delta\xi$ on $\mathbb{R}$.
Then 
\begin{eqnarray*}
f^+_{i-\frac12,j} &&=a^T_{i-\frac{1}{2},j}\sum_{k=1}^{N_\xi}h\left(\xi^--\xi_k\right)f_{i-1,k}+a^R_{i-\frac{1}{2},j}f_{i,j'}\\
&&=a^T_{i-\frac{1}{2},j}\sum_{k=1}^{N_\xi}h\left(\sqrt{\xi_{j}^{2}+2\left(V_{i-\frac{1}{2}}^{+}-V_{i-\frac{1}{2}}^{-}\right)}-\xi_k\right)f_{i-1,k}+a^R_{i-\frac{1}{2},j}f_{i,j'}
\end{eqnarray*}
and 
\begin{equation*}
    \frac{\xi_{j}}{\Delta x}\left(f_{i+\frac{1}{2},j}^{-}-f_{i-\frac{1}{2},j}^{+}\right)=\frac{\xi_{j}}{\Delta x}f_{ij}-\frac{\xi_{j}}{\Delta x}a^T_{i-\frac{1}{2},j}\sum_{k=1}^{N_\xi}h\left(\sqrt{\xi_{j}^{2}+2\left(V_{i-\frac{1}{2}}^{+}-V_{i-\frac{1}{2}}^{-}\right)}-\xi_k\right)f_{i-1,k}-\frac{\xi_{j}}{\Delta x}a^R_{i-\frac{1}{2},j}f_{i,j'}.
\end{equation*}

The use of step functions and compactly supported hat functions follows the same general matrix-encoding philosophy as in \cite{jinQuantumSimulationLiouville2026}. However, the structure encoded here is substantially different. In the geometrical optics problem, the interface contribution contains partial transmission and reflection coefficients, and the transmitted velocity is obtained from a Snell law's rescaling. In the present classical mechanics problem, the interface decision is an energy threshold branch selection: the particle is either transmitted with a velocity determined by the nonlinear square root map imposed by conservation of \(|v|^2/2+V(x)\), or reflected if the kinetic energy is below the potential jump. Therefore the nonzero pattern of the interface matrices, especially the column sparsity, requires a separate analysis.

If $\xi_j<0$, for the term $\frac{\xi_{j}}{\Delta x}\left(f_{i+\frac{1}{2},j}^{-}-f_{i-\frac{1}{2},j}^{+}\right)$ in \eqref{equ:1dimLiouvilleClassicalSemidis}, similarly, define
\begin{gather}
    a^T_{i+\frac{1}{2},j}=g^T\left(\xi_{j}^{2}+2\left(V_{i+\frac{1}{2}}^{-}-V_{i+\frac{1}{2}}^{+}\right)\right),\\
    a^R_{i+\frac{1}{2},j}=g^R\left(\xi_{j}^{2}+2\left(V_{i+\frac{1}{2}}^{-}-V_{i+\frac{1}{2}}^{+}\right)\right),
\end{gather}
then we can get 
\begin{equation*}
    \frac{\xi_{j}}{\Delta x}\left(f_{i+\frac{1}{2},j}^{-}-f_{i-\frac{1}{2},j}^{+}\right)=\frac{\xi_{j}}{\Delta x}a^T_{i+\frac{1}{2},j}\sum_{k=1}^{N_\xi}h\left(-\sqrt{\xi_{j}^{2}+2\left(V_{i+\frac{1}{2}}^{-}-V_{i+\frac{1}{2}}^{+}\right)}-\xi_k\right)f_{i+1,k}+\frac{\xi_{j}}{\Delta x}a^R_{i+\frac{1}{2},j}f_{i,j'}-\frac{\xi_{j}}{\Delta x}f_{ij}.
\end{equation*}

Define 
\begin{equation}\label{equ:beta}
    \beta_{i+\frac12,j,k}=\left\{\begin{array}{ll}
    h\left(\sqrt{\xi_{j}^{2}+2\left(V_{i+\frac{1}{2}}^{+}-V_{i+\frac{1}{2}}^{-}\right)}-\xi_k\right),     \\ 
    i=0,\ldots,N_x,\quad j=\frac{N_\xi}{2}+1,\ldots,N_\xi,\quad k=1,\ldots N_\xi, \\
    h\left(-\sqrt{\xi_{j}^{2}+2\left(V_{i+\frac{1}{2}}^{-}-V_{i+\frac{1}{2}}^{+}\right)}-\xi_k\right),     \\ 
    i=0,\ldots,N_x,\quad j=1,\ldots,\frac{N_\xi}{2},\quad k=1,\ldots N_\xi. 
    \end{array}\right.
\end{equation}
Here we use $i+\frac12$ to show that $\beta_{i+\frac12,j,k}$ is related to the coefficients $V^-_{i+\frac{1}{2}}$ and $V^+_{i+\frac{1}{2}}$ of the interface at $x_{i+\frac{1}{2}}$.
Then
\begin{equation*}
    \frac{\xi_{j}}{\Delta x}\left(f_{i+\frac{1}{2},j}^{-}-f_{i-\frac{1}{2},j}^{+}\right)=\left\{\begin{array}{ll}
    \frac{\xi_{j}}{\Delta x}f_{ij}-\frac{\xi_{j}}{\Delta x}a^T_{i-\frac{1}{2},j}\sum_{k=1}^{N_\xi}\beta_{i-\frac12,j,k}f_{i-1,k}-\frac{\xi_{j}}{\Delta x}a^R_{i-\frac{1}{2},j}f_{i,j'},     \\ \qquad\qquad i=1,\ldots,N_x,\quad j=\frac{N_\xi}{2}+1,\ldots,N_\xi, \\
    -\frac{\xi_{j}}{\Delta x}f_{ij}+\frac{\xi_{j}}{\Delta x}a^T_{i+\frac{1}{2},j}\sum_{k=1}^{N_\xi}\beta_{i+\frac12,j,k}f_{i+1,k}+\frac{\xi_{j}}{\Delta x}a^R_{i+\frac{1}{2},j}f_{i,j'},     \\ \qquad\qquad i=1,\ldots,N_x,\quad j=1,\ldots,\frac{N_\xi}{2}.
    \end{array}
    \right.
\end{equation*}
When $V^-_{i+\frac{1}{2}}=V^+_{i+\frac{1}{2}}$ the above scheme degenerates to the upwind scheme. 
Since we have assumed that interfaces don't appear at $x_\frac{1}{2}$ and $x_{N_x+\frac12}$, we can use the boundary condition to give the value of $f_{0j}(t)$ and $f_{N_x+1,j}(t)$.
When $\xi_j>0$ (resp. $\xi_j<0$), we impose inflow boundary condition at $(x_0,\xi_j)$ (resp. $(x_{N_x+1},\xi_j)$) and outflow boundary condition at $(x_{N_x+1},\xi_j)$ (resp. $(x_0,\xi_j)$).
Write the above equations in the form $A_1\boldsymbol{f}+b_1$ of a product of a matrix and a vector with an inhomogeneous term representing boundary condition, where the equations and variables are in order of $\boldsymbol{f}$, we can get a $N_xN_\xi\times N_xN_\xi$ block tridiagonal matrix
\begin{equation}\label{equ:blocktridiagmtx}
    A_1=\frac{1}{\Delta x}(I_{N_x}\otimes\Xi)
    \begin{pmatrix}
        C_1 & - B^u_1 &&&&&\\
        - B^l_1 & C_2 & - B^u_2 &&&&\\
        & - B^l_2 & C_3 & - B^u_3 &&&\\
        && \ddots & \ddots & \ddots &&\\
        &&& - B^l_{N_x-2} &C_{N_x-1}& - B^u_{N_x-1}\\
        &&&& - B^l_{N_x-1} &C_{N_x}\\
    \end{pmatrix}_{N_xN_\xi\times N_xN_\xi}
\end{equation}
and a $N_x N_\xi$-dimensional vector 
\begin{equation}\label{equ:inhomogeneoustermx}
    b_1=-\frac{1}{\Delta x}\sum_{j=\frac{N_\xi}{2}+1}^{N_\xi}\xi_jf_{0j}(t)\ket{0}\ket{j-1}+\frac{1}{\Delta x}\sum_{j=1}^{\frac{N_\xi}{2}}\xi_jf_{N_x+1, j}(t)\ket{N_x-1}\ket{j-1},
\end{equation}
where
\begin{equation}
    \Xi = \begin{pmatrix}
        |\xi_1| & & \\
        & \ddots & \\ 
        & & |\xi_{N_\xi}| \\
    \end{pmatrix},\quad
    C_i=I_{N_\xi}+\begin{pmatrix}
        & & & & & -a^R_{i+\frac12,1}\\
        & & & & \iddots &\\
         & &  & -a^R_{i+\frac12,\frac{N_\xi}{2}} & &\\
         &  & -a^R_{i-\frac12,\frac{N_\xi}{2}+1} & & &\\
        & \iddots & & & &  \\
        -a^R_{i-\frac12,N_\xi} &  & & & & 
    \end{pmatrix}_{N_\xi\times N_\xi}, 
\end{equation}
$B_i$'s $(j,k)$ element is $a^T_{i+\frac12,j}\beta_{i+\frac12,j,k}$ and $B^u_i$ (resp. $B^l_i$) refers to taking the upper (resp. lower) half of the matrix $B_i$ and making the rest $0$, i.e., 
\begin{equation*}
    B_i=\begin{pmatrix}
        a^T_{i+\frac12,1} & & & \\ 
        & a^T_{i+\frac12,2} & & \\ 
        & & \ddots & \\ 
        & & & a^T_{i+\frac12,N_\xi}
    \end{pmatrix} 
    \begin{pmatrix}
        \beta_{i+\frac12,1,1}&\beta_{i+\frac12,1,2}&\cdots&\beta_{i+\frac12,1,N_\xi}\\
        \beta_{i+\frac12,2,1}&\beta_{i+\frac12,2,2}&\cdots&\beta_{i+\frac12,2,N_\xi}\\
        \vdots&\vdots&\ddots&\vdots\\
        \beta_{i+\frac12,N_\xi,1}&\beta_{i+\frac12,N_\xi,2}&\cdots&\beta_{i+\frac12,N_\xi,N_\xi}\\
    \end{pmatrix}_{N_\xi\times N_\xi},
\end{equation*}
and
\begin{equation*}
    B_i^u=\begin{pmatrix}
        a^T_{i+\frac12,1} & & & \\ 
        & a^T_{i+\frac12,2} & & \\ 
        & & \ddots & \\ 
        & & & a^T_{i+\frac12,N_\xi}
    \end{pmatrix}
    \begin{pmatrix}
        \beta_{i+\frac12,1,1}&\beta_{i+\frac12,1,2}&\cdots&\beta_{i+\frac12,1,N_\xi}\\
        \vdots&\vdots&\ddots&\vdots\\
        \beta_{i+\frac12,\frac{N_\xi}{2},1}&\beta_{i+\frac12,\frac{N_\xi}{2},2}&\cdots&\beta_{i+\frac12,\frac{N_\xi}{2},N_\xi}\\
        0&0&\cdots&0\\
        \vdots&\vdots&\ddots&\vdots\\
        0&0&\cdots&0\\
    \end{pmatrix}_{N_\xi\times N_\xi},
\end{equation*}
\begin{equation*}
    B_i^l=\begin{pmatrix}
        a^T_{i+\frac12,1} & & & \\ 
        & a^T_{i+\frac12,2} & & \\ 
        & & \ddots & \\ 
        & & & a^T_{i+\frac12,N_\xi}
    \end{pmatrix}
    \begin{pmatrix}
        0&0&\cdots&0\\
        \vdots&\vdots&\ddots&\vdots\\
        0&0&\cdots&0\\
        \beta_{i+\frac12,\frac{N_\xi}{2}+1,1}&\beta_{i+\frac12,\frac{N_\xi}{2}+1,2}&\cdots&\beta_{i+\frac12,\frac{N_\xi}{2}+1,N_\xi}\\
        \vdots&\vdots&\ddots&\vdots\\
        \beta_{i+\frac12,N_\xi,1}&\beta_{i+\frac12,N_\xi,2}&\cdots&\beta_{i+\frac12,N_\xi,N_\xi}\\
    \end{pmatrix}_{N_\xi\times N_\xi}.
\end{equation*}
To facilitate subsequent complexity analysis, we will now discuss the sparsity of $B_i^u$ and $B_i^l$. 
Hermitian matrices have the same sparsity in row and column. For general matrix, we give the following definition.
\begin{definition}(Sparsity in row or column)
    For a general matrix $A\in \mathbb{C}^{m\times n}$, the sparsity in row (i.e., the number of nonzero elements in row) is denoted as $s_r(A)$ and the sparsity in column (i.e., the number of nonzero elements in column) is denoted as $s_c(A)$. For a Hermitian matrix $A\in \mathbb{C}^{m\times m}$, $s(A):=s_r(A)=s_c(A)$.
\end{definition}
Note that since the hat function \eqref{equ:hat} has a compact support of size $2\Delta\xi$, $s_r(B_i)=s_r(B_i^u)=s_r(B_i^l)\le2$, $s_c(B_i^u)$ and $s_c(B_i^l)$ is dependent on the velocity and potential along both sides of the interface.
\begin{itemize}
    \item If $V^-_{i+\frac{1}{2}}>V^+_{i+\frac{1}{2}}$, $s_c(B_i^l)\le2$ and $\beta_{i+\frac12,j,k}=0$ for $j=\frac{N_\xi}{2}+1,\ldots,N_\xi,k=1,\ldots,\frac{N_\xi}{2}$  ($\beta_{i+\frac12,\frac{N_\xi}{2}+1,\frac{N_\xi}{2}}$ may be nonzero when $\xi_{\frac{N_\xi}{2}+1}>\sqrt{2\left(V_{i+\frac{1}{2}}^--V_{i+\frac{1}{2}}^+\right)}$ ) or for all $j,k$ satisfying $\xi_j\le\sqrt{2\left(V_{i+\frac{1}{2}}^--V_{i+\frac{1}{2}}^+\right)}$, because if $\sqrt{\xi_j^2+2\left(V_{i+\frac12}^+-V_{i+\frac12}^-\right)}$ lies in $[\xi_{k},\xi_{k+1})$, then $\sqrt{\xi_{j+1}^2+2\left(V_{i+\frac12}^+-V_{i+\frac12}^-\right)}$ must lie in $[\xi_{k+l},\xi_{k+1+l})$ for $l\ge1$, 
    \[    s_c(B_i^u)\le2\left\lceil\sqrt{\frac54+\frac{2}{\Delta\xi}\sqrt{\left(\frac{\Delta\xi}{2}\right)^2+2\left(V_{i+\frac{1}{2}}^--V_{i+\frac{1}{2}}^+\right)}}-\frac12\right\rceil
    \]
    and $\beta_{i+\frac12,j,k}=0$ for $j=1,\ldots,\frac{N_\xi}{2},k=\frac{N_\xi}{2}+1,\ldots,N_\xi$, since $[\xi_{k},\xi_{k+1})$ may contain at most $\left\lfloor\sqrt{\frac54+\frac{2}{\Delta\xi}\sqrt{\left(\frac{\Delta\xi}{2}\right)^2+2\left(V_{i+\frac{1}{2}}^--V_{i+\frac{1}{2}}^+\right)}}-\frac12\right\rfloor$ rescaled intervals like 
    \[
    \left[-\sqrt{\xi_j^2+2\left(V_{i+\frac{1}{2}}^--V_{i+\frac{1}{2}}^+\right)},-\sqrt{\xi_{j+1}^2+2\left(V_{i+\frac{1}{2}}^--V_{i+\frac{1}{2}}^+\right)}\right).
    \]
    See Appendix \ref{app:B_i} for the detailed analysis.
    \item If $V^-_{i+\frac{1}{2}}<V^+_{i+\frac{1}{2}}$, similarly, $s_c(B_i^u)\le2$ and $\beta_{i+\frac12,j,k}=0$ for $j=1,\ldots,\frac{N_\xi}{2},k=\frac{N_\xi}{2}+1,\ldots,N_\xi$ ( $\beta_{i+\frac12,\frac{N_\xi}{2},\frac{N_\xi}{2}+1}$ may be nonzero when $|\xi_{\frac{N_\xi}{2}}|>\sqrt{2\left(V_{i+\frac{1}{2}}^+-V_{i+\frac{1}{2}}^-\right)}$ ) or for all $j,k$ satisfying $|\xi_j|\le\sqrt{2\left(V_{i+\frac{1}{2}}^+-V_{i+\frac{1}{2}}^-\right)}$, 
    \[
    s_c(B_i^l)\le 2\left\lceil\sqrt{\frac54+\frac{2}{\Delta\xi}\sqrt{\left(\frac{\Delta\xi}{2}\right)^2+2\left(V_{i+\frac{1}{2}}^+-V_{i+\frac{1}{2}}^-\right)}}-\frac12\right\rceil
    \]
    and $\beta_{i+\frac12,j,k}=0$ for $j=\frac{N_\xi}{2}+1,\ldots,N_\xi,k=1,\ldots,\frac{N_\xi}{2}$.
    \item If $V^-_{i+\frac{1}{2}}=V^+_{i+\frac{1}{2}}$, 
    \begin{equation*}
        B_i^u=\begin{pmatrix}
        a^T_{i+\frac12,1} & & & \\ 
        & a^T_{i+\frac12,2} & & \\ 
        & & \ddots & \\ 
        & & & a^T_{i+\frac12,N_\xi}
    \end{pmatrix}\begin{pmatrix}
            I_{\frac{N_\xi}{2}\times\frac{N_\xi}{2}}&O_{\frac{N_\xi}{2}\times\frac{N_\xi}{2}}\\
            O_{\frac{N_\xi}{2}\times\frac{N_\xi}{2}}&O_{\frac{N_\xi}{2}\times\frac{N_\xi}{2}}
        \end{pmatrix},
    \end{equation*}
    \begin{equation*}
        B_i^l=\begin{pmatrix}
        a^T_{i+\frac12,1} & & & \\ 
        & a^T_{i+\frac12,2} & & \\ 
        & & \ddots & \\ 
        & & & a^T_{i+\frac12,N_\xi}
    \end{pmatrix}\begin{pmatrix}
            O_{\frac{N_\xi}{2}\times\frac{N_\xi}{2}}&O_{\frac{N_\xi}{2}\times\frac{N_\xi}{2}}\\
            O_{\frac{N_\xi}{2}\times\frac{N_\xi}{2}}&I_{\frac{N_\xi}{2}\times\frac{N_\xi}{2}}
        \end{pmatrix},
    \end{equation*}
    which degenerates to the ordinary upwind scheme and $s_c(B_i^u)=s_c(B_i^l)=1$.
\end{itemize}
Denote the set of interfaces as $\mathcal{I}=\{x_{i+\frac{1}{2}}:V(x)\text{ is discontinuous on }x_{i+\frac{1}{2}}\}$. Then $s_r(A_1)\le3$ and $s_c(A_1)\le2\left\lceil\sqrt{\frac54+\frac{2}{\Delta\xi}\sqrt{\left(\frac{\Delta\xi}{2}\right)^2+2\max_{x\in\mathcal{I}}\left\{|V^-(x)-V^+(x)|\right\}}}-\frac12\right\rceil+2:=Q$. The details of the proof can be found in Appendix~\ref{app:A_1}.

For the term $-\frac{V_{i+\frac{1}{2}}^{-}-V_{i-\frac{1}{2}}^{+}}{\Delta x\Delta\xi}\left(f_{i,j+\frac{1}{2}}-f_{i,j-\frac{1}{2}}\right)$ in \eqref{equ:1dimLiouvilleClassicalSemidis}, define $d_{i}=-\frac{V_{i+\frac{1}{2}}^{-}-V_{i-\frac{1}{2}}^{+}}{\Delta x\Delta\xi}$. 
Since the numerical fluxes $f_{i,j+\frac{1}{2}},f_{i,j-\frac{1}{2}}$ are defined using the upwind discretization, we can get 
\begin{equation*}
    d_{i}\left(f_{i,j+\frac{1}{2}}-f_{i,j-\frac{1}{2}}\right)=-\frac{|d_{i}|+d_{i}}{2}f_{i,j-1}(t)+|d_{i}|f_{ij}(t)-\frac{|d_{i}|-d_{i}}{2}f_{i,j+1}(t).
\end{equation*}
When $d_{i}>0$ (resp. $d_{i}<0$) we impose inflow boundary condition at $(x_i,\xi_0)$ (resp. $(x_i,\xi_{N_\xi+1})$) and outflow boundary condition at $(x_i,\xi_{N_\xi+1})$ (resp. $(x_i,\xi_0)$). 
Write the above equations in the form $A_2\boldsymbol{f}+b_2$ of a product of a matrix and a vector with an inhomogeneous term representing boundary condition, where the equations and variables are in order of $\boldsymbol{f}$, we can get a $N_xN_\xi\times N_xN_\xi$ block diagonal matrix 
\begin{equation}\label{equ:blockdiagmtx}
    A_2=\begin{pmatrix}
        D_1 &&&\\
        & D_2&&\\
        &&\ddots&\\
        &&&D_{N_x}
    \end{pmatrix}_{N_xN_\xi\times N_xN_\xi}
\end{equation}
and a $N_xN_\xi$-dimensional vector
\begin{equation}\label{equ:inhomogeneoustermxi}
    b_2=-\sum_{i=1}^{N_x}\frac{|d_{i}|+d_{i}}{2}f_{i0}(t)\ket{i-1}\ket{0}-\sum_{i=1}^{N_x}\frac{|d_{i}|-d_{i}}{2}f_{i,N_\xi+1}(t)\ket{i-1}\ket{N_\xi-1},
\end{equation}
where 
\begin{align*}
    D_i=&\begin{pmatrix}
        |d_{i}|&-\frac{|d_{i}|-d_{i}}{2}&&&\\
        -\frac{|d_{i}|+d_{i}}{2}&|d_{i}|&-\frac{|d_{i}|-d_{i}}{2}&&\\
        &\ddots&\ddots&\ddots&\\
        && -\frac{|d_{i}|+d_{i}}{2}&|d_{i}|&-\frac{|d_{i}|-d_{i}}{2}\\
        &&& -\frac{|d_{i}|+d_{i}}{2}&|d_{i}|
    \end{pmatrix}_{N_\xi\times N_\xi}\\
    =&\frac12
        d_{i}
    \begin{pmatrix}
        0 & 1 &&&\\
        -1 &0&1&&\\
        & \ddots&\ddots&\ddots&\\
        &&-1&0&1\\
        &&&-1&0
    \end{pmatrix}+\frac12
        |d_{i}|
    \begin{pmatrix}
        2 & -1 &&&\\
        -1 &2&-1&&\\
        & \ddots&\ddots&\ddots&\\
        &&-1&2&-1\\
        &&&-1&2
    \end{pmatrix}.
\end{align*}

Combining \eqref{equ:1dimLiouvilleClassicalSemidis}, \eqref{equ:blocktridiagmtx}, \eqref{equ:blockdiagmtx}, \eqref{equ:inhomogeneoustermx} and \eqref{equ:inhomogeneoustermxi}, we can get the ODE $\boldsymbol{f}'(t)=A\boldsymbol{f}(t)+\boldsymbol{b}(t)$ where $A=-A_1-A_2$ and $\boldsymbol{b}(t)=-b_1-b_2$.

\begin{remark}
    Algorithm \ref{alg:fluxClassical1} for evaluating numerical fluxes is of first order. One can obtain a formally second order flux by incorporating the slope limiter, such as the van Leer or minmod slope limiter, into the above algorithm. In a classical computer, this can be achieved by replacing $f_{ik}$ with $f_{ik}+\frac{\Delta x}2s_{ik}$ and replacing $f_{i+1,k}$ with $f_{i+1,k}-\frac{\Delta x}{2}s_{i+1,k}$, where $s_{ik}=\frac{1}{\Delta x}(f_{i+1,k}-f_{ik})\psi(\theta_i)$ is the slope limiter in the $x$-direction, $\psi(\theta)$ is a nonlinear function (for the van Leer slope limiter, $\psi(\theta)=\frac{|\theta|+\theta}{1+|\theta|}$) and $\theta_i=\frac{f_{i,k}-f_{i-1,k}}{f_{i+1,k}-f_{i,k}}$.
    To compute $s_{i,k}$ we need to know the value of $f_{i,k}$ at every time step. This can be done easily in a classical computer. However, in a quantum computer, it is difficult to use the slope limiter which is nonlinear. If we use the Hamiltonian simulation to evolve the solution vector from $t=0$ to $t=T$, we can't get the solution at every time step. If we use the Hamiltonian simulation to evolve the solution vector from $t=(n-1)\Delta t$ to $t=n\Delta t$ for $n=1,\ldots,T/\Delta t$, then we need to measure the quantum state at every time step, which costs a lot.
\end{remark}

\subsection{Schr\"odingerization}\label{sec:interface1Schro}
If $\boldsymbol{b}(t)\neq0$, we then write the above ODE into a homogeneous form:
\begin{equation}\label{equ:interface1homogeneous}
    \frac{\d }{\d t}\begin{pmatrix}
        \boldsymbol{f}(t)\\
        \boldsymbol{r}(t)
    \end{pmatrix}=\begin{pmatrix}
        A & \text{diag}(\boldsymbol{b}(t))/\delta\\
        \boldsymbol{0} & \boldsymbol{0}
    \end{pmatrix}\begin{pmatrix}
        \boldsymbol{f}(t)\\
        \boldsymbol{r}(t)
    \end{pmatrix},\quad
    \begin{pmatrix}
        \boldsymbol{f}(0)\\
        \boldsymbol{r}(0)
    \end{pmatrix}=\begin{pmatrix}
        \boldsymbol{f}(0)\\
        \delta \boldsymbol{1}
    \end{pmatrix},
\end{equation}
where $\delta=\max_t\|\boldsymbol{b}(t)\|_\infty$, $\boldsymbol{0}$ is a $N_xN_\xi\times N_xN_\xi$ matrix with all entries  $0$ and $\boldsymbol{1}$ is a $N_xN_\xi$-dimensional vector with all components $1$. Then according to \eqref{equ:H1H2}
\begin{equation}
    H_{1}=\frac{1}{2}\begin{pmatrix}
    A+A^\dagger&\mathrm{diag}(\boldsymbol{b}(t))/\delta\\
    \mathrm{diag}(\boldsymbol{b}(t)^{\dagger})/\delta&\boldsymbol{0}\end{pmatrix},\quad H_{2}=\frac{1}{2i}\begin{pmatrix}
    A-A^\dagger&\mathrm{diag}(\boldsymbol{b}(t))/\delta\\
    -\mathrm{diag}(\boldsymbol{b}(t)^{\dagger})/\delta&\boldsymbol{0}\end{pmatrix}
\end{equation}
and \eqref{equ:interface1homogeneous} can be rewritten as 
\begin{equation}
    \frac{\d}{\d t}\boldsymbol{u}=H_1\boldsymbol{u}+iH_2\boldsymbol{u}.\quad\boldsymbol{u}(0)=\begin{pmatrix}
        \boldsymbol{f}(0)\\
        \delta \boldsymbol{1}
    \end{pmatrix}.
\end{equation}
Using the warped phase transformation and then taking the Fourier spectral method on $p$ in position basis, we can get \eqref{equ:v2w}. After change of variables $\tilde{\bm{w}}(t) = (I \otimes \Phi^{-1})\bm{w}(t)$ and simulating the Hamiltonian system \eqref{equ:generalSchr} from $t=0$ to $t=T$, we can get $\tilde{\boldsymbol{w}}(T)$. Then applying the inverse quantum Fourier transform (IQFT) $\bm{w}(T) = (I \otimes \Phi)\tilde{\bm{w}}(T)$ and using the one point recovery \eqref{equ:restore1point}, we can get the solution quantum state.
The procedure from initial value $\boldsymbol{f}(0)$ to solution $\boldsymbol{f}(T)$ can be summarized as follows:
\begin{multline}\label{equ:procedure}
    \ket{\boldsymbol{f}(0)}\xrightarrow{\text{homogeneous}} \ket{\boldsymbol{u }(0)}\xrightarrow{\text{warped phase}} \ket{\boldsymbol{w}(0)}\xrightarrow{\text{QFT on } p} \ket{\tilde{\boldsymbol{w}}(0)} \xrightarrow{\e^{-iHT}} \ket{\tilde{\boldsymbol{w}}(T)} \\ 
    \xrightarrow{\text{IQFT on } p} \ket{\boldsymbol{w}(T)} \xrightarrow{\text{recovery}} 
    \ket{\boldsymbol{u }(T)} \xrightarrow{\text{measurements}} \ket{\boldsymbol{f }(T)}
\end{multline}

The following theorem gives the complexity of the Schr\"odingerization based Hamiltonian-preserving scheme.

\begin{theorem}\label{thm:1dim_complexity1}
    In the 1-dimensional case, given the initial quantum state $\ket{\boldsymbol{u}(0)}$, assume that the $x$ grid number, $\xi$ grid number and extended $p$ grid number are $N_x=2^{n_x}$, $N_\xi=2^{n_\xi}$ and $N_p=2^{n_p}$ respectively, $R-L=\mathcal O(1)$, the inhomogeneous term $\boldsymbol{b}$ is independent of time and the errors of Hamiltonian-preserving scheme, spectral method in $p$ and Hamiltonian simulation are all $\epsilon$. With the Schr\"odingerization method, the state $\ket{\boldsymbol{u}(t)}$ can be simulated for time $T$, and success probability is at least $1-2\epsilon$ with
    \begin{itemize}
        \item $\mathcal{O}(QT\epsilon^{-2}+\frac{\log\epsilon^{-1}}{\log\log\epsilon^{-1}})$ queries and a factor $\mathcal{O}(3\log(\epsilon^{-1})+\log(\epsilon^{-1})\polylog(\log(\epsilon^{-1})))$ additional quantum gates, when the solution has discontinuities only at interfaces,
        \item $\mathcal{O}(QT\epsilon^{-3}+\frac{\log\epsilon^{-1}}{\log\log\epsilon^{-1}})$ queries and a factor $\mathcal{O}(5\log(\epsilon^{-1})+\log(\epsilon^{-1})\polylog(\log(\epsilon^{-1})))$ additional quantum gates, when the solution has discontinuities not only at interfaces,
    \end{itemize}
    where $Q:=2\left\lceil\sqrt{\frac54+\frac{2}{\Delta\xi}\sqrt{\left(\frac{\Delta\xi}{2}\right)^2+2\max_{x\in\mathcal{I}}\left\{|V^-(x)-V^+(x)|\right\}}}-\frac12\right\rceil+2\ge s_c(A_1)\ge 3$.
\end{theorem}
\begin{proof}
    We use Lemma \ref{lem:complexity} to give the proof. Matrix elements are specified to $\log(\epsilon^{-1})$ bits of precision.
    
    From previous discussion, 
    \begin{equation*}
        \begin{aligned}
            s(H)=s(H_1\otimes D_\mu-H_2\otimes I)&\le \max\{s(A+A^\dagger),s(A-A^\dagger)\}+s(\text{diag}(\boldsymbol{b}(t))/\delta)\\
            &\le \max\{s(A_1+A_1^\dagger),s(A_1-A_1^\dagger)\}+\max\{s(A_2+A_2^\dagger),s(A_2-A_2^\dagger)\}-1+1\\
            &\le (s_r(A_1)+s_c(A_1)-1)+(3)\\
            &\le Q+5.
        \end{aligned}
    \end{equation*}
    In the $p$-direction, the initial value is only continuous but not differentiable. According to spectral method, in order to reach precision $\epsilon$, the mesh size $\Delta p$ is $\mathcal{O}(\epsilon)$.
    The total number of points in $x$ and $\xi$-direction is related to where the discontinuities lies.
    \begin{itemize}
        \item when the solution has discontinuities only at interfaces, Hamiltonian-preserving scheme can give the first order precision. Thus $N_x=N_\xi=\mathcal{O}(\epsilon^{-1})$ and $\|H\|_{\max}\le\|H_{1}\otimes D_{\mu}\|_{\max}+\|H_{2}\|_{\max}\le\mathcal{O}(\left\|H_{1}\right\|_{\max}\mu_{N_p/2})=\mathcal{O}((N_x+N_\xi)/\Delta p)=\mathcal{O}(\epsilon^{-2})$. The query complexity is 
        \begin{equation*}
            \mathcal{O}(QT\epsilon^{-2}+\frac{\log\epsilon^{-1}}{\log\log\epsilon^{-1}}).
        \end{equation*} 
        Since the quantum Fourier transform in Schr\"odingerization can be implemented by using $\mathcal{O}(m\log m)$ gates \cite{namApproximateQuantumFourier2020}, the number of overall additional quantum gates is 
        \begin{equation*}
            \mathcal{O}(3\log(\epsilon^{-1})+\log(\epsilon^{-1})\polylog(\log(\epsilon^{-1}))+\log(\epsilon^{-1})\log(\log(\epsilon^{-1}))).
        \end{equation*}
        \item when the solution has discontinuities not only at interfaces, Hamiltonian-preserving scheme can give the half-order precision. Thus $N_x=N_\xi=\mathcal{O}(\epsilon^{-2})$ and $\|H\|_{\max}\le\mathcal{O}(\epsilon^{-3})$. The query complexity is 
        \begin{equation*} 
            \mathcal{O}(QT\epsilon^{-3}+\frac{\log\epsilon^{-1}}{\log\log\epsilon^{-1}}).
        \end{equation*} 
        and the number of overall additional quantum gates is 
        \begin{equation*}
            \mathcal{O}(5\log(\epsilon^{-1})+\log(\epsilon^{-1})\polylog(\log(\epsilon^{-1}))+\log(\epsilon^{-1})\log(\log(\epsilon^{-1}))).
        \end{equation*}
    \end{itemize}
\end{proof}
\begin{remark}
    For the smoother initial data $g(p)\in H^r((L,R))$, $N_p=\mathcal{O}(\epsilon^{-\frac{1}{r}})$.
\end{remark}
\begin{remark}
    In the classical implementation of Hamiltonian-preserving schemes, one can use any time discretization for the time derivative. Here we use the forward Euler time discretization as an example to analyze the classical complexity which contains the number of products and quotients. After discretization, the ODE $\boldsymbol{f}'(t)=A\boldsymbol{f}(t)+\boldsymbol{b}(t)$ becomes $\boldsymbol{f}^{n+1}=(\Delta tA+I)\boldsymbol{f}^n+\Delta t\boldsymbol{b}$ for $n=0,\ldots,{N_t}-1$. To satisfy the CFL condition $\Delta t\left[\frac{\max_j|\xi_j|}{\Delta x}+\frac{\max_i\left|\frac{V_{i+\frac{1}{2}}^--V_{i-\frac{1}{2}}^+}{\Delta x}\right|}{\Delta\xi}\right]\leq1$, a time step $\Delta t=\mathcal{O}(\Delta x,\Delta \xi)$ is needed \cite{jinHamiltonianPreservingSchemesLiouville2005}, so we take $N_t=\mathcal{O}(TN_x)$. Then the classical complexity is 
    \begin{equation}
        \mathcal{O}(s_r(\Delta tA+I)N_x N_\xi N_t)=\begin{cases}
            \mathcal{O}(5T\epsilon^{-3}) & \text{when the solution has discontinuities only at interfaces},\\
            \mathcal{O}(5T\epsilon^{-6})& \text{when the solution has discontinuities not only at interfaces}.
        \end{cases}
    \end{equation}
    There is a polynomial advantage in $\epsilon$ for quantum algorithms.
\end{remark}

\begin{remark}
Theorem \ref{thm:1dim_complexity1} gives the query complexity and the number of additional quantum gates in the sparse-access model. The gate complexity of the unitary oracle that provides a description of $H$ may depend on the matrix dimension, so the actual gate complexity may be larger than the query complexity.

The estimates above treat the auxiliary discretization size $N_p$ through the accuracy requirement in the $p$ variable. If one further makes the dependence of the truncated $p$-interval on the spectrum of $H_1$ explicit, then the query complexity remains unchanged, while only the logarithmic register-size contribution in the additional gate count is modified. Indeed, by the truncation condition of the Schr\"odingerization method, one may choose
\[
  R-L
  =
  \mathcal{O}\left(
  (\lambda^+_{\max}(H_1)+\lambda^-_{\max}(H_1))T
  + \log(\epsilon^{-1})
  \right).
\]
Together with $\Delta p=\mathcal{O}(\epsilon)$, this gives
\[
  N_p
  =
  \mathcal{O}\left(
  \frac{
  (\lambda^+_{\max}(H_1)+\lambda^-_{\max}(H_1))T
  + \log(\epsilon^{-1})
  }{\epsilon}
  \right).
\]
Under the scaling used in Theorem \ref{thm:1dim_complexity1},
\[
  \lambda^+_{\max}(H_1)+\lambda^-_{\max}(H_1)
  =
  \mathcal{O}(N_x+N_\xi).
\]
Hence, when the solution has discontinuities only at interfaces,
$N_x=N_\xi=\mathcal{O}(\epsilon^{-1})$, and
\[
  N_p
  =
  \mathcal{O}\left(
  T\epsilon^{-2}
  +
  \epsilon^{-1}\log(\epsilon^{-1})
  \right).
\]
For fixed final time $T>0$, this implies
\[
  n_p=\log N_p=2\log(\epsilon^{-1})+\mathcal{O}(1).
\]
Therefore the total register size becomes
\[
  n_x+n_\xi+n_p+\mathcal{O}(1)
  =
  4\log(\epsilon^{-1})+\mathcal{O}(1).
\]
Similarly, when the solution has discontinuities not only at interfaces,
$N_x=N_\xi=\mathcal{O}(\epsilon^{-2})$, and
\[
  N_p
  =
  \mathcal{O}\left(
  T\epsilon^{-3}
  +
  \epsilon^{-1}\log(\epsilon^{-1})
  \right),
\]
so that, for fixed $T>0$,
\[
  n_p=3\log(\epsilon^{-1})+\mathcal{O}(1),
\]
and the total register size becomes
\[
  n_x+n_\xi+n_p+\mathcal{O}(1)
  =
  7\log(\epsilon^{-1})+\mathcal{O}(1).
\]

Thus, if the dependence of the truncated $p$-interval on the mesh size is made explicit, the query complexities in Theorem \ref{thm:1dim_complexity1} are unchanged, while the logarithmic register-size terms in the additional gate counts may be written as
\[
  \mathcal{O}\left(
  4\log(\epsilon^{-1})
  + \log(\epsilon^{-1})\operatorname{polylog}(\log(\epsilon^{-1}))
  \right)
\]
in the first-order case, and
\[
  \mathcal{O}\left(
  7\log(\epsilon^{-1})
  + \log(\epsilon^{-1})\operatorname{polylog}(\log(\epsilon^{-1}))
  \right)
\]
in the half-order case, up to the cost of the sparse-access oracle.
The cost of the QFT and IQFT in the auxiliary $p$ variable is
$\mathcal{O}(n_p\log n_p)$, which is
$\mathcal{O}(\log(\epsilon^{-1})\log\log(\epsilon^{-1}))$ for fixed final time $T$.
This term is absorbed in the
$\log(\epsilon^{-1})\operatorname{polylog}(\log(\epsilon^{-1}))$ factor.
\end{remark}

\section{Quantum simulation of Liouville equation with interface condition \eqref{equ:ddimLiouvilleClassicalinterfacecondition} in two spatial dimensions }\label{sec:interface2d}
Consider the Liouville equation in two spatial dimensions:
\begin{equation}\label{equ:2dimLiouvilleClassical}
    f_t+\xi f_x+\eta f_y-V_xf_\xi-V_yf_\eta=0.
\end{equation}

\subsection{The original algorithm}
We employ a uniform mesh with grid points at $x_{i+\frac{1}{2}},y_{j+\frac{1}{2}},\xi_{k+\frac{1}{2}},\eta_{l+\frac{1}{2}},i=0,\ldots,N_x,j=0,\ldots,N_y,k=0,\ldots,N_\xi,l=0,\ldots,N_\eta$ in each direction. The cells are centered at $(x_{i},y_{j},\xi_{k},\eta_{l})$ with $x_{i}=\frac{1}{2}(x_{i+\frac{1}{2}}+x_{i-\frac{1}{2}}),y_{j}=\frac{1}{2}(y_{j+\frac{1}{2}}+y_{j-\frac{1}{2}}),\xi_{k}=\frac{1}{2}(\xi_{k+\frac{1}{2}}+\xi_{k-\frac{1}{2}}),\eta_{l}=\frac{1}{2}(\eta_{l+\frac{1}{2}}+\eta_{l-\frac{1}{2}})$. The mesh size is denoted by $\Delta x=x_{i+\frac{1}{2}}-x_{i-\frac{1}{2}},\Delta y=y_{j+\frac{1}{2}}-y_{j-\frac{1}{2}},\Delta\xi=\xi_{k+\frac{1}{2}}-\xi_{k-\frac{1}{2}},\Delta\eta=\eta_{l+\frac{1}{2}}-\eta_{l-\frac{1}{2}}.$ We define the cell average of $f$ as
\begin{equation*}
    f_{ijkl}=\frac{1}{\Delta x\Delta y\Delta\xi\Delta\eta}\int_{x_{i-\frac{1}{2}}}^{x_{i+\frac{1}{2}}}\int_{y_{j-\frac{1}{2}}}^{y_{j+\frac{1}{2}}}\int_{\xi_{k-\frac{1}{2}}}^{\xi_{k+\frac{1}{2}}}\int_{\eta_{l-\frac{1}{2}}}^{\eta_{l+\frac{1}{2}}}f(x,y,\xi,\eta,t)\mathrm{d}\eta\mathrm{d}\xi\mathrm{d}y\mathrm{d}x.
\end{equation*}
Similar to the 1D case, we approximate $V(x, y)$ by a piecewise bilinear function, and for convenience, we always provide two interface values of $V$ at each cell interface. When $V$ is smooth at a cell interface, the two interface values are identical.

The 2D Liouville equation \eqref{equ:2dimLiouvilleClassical} can be semi-discretized as
\begin{multline}\label{equ:2dimLiouvilleClassicalSemidis}
    \left(f_{ijkl}\right)_{t}+\frac{\xi_k}{\Delta x}\left(f_{i+\frac{1}{2},jkl}^--f_{i-\frac{1}{2},jkl}^+\right)+\frac{\eta_l}{\Delta y}\left(f_{i,j+\frac{1}{2},kl}^--f_{i,j-\frac{1}{2},kl}^+\right)\\
    -\frac{V_{i+\frac{1}{2},j}^{-}-V_{i-\frac{1}{2},j}^{+}}{\Delta x\Delta\xi}\left(f_{ij,k+\frac{1}{2},l}-f_{ij,k-\frac{1}{2},l}\right)-\frac{V_{i,j+\frac{1}{2}}^{-}-V_{i,j-\frac{1}{2}}^{+}}{\Delta y\Delta\eta}\left(f_{ijk,l+\frac{1}{2}}-f_{ijk,l-\frac{1}{2}}\right)=0,
\end{multline}
where the interface values $f_{ij,k+\frac{1}{2},l},f_{ijk,l+\frac{1}{2}}$ are provided by the upwind approximation, and the split interface values
$f_{i+\frac{1}{2},jkl}^{-},f_{i-\frac{1}{2},jkl}^{+},f_{i,j+\frac{1}{2},kl}^{-},f_{i,j-\frac{1}{2},kl}^{+}$ can be obtained using essentially the same algorithm described in Algorithm \ref{alg:fluxClassical1} for the 1D case when the discontinuity of $V(x,y)$ aligns with the grids.

\subsection{The ODE form}\label{sec:interface2dODEform}
In the following, we convert the semi-discrete scheme \eqref{equ:2dimLiouvilleClassicalSemidis} and Algorithm \ref{alg:fluxClassical1} into ODE form \eqref{equ:ODElinear}. Assume that $N_x$, $N_y$, $N_\xi$ and $N_\eta$ are even numbers and the truncated interval in $x$, $y$, $\xi$ and $\eta$-direction are symmetric about $0$. (For asymmetrical case, the matrix $A$ in ODE \eqref{equ:ODElinear} will be slightly modified.) We also assume that interfaces don't appear at $x_\frac{1}{2}$, $x_{N_x+\frac12}$, $y_\frac{1}{2}$ and $y_{N_y+\frac12}$. Denote $\boldsymbol{f}(t)=\sum_{i=1}^{N_x}\sum_{j=1}^{N_y}\sum_{k=1}^{N_\xi}\sum_{l=1}^{N_\eta}f_{ijkl}(t)\ket{i-1}\ket{j-1}\ket{k-1}\ket{l-1}$.  This definition specifies not only the order of the equations but also the order of the variables. We can also arrange $\boldsymbol{f}$ in different orders, then we only need to change the corresponding columns and rows of $A$ to get the new matrix.

If $\xi_k>0$, for the term $\frac{\xi_k}{\Delta x}\left(f_{i+\frac{1}{2},jkl}^--f_{i-\frac{1}{2},jkl}^+\right)$ in \eqref{equ:2dimLiouvilleClassicalSemidis}, $f_{i+\frac{1}{2},jkl}^{-}=f_{ijkl}$ and $f^+_{i-\frac12,jkl}$ depends on the sign of $\xi_k^2+2\left(V_{i-\frac12,j}^+-V_{i-\frac12,j}^-\right)$. We use the step functions \eqref{equ:stepfun} to define the reflection and transmission coefficients at the interface
\begin{gather}
    a^T_{i+\frac{1}{2},jk}=g^T\left(\xi_k^2+2\left(V_{i+\frac12,j}^+-V_{i+\frac12,j}^-\right)\right),\\
    a^R_{i+\frac{1}{2},jk}=g^R\left(\xi_k^2+2\left(V_{i+\frac12,j}^+-V_{i+\frac12,j}^-\right)\right).
\end{gather}
Then $a^T_{i+\frac{1}{2},jk}+a^R_{i+\frac{1}{2},jk}=1$ and either $a^T_{i+\frac{1}{2},jk}$ or $a^R_{i+\frac{1}{2},jk}$ is $1$.

The search of which interval $\xi^-$ lies in involves a nonlinear operation. We need to use hat function \eqref{equ:hat} to encode this nonlinearity into the matrix $A$. 
Then $f_{i-\frac{1}{2},jkl}^{+}=a^{T}_{i-\frac{1}{2},jk}\sum_{k'=1}^{N_\xi}h\left(\xi^--\xi_{k'}\right)f_{i-1,jk'}+a^{R}_{i-\frac{1}{2},jk}f_{i,j,k_{1},l}=
a^{T}_{i-\frac{1}{2},jk}\sum_{k'=1}^{N_\xi}h\left(\sqrt{\xi_k^2+2\left(V_{i-\frac12,j}^+-V_{i-\frac12,j}^-\right)}-\xi_{k'}\right)f_{i-1,jk'l}+a^{R}_{i-\frac{1}{2},jk}f_{i,j,k_{1},l}$ and 
\begin{multline*}
    \frac{\xi_k}{\Delta x}\left(f_{i+\frac{1}{2},jk}^{-}-f_{i-\frac{1}{2},jkl}^{+}\right)=
    \frac{\xi_k}{\Delta x}f_{ijkl}\\
    -\frac{\xi_k}{\Delta x}a^{T}_{i-\frac{1}{2},jk}\sum_{k'=1}^{N_\xi}h\left(\sqrt{\xi_k^2+2\left(V_{i-\frac12,j}^+-V_{i-\frac12,j}^-\right)}-\xi_{k'}\right)f_{i-1,jk'l}
    -\frac{\xi_k}{\Delta x}a^{R}_{i-\frac{1}{2},jk}f_{i,j,k_{1},l},
\end{multline*}
where $k_{1}=N_\xi-k+1$.

If $\xi_k<0$, for the term $\frac{\xi_k}{\Delta x}\left(f_{i+\frac{1}{2},jkl}^{-}-f_{i-\frac{1}{2},jkl}^{+}\right)$ in \eqref{equ:2dimLiouvilleClassicalSemidis}, similarly, define
\begin{gather}
    a^T_{i+\frac{1}{2},jk}=g^T\left(\xi_k^2+2\left(V_{i+\frac12,j}^--V_{i+\frac12,j}^+\right)\right),\\
    a^R_{i+\frac{1}{2},jk}=g^R\left(\xi_k^2+2\left(V_{i+\frac12,j}^--V_{i+\frac12,j}^+\right)\right),
\end{gather}
then we can get 
\begin{multline*}
    \frac{\xi_k}{\Delta x}\left(f_{i+\frac{1}{2},jkl}^{-}-f_{i-\frac{1}{2},jkl}^{+}\right)=
    \frac{\xi_k}{\Delta x}a^{T}_{i+\frac{1}{2},jk}\sum_{k'=1}^{N_\xi}h\left(-\sqrt{\xi_k^2+2\left(V_{i+\frac12,j}^--V_{i+\frac12,j}^+\right)}-\xi_{k'}\right)f_{i+1,jk'l}\\ 
    +\frac{\xi_k}{\Delta x}a^{R}_{i+\frac{1}{2},jk}f_{i,j,k_{1},l}
    -\frac{\xi_k}{\Delta x}f_{ijkl},
\end{multline*}
where $k_{1}=N_\xi-k+1$.

Define 
\begin{equation}\label{equ:betaxi}
    \beta_{i+\frac12,j,k,k'}^\xi=\left\{\begin{array}{ll}
    h\left(\sqrt{\xi_k^2+2\left(V_{i+\frac12,j}^+-V_{i+\frac12,j}^-\right)}-\xi_{k'}\right),     \\
    i=0,\ldots,N_x,\quad j=1,\ldots,N_y,\quad k=\frac{N_\xi}{2}+1,\ldots,N_\xi,\quad k'=1,\ldots,N_\xi, \\
    h\left(-\sqrt{\xi_k^2+2\left(V_{i+\frac12,j}^--V_{i+\frac12,j}^+\right)}-\xi_{k'}\right),     \\
    i=0,\ldots,N_x,\quad j=1,\ldots,N_y,\quad k=1,\ldots, \frac{N_\xi}{2},\quad k'=1,\ldots,N_\xi. 
    \end{array}\right.
\end{equation}
Here we use $i+\frac12$ to show that $\beta_{i+\frac12,j,k,k'}^\xi$ is related to the coefficients $V^-_{i+\frac{1}{2},j}$ and $V^+_{i+\frac{1}{2},j}$ of the interface at $x_{i+\frac{1}{2},j}$.
Then
\begin{multline*}
    \frac{\xi_k}{\Delta x}\left(f_{i+\frac{1}{2},jkl}^{-}-f_{i-\frac{1}{2},jkl}^{+}\right)=\\
    \left\{\begin{array}{l}
        \frac{\xi_k}{\Delta x}\left(f_{ijkl}-a_{i-\frac12,jk}^T\sum_{k'=1}^{N_\xi}\beta_{i-\frac12,j,k,k'}^\xi f_{i-1,jk'l}-a_{i-\frac12,jk}^R f_{i,j,k_1,l}\right),     \\ 
        \qquad i=1,\ldots,N_x,\quad j=1,\ldots,N_y,\quad k=\frac{N_\xi}{2}+1,\ldots,N_\xi,\quad l=1,\ldots,N_\eta, \\
        \frac{\xi_k}{\Delta x}\left(a^{T}_{i+\frac{1}{2},jk}\sum_{k'=1}^{N_\xi}\beta_{i+\frac12,j,k,k'}^\xi f_{i+1,jk'l}
        +a^{R}_{i+\frac{1}{2},jk}f_{i,j,k_{1},l}-f_{ijkl}\right),    \\ 
        \qquad i=1,\ldots,N_x,\quad j=1,\ldots,N_y,\quad k=1,\ldots,\frac{N_\xi}{2},\quad l=1,\ldots,N_\eta.
    \end{array}
    \right.
\end{multline*}
When $V^-_{i+\frac{1}{2},j}=V^+_{i+\frac{1}{2},j}$ the above scheme degenerates to the upwind scheme. 
Since we have assumed that interfaces don't appear at $x_\frac{1}{2}$ and $x_{N_x+\frac12}$, we can use the boundary condition to give the value of $f_{0jkl}(t)$ and $f_{N_x+1,jkl}(t)$.
When $\xi_k>0$ (resp. $\xi_k<0$), we impose inflow boundary condition at $(x_0,y_j,\xi_k,\eta_l)$ (resp. $(x_{N_x+1},y_j,\xi_k,\eta_l)$) and outflow boundary condition at $(x_{N_x+1},y_j,\xi_k,\eta_l)$ (resp. $(x_0,y_j,\xi_k,\eta_l)$).
Write the above equations in the form $A_1\boldsymbol{f}+b_1$ of a product of a matrix and a vector with an inhomogeneous term representing boundary condition, where the equations and variables are in order of $\boldsymbol{f}$, we can get a $N_xN_yN_\xi N_\eta\times N_xN_yN_\xi N_\eta$ matrix
\begin{multline}\label{equ:mtx1}
    A_1=\sum_{i=1}^{N_x}\sum_{j=1}^{N_y}\sum_{k=\frac{N_\xi}{2}+1}^{N_\xi}\sum_{l=1}^{N_\eta}\frac{\xi_k}{\Delta x}\ket{i-1}\ket{j-1}\ket{k-1}\ket{l-1}\bra{i-1}\bra{j-1}\bra{k-1}\bra{l-1}\\
    -\sum_{i=2}^{N_x}\sum_{j=1}^{N_y}\sum_{k=\frac{N_\xi}{2}+1}^{N_\xi}\sum_{l=1}^{N_\eta}\sum_{k'=1}^{N_\xi}\frac{\xi_k}{\Delta x}a_{i-\frac12,jk}^T\beta_{i-\frac12,j,k,k'}^\xi\ket{i-1}\ket{j-1}\ket{k-1}\ket{l-1}\bra{i-2}\bra{j-1}\bra{k'-1}\bra{l-1}\\
    -\sum_{i=2}^{N_x}\sum_{j=1}^{N_y}\sum_{k=\frac{N_\xi}{2}+1}^{N_\xi}\sum_{l=1}^{N_\eta}\frac{\xi_k}{\Delta x}a_{i-\frac12,jk}^R\ket{i-1}\ket{j-1}\ket{k-1}\ket{l-1}\bra{i-1}\bra{j-1}\bra{N_\xi-k}\bra{l-1}\\
    +\sum_{i=1}^{N_x-1}\sum_{j=1}^{N_y}\sum_{k=1}^{\frac{N_\xi}{2}}\sum_{l=1}^{N_\eta}\sum_{k'=1}^{N_\xi}\frac{\xi_k}{\Delta x}a^{T}_{i+\frac{1}{2},jk}\beta_{i+\frac12,j,k,k'}^\xi\ket{i-1}\ket{j-1}\ket{k-1}\ket{l-1}\bra{i}\bra{j-1}\bra{k'-1}\bra{l-1}\\
    +\sum_{i=1}^{N_x-1}\sum_{j=1}^{N_y}\sum_{k=1}^{\frac{N_\xi}{2}}\sum_{l=1}^{N_\eta}\frac{\xi_k}{\Delta x}a^{R}_{i+\frac{1}{2},jk}\ket{i-1}\ket{j-1}\ket{k-1}\ket{l-1}\bra{i-1}\bra{j-1}\bra{N_\xi-k}\bra{l-1}\\
    -\sum_{i=1}^{N_x}\sum_{j=1}^{N_y}\sum_{k=1}^{\frac{N_\xi}{2}}\sum_{l=1}^{N_\eta}\frac{\xi_k}{\Delta x}\ket{i-1}\ket{j-1}\ket{k-1}\ket{l-1}\bra{i-1}\bra{j-1}\bra{k-1}\bra{l-1},
\end{multline}
and a $N_xN_yN_\xi N_\eta$-dimensional vector 
\begin{multline}\label{equ:inhomogeneoustermx1}
    b_1=-\sum_{j=1}^{N_y}\sum_{k=\frac{N_\xi}{2}+1}^{N_\xi}\sum_{l=1}^{N_\eta}\frac{\xi_{k}}{\Delta x}f_{0jkl}(t)\ket{0}\ket{j-1}\ket{k-1}\ket{l-1}\\
    +\sum_{j=1}^{N_y}\sum_{k=1}^{\frac{N_\xi}{2}}\sum_{l=1}^{N_\eta}\frac{\xi_{k}}{\Delta x}f_{N_x+1, jkl}(t)\ket{N_x-1}\ket{j-1}\ket{k-1}\ket{l-1}.
\end{multline}
To facilitate subsequent complexity analysis, we discuss the sparsity of $A_1$. The sparsity in row is easy to compute $s_r(A_1)\le 3$. Denote the set of interfaces vertical to the $x$-axis as $\mathcal{I}_x=\{(x_{i+\frac{1}{2}},y_j):V(x,y)\text{ is discontinuous on }(x_{i+\frac{1}{2}},y_j)\}$. Then 
\[s_c(A_1)\le2\left\lceil\sqrt{\frac54+\frac{2}{\Delta\xi}\sqrt{\left(\frac{\Delta\xi}{2}\right)^2+2\max_{(x,y)\in\mathcal{I}_x}\left\{|V^-(x,y)-V^+(x,y)|\right\}}}-\frac12\right\rceil+2:=Q_1.\]

Similarly, for the term $\frac{\eta_{l}}{\Delta y}\left(f_{i,j+\frac{1}{2},kl}^{-}-f_{i,j-\frac{1}{2},kl}^{+}\right)$ in \eqref{equ:2dimLiouvilleClassicalSemidis},
when $\eta_l>0$, define
\begin{gather}
    a^T_{i,j+\frac{1}{2},l}=g^T\left(\eta_l^2+2\left(V_{i,j+\frac12}^+-V_{i,j+\frac12}^-\right)\right),\\
    a^R_{i,j+\frac{1}{2},l}=g^R\left(\eta_l^2+2\left(V_{i,j+\frac12}^+-V_{i,j+\frac12}^-\right)\right),
\end{gather}
when $\eta_l<0$,
\begin{gather}
    a^T_{i,j+\frac{1}{2},l}=g^T\left(\eta_l^2+2\left(V_{i,j+\frac12}^--V_{i,j+\frac12}^+\right)\right),\\
    a^R_{i,j+\frac{1}{2},l}=g^R\left(\eta_l^2+2\left(V_{i,j+\frac12}^--V_{i,j+\frac12}^+\right)\right),
\end{gather}
and
\begin{equation}\label{equ:betaeta}
    \beta_{i,j+\frac12,l,l'}^\eta=\left\{\begin{array}{ll}
    h\left(\sqrt{\eta_l^2+2\left(V_{i,j+\frac12}^+-V_{i,j+\frac12}^-\right)}-\eta_{l'}\right),     \\
    i=1,\ldots,N_x,\quad j=0,\ldots,N_y,\quad l=\frac{N_\eta}{2}+1,\ldots,N_\eta,\quad l'=1,\ldots,N_\eta, \\
    h\left(-\sqrt{\eta_l^2+2\left(V_{i,j+\frac12}^--V_{i,j+\frac12}^+\right)}-\eta_{l'}\right),     \\
    i=1,\ldots,N_x,\quad j=0,\ldots,N_y,\quad l=1,\ldots,\frac{N_\eta}{2},\quad l'=1,\ldots,N_\eta. 
    \end{array}\right.
\end{equation}
Then
\begin{multline*}
    \frac{\eta_{l}}{\Delta y}\left(f_{i,j+\frac{1}{2},kl}^{-}-f_{i,j-\frac{1}{2},kl}^{+}\right)=\\
    \left\{\begin{array}{l}
        \frac{\eta_{l}}{\Delta y}\left(f_{ijkl}-a_{i,j-\frac12,l}^T\sum_{l'=1}^{N_\eta}\beta_{i,j-\frac12,l,l'}^\eta f_{i,j-1,k,l'}-a_{i,j-\frac12,l}^R f_{i,j,k,l_1}\right),     \\ 
        \qquad i=1,\ldots,N_x,\quad j=1,\ldots,N_y,\quad k=1,\ldots,N_\xi,\quad l=\frac{N_\eta}{2}+1,\ldots,N_\eta, \\
        \frac{\eta_{l}}{\Delta y}\left(a^{T}_{i,j+\frac{1}{2},l}\sum_{l'=1}^{N_\eta}\beta_{i,j+\frac12,l,l'}^\eta f_{i,j+1,k,l'}
        +a^{R}_{i,j+\frac{1}{2},l}f_{i,j,k,l_{1}}-f_{ijkl}\right),    \\ 
        \qquad i=1,\ldots,N_x,\quad j=1,\ldots,N_y,\quad k=1,\ldots,N_\xi,\quad l=1,\ldots,\frac{N_\eta}{2}.
    \end{array}
    \right.
\end{multline*}
Write the above equations in the form $A_2\boldsymbol{f}+b_2$ of a product of a matrix and a vector with an inhomogeneous term representing boundary condition, where the equations and variables are in order of $\boldsymbol{f}$, we can get a $N_xN_yN_\xi N_\eta\times N_xN_yN_\xi N_\eta$ matrix
\begin{multline}\label{equ:mtx2}
    A_2=\sum_{i=1}^{N_x}\sum_{j=1}^{N_y}\sum_{k=1}^{N_\xi}\sum_{l=\frac{N_\eta}{2}+1}^{N_\eta}\frac{\eta_{l}}{\Delta y}\ket{i-1}\ket{j-1}\ket{k-1}\ket{l-1}\bra{i-1}\bra{j-1}\bra{k-1}\bra{l-1}\\
    -\sum_{i=1}^{N_x}\sum_{j=2}^{N_y}\sum_{k=1}^{N_\xi}\sum_{l=\frac{N_\eta}{2}+1}^{N_\eta}\sum_{l'=1}^{N_\eta}\frac{\eta_{l}}{\Delta y}a_{i,j-\frac12,l}^T\beta_{i,j-\frac12,l,l'}^\eta\ket{i-1}\ket{j-1}\ket{k-1}\ket{l-1}\bra{i-1}\bra{j-2}\bra{k-1}\bra{l'-1}\\
    -\sum_{i=1}^{N_x}\sum_{j=2}^{N_y}\sum_{k=1}^{N_\xi}\sum_{l=\frac{N_\eta}{2}+1}^{N_\eta}\frac{\eta_{l}}{\Delta y}a_{i,j-\frac12,l}^R\ket{i-1}\ket{j-1}\ket{k-1}\ket{l-1}\bra{i-1}\bra{j-1}\bra{k-1}\bra{N_\eta-l}\\
    +\sum_{i=1}^{N_x}\sum_{j=1}^{N_y-1}\sum_{k=1}^{N_\xi}\sum_{l=1}^{\frac{N_\eta}{2}}\sum_{l'=1}^{N_\eta}\frac{\eta_{l}}{\Delta y}a^{T}_{i,j+\frac{1}{2},l}\beta_{i,j+\frac12,l,l'}^\eta\ket{i-1}\ket{j-1}\ket{k-1}\ket{l-1}\bra{i-1}\bra{j}\bra{k-1}\bra{l'-1}\\
    +\sum_{i=1}^{N_x}\sum_{j=1}^{N_y-1}\sum_{k=1}^{N_\xi}\sum_{l=1}^{\frac{N_\eta}{2}}\frac{\eta_{l}}{\Delta y}a^{R}_{i,j+\frac{1}{2},l}\ket{i-1}\ket{j-1}\ket{k-1}\ket{l-1}\bra{i-1}\bra{j-1}\bra{k-1}\bra{N_\eta-l}\\
    -\sum_{i=1}^{N_x}\sum_{j=1}^{N_y}\sum_{k=1}^{N_\xi}\sum_{l=1}^{\frac{N_\eta}{2}}\frac{\eta_{l}}{\Delta y}\ket{i-1}\ket{j-1}\ket{k-1}\ket{l-1}\bra{i-1}\bra{j-1}\bra{k-1}\bra{l-1},
\end{multline}
and a $N_xN_yN_\xi N_\eta$-dimensional vector 
\begin{multline}\label{equ:inhomogeneoustermx2}
    b_2=-\sum_{i=1}^{N_x}\sum_{k=1}^{N_\xi}\sum_{l=\frac{N_\eta}{2}+1}^{N_\eta}\frac{\eta_{l}}{\Delta y}f_{i0kl}(t)\ket{i-1}\ket{0}\ket{k-1}\ket{l-1}\\
    +\sum_{i=1}^{N_x}\sum_{k=1}^{N_\xi}\sum_{l=1}^{\frac{N_\eta}{2}}\frac{\eta_{l}}{\Delta y}f_{i,N_y+1,k,l}(t)\ket{i-1}\ket{N_y-1}\ket{k-1}\ket{l-1}.
\end{multline}
The sparsity in row is easy to compute $s_r(A_2)\le 3$. Denote the set of interfaces vertical to the $y$-axis as $\mathcal{I}_y=\{(x_i,y_{j+\frac{1}{2}}):V(x,y)\text{ is discontinuous on }(x_i,y_{j+\frac{1}{2}})\}$. Then 
\[s_c(A_2)\le2\left\lceil\sqrt{\frac54+\frac{2}{\Delta\eta}\sqrt{\left(\frac{\Delta\eta}{2}\right)^2+2\max_{(x,y)\in\mathcal{I}_y}\left\{|V^-(x,y)-V^+(x,y)|\right\}}}-\frac12\right\rceil+2:=Q_2.\]

For the term $-\frac{V_{i+\frac{1}{2},j}^{-}-V_{i-\frac{1}{2},j}^{+}}{\Delta x\Delta\xi}\left(f_{ij,k+\frac{1}{2},l}-f_{ij,k-\frac{1}{2},l}\right)$ in \eqref{equ:2dimLiouvilleClassicalSemidis}, define 
\begin{equation*}
    d_{ij}^x=-\frac{V_{i+\frac{1}{2},j}^{-}-V_{i-\frac{1}{2},j}^{+}}{\Delta x\Delta\xi}.
\end{equation*}
Since the numerical fluxes $f_{ij,k+\frac{1}{2},l},f_{ij,k-\frac{1}{2},l}$ are defined using the upwind discretization, we can get
\begin{equation*}
    d_{ij}^x\left(f_{ij,k+\frac{1}{2},l}-f_{ij,k-\frac{1}{2},l}\right)=-\frac{|d_{ij}^x|+d_{ij}^x}{2}f_{ij,k-1,l}(t)+|d_{ij}^x|f_{ijkl}(t)-\frac{|d_{ij}^x|-d_{ij}^x}{2}f_{ij,k+1,l}(t).
\end{equation*}
When $d_{ij}^x>0$ (resp. $d_{ij}^x<0$), we impose inflow boundary condition at $(x_i,y_j,\xi_0,\eta_l)$ (resp. $(x_i,y_j,\xi_{N_\xi+1},\eta_l)$) and outflow boundary condition at $(x_i,y_j,\xi_{N_\xi+1},\eta_l)$ (resp. $(x_i,y_j,\xi_0,\eta_l)$).
Write the above equations in the form $A_3\boldsymbol{f}+b_3$ of a product of a matrix and a vector with an inhomogeneous term representing boundary condition, where the equations and variables are in order of $\boldsymbol{f}$, we can get a $N_xN_yN_\xi N_\eta\times N_xN_yN_\xi N_\eta$ matrix 
\begin{multline}\label{equ:mtx3}
    A_3=-\sum_{i=1}^{N_x}\sum_{j=1}^{N_y}\sum_{k=2}^{N_\xi}\sum_{l=1}^{N_\eta}\frac{|d_{ij}^x|+d_{ij}^x}{2}\ket{i-1}\ket{j-1}\ket{k-1}\ket{l-1}\bra{i-1}\bra{j-1}\bra{k-2}\bra{l-1}\\
    +\sum_{i=1}^{N_x}\sum_{j=1}^{N_y}\sum_{k=1}^{N_\xi}\sum_{l=1}^{N_\eta}|d_{ij}^x|\ket{i-1}\ket{j-1}\ket{k-1}\ket{l-1}\bra{i-1}\bra{j-1}\bra{k-1}\bra{l-1}\\ 
    -\sum_{i=1}^{N_x}\sum_{j=1}^{N_y}\sum_{k=1}^{N_\xi-1}\sum_{l=1}^{N_\eta}\frac{|d_{ij}^x|-d_{ij}^x}{2}\ket{i-1}\ket{j-1}\ket{k-1}\ket{l-1}\bra{i-1}\bra{j-1}\bra{k}\bra{l-1},
\end{multline}
and a $N_xN_yN_\xi N_\eta$-dimensional vector
\begin{multline}\label{equ:inhomogeneoustermxi1}
    b_3=-\sum_{i=1}^{N_x}\sum_{j=1}^{N_y}\sum_{l=1}^{N_\eta}\frac{|d_{ij}^x|+d_{ij}^x}{2}f_{ij0l}(t)\ket{i-1}\ket{j-1}\ket{0}\ket{l-1}\\ 
    -\sum_{i=1}^{N_x}\sum_{j=1}^{N_y}\sum_{l=1}^{N_\eta}\frac{|d_{ij}^x|-d_{ij}^x}{2}f_{ij,N_\xi+1,l}(t)\ket{i-1}\ket{j-1}\ket{N_\xi-1}\ket{l-1}.
\end{multline}

Similarly, for the term $-\frac{V_{i,j+\frac{1}{2}}^{-}-V_{i,j-\frac{1}{2}}^{+}}{\Delta y\Delta\eta}\left(f_{ijk,l+\frac{1}{2}}-f_{ijk,l-\frac{1}{2}}\right)$ in \eqref{equ:2dimLiouvilleClassicalSemidis}, define 
\begin{equation*}
    d_{ij}^y=-\frac{V_{i,j+\frac{1}{2}}^{-}-V_{i,j-\frac{1}{2}}^{+}}{\Delta y\Delta\eta}.
\end{equation*}
Then
\begin{equation*}
    d_{ij}^y\left(f_{ijk,l+\frac{1}{2}}-f_{ijk,l-\frac{1}{2}}\right)=-\frac{|d_{ij}^y|+d_{ij}^y}{2}f_{ijk,l-1}(t)+|d_{ij}^y|f_{ijkl}(t)-\frac{|d_{ij}^y|-d_{ij}^y}{2}f_{ijk,l+1}(t).
\end{equation*}
Write the above equations in the form $A_4\boldsymbol{f}+b_4$ of a product of a matrix and a vector with an inhomogeneous term representing boundary condition, where the equations and variables are in order of $\boldsymbol{f}$, we can get a $N_xN_yN_\xi N_\eta\times N_xN_yN_\xi N_\eta$ matrix 
\begin{multline}\label{equ:mtx4}
    A_4=-\sum_{i=1}^{N_x}\sum_{j=1}^{N_y}\sum_{k=1}^{N_\xi}\sum_{l=2}^{N_\eta}\frac{|d_{ij}^y|+d_{ij}^y}{2}\ket{i-1}\ket{j-1}\ket{k-1}\ket{l-1}\bra{i-1}\bra{j-1}\bra{k-1}\bra{l-2}\\
    +\sum_{i=1}^{N_x}\sum_{j=1}^{N_y}\sum_{k=1}^{N_\xi}\sum_{l=1}^{N_\eta}|d_{ij}^y|\ket{i-1}\ket{j-1}\ket{k-1}\ket{l-1}\bra{i-1}\bra{j-1}\bra{k-1}\bra{l-1}\\ 
    -\sum_{i=1}^{N_x}\sum_{j=1}^{N_y}\sum_{k=1}^{N_\xi}\sum_{l=1}^{N_\eta-1}\frac{|d_{ij}^y|-d_{ij}^y}{2}\ket{i-1}\ket{j-1}\ket{k-1}\ket{l-1}\bra{i-1}\bra{j-1}\bra{k-1}\bra{l},
\end{multline}
and a $N_xN_yN_\xi N_\eta$-dimensional vector
\begin{multline}\label{equ:inhomogeneoustermxi2}
    b_4=-\sum_{i=1}^{N_x}\sum_{j=1}^{N_y}\sum_{k=1}^{N_\xi}\frac{|d_{ij}^y|+d_{ij}^y}{2}f_{ijk0}(t)\ket{i-1}\ket{j-1}\ket{k-1}\ket{0}\\ 
    -\sum_{i=1}^{N_x}\sum_{j=1}^{N_y}\sum_{k=1}^{N_\xi}\frac{|d_{ij}^y|-d_{ij}^y}{2}f_{ijk,N_\eta+1}(t)\ket{i-1}\ket{j-1}\ket{k-1}\ket{N_\eta-1}.
\end{multline}

Combining \eqref{equ:mtx1}, 
\eqref{equ:mtx2}, 
\eqref{equ:mtx3}, \eqref{equ:mtx4}, \eqref{equ:inhomogeneoustermx1}, \eqref{equ:inhomogeneoustermx2}, \eqref{equ:inhomogeneoustermxi1} and \eqref{equ:inhomogeneoustermxi2}, we can get the ODE $\boldsymbol{f}'(t)=A\boldsymbol{f}(t)+\boldsymbol{b}(t)$ where $A=-A_1-A_2-A_3-A_4$ and $\boldsymbol{b}(t)=-b_1-b_2-b_3-b_4$.

\subsection{Schr\"odingerization}\label{sec:interface1Schro2d}
We first write the above ODE into the same homogeneous form \eqref{equ:interface1homogeneous} where $\boldsymbol{0}$ is a $N_xN_yN_\xi N_\eta\times N_xN_yN_\xi N_\eta$ matrix all of $0$ and $\boldsymbol{1}$ is a $N_xN_yN_\xi N_\eta$-dimensional vector all of $1$. The subsequent steps are the same as \eqref{equ:procedure}.

The following theorem gives the complexity of Hamiltonian-preserving scheme combined with Schr\"odingerization in two dimensional case.

\begin{theorem}\label{thm:2dim_complexity1}
    In the 2-dimensional case, given the initial quantum state $\ket{\boldsymbol{u}(0)}$, assume that the $x$ grid number, $y$ grid number, $\xi$ grid number, $\eta$ grid number and extended $p$ grid number are $N_x=2^{n_x}$, $N_y=2^{n_y}$, $N_\xi=2^{n_\xi}$, $N_\eta=2^{n_\eta}$ and $N_p=2^{n_p}$ respectively, $R-L=\mathcal O(1)$, the inhomogeneous term $\boldsymbol{b}$ is independent of time and the errors of the Hamiltonian-preserving scheme, spectral method in $p$ and Hamiltonian simulation are all $\epsilon$. With the Schr\"odingerization method, the state $\ket{\boldsymbol{u}(t)}$ can be simulated to time $T$, and success probability is at least $1-2\epsilon$ with
    \begin{itemize}
        \item $\mathcal{O}((Q_1+Q_2)T\epsilon^{-2}+\frac{\log\epsilon^{-1}}{\log\log\epsilon^{-1}})$ queries and a factor $\mathcal{O}(5\log(\epsilon^{-1})+\log(\epsilon^{-1})\polylog(\log(\epsilon^{-1})))$ additional quantum gates, when the solution has discontinuities only at interfaces,
        \item $\mathcal{O}((Q_1+Q_2)T\epsilon^{-3}+\frac{\log\epsilon^{-1}}{\log\log\epsilon^{-1}})$ queries and a factor $\mathcal{O}(9\log(\epsilon^{-1})+\log(\epsilon^{-1})\polylog(\log(\epsilon^{-1})))$ additional quantum gates, when the solution has discontinuities not only at interfaces,
    \end{itemize}
    where \[
Q_1:=2\left\lceil\sqrt{\frac54+\frac{2}{\Delta\xi}\sqrt{\left(\frac{\Delta\xi}{2}\right)^2+2\max_{(x,y)\in\mathcal{I}_x}\left\{|V^-(x,y)-V^+(x,y)|\right\}}}-\frac12\right\rceil+2\ge s_c(A_1)\ge 2
\]
and \[
Q_2:=2\left\lceil\sqrt{\frac54+\frac{2}{\Delta\eta}\sqrt{\left(\frac{\Delta\eta}{2}\right)^2+2\max_{(x,y)\in\mathcal{I}_y}\left\{|V^-(x,y)-V^+(x,y)|\right\}}}-\frac12\right\rceil+2\ge s_c(A_2)\ge 2.
\]
\end{theorem}
\begin{proof}
    We use Lemma \ref{lem:complexity} to give the proof. Matrix elements are specified to $\log(\epsilon^{-1})$ bits of precision.
    
    From previous discussion, 
    \begin{equation*}
        \begin{aligned}
            s(H)=s(H_1\otimes D_\mu-H_2\otimes I)&\le \max\{s(A+A^\dagger),s(A-A^\dagger)\}+s(\text{diag}(\boldsymbol{b}(t))/\delta)\\
            &\le \max\{s(A_1+A_1^\dagger),s(A_1-A_1^\dagger)\}+\max\{s(A_2+A_2^\dagger),s(A_2-A_2^\dagger)\}\\
            &\quad +\max\{s(A_3+A_3^\dagger),s(A_3-A_3^\dagger)\}+\max\{s(A_4+A_4^\dagger),s(A_4-A_4^\dagger)\}\\ 
            &\quad -3+1\\
            &\le (s_r(A_1)+s_c(A_1)-1)+(s_r(A_2)+s_c(A_2)-1)+(3)\\
            &\le Q_1+Q_2+2+2+3=Q_1+Q_2+7.
        \end{aligned}
    \end{equation*}
    In the $p$-direction, the initial value is only continuous but not differentiable. According to spectral method, in order to reach precision $\epsilon$, the mesh size $\Delta p$ is $\mathcal{O}(\epsilon)$.
    The total number of points in $x$, $y$, $\xi$ and $\eta$-direction is related to where the discontinuities lies.
    \begin{itemize}
        \item when the solution has discontinuities only at interfaces, the Hamiltonian-preserving scheme can give the first order precision. Thus $N_x=N_y=N_\xi=N_\eta=\mathcal{O}(\epsilon^{-1})$ and $\|H\|_{\max}\le\|H_{1}\otimes D_{\mu}\|_{\max}+\|H_{2}\|_{\max}\le\mathcal{O}(\left\|H_{1}\right\|_{\max}\mu_{N_p/2})=\mathcal{O}((N_x+N_y+N_\xi+N_\eta)/\Delta p)=\mathcal{O}(\epsilon^{-2})$. The query complexity is 
        \begin{equation*}
            \mathcal{O}((Q_1+Q_2)T\epsilon^{-2}+\frac{\log\epsilon^{-1}}{\log\log\epsilon^{-1}}).
        \end{equation*} 
        Since the quantum Fourier transform in Schr\"odingerization can be implemented by using $\mathcal{O}(m\log m)$ gates \cite{namApproximateQuantumFourier2020}, the number of overall additional quantum gates is 
        \begin{equation*}
            \mathcal{O}(5\log(\epsilon^{-1})+\log(\epsilon^{-1})\polylog(\log(\epsilon^{-1}))+\log(\epsilon^{-1})\log(\log(\epsilon^{-1}))).
        \end{equation*}
        \item when the solution has discontinuities not only at interfaces, the Hamiltonian-preserving scheme has only  half-order precision. Thus $N_x=N_y=N_\xi=N_\eta=\mathcal{O}(\epsilon^{-2})$ and $\|H\|_{\max}\le\mathcal{O}(\epsilon^{-3})$. The query complexity is 
        \begin{equation*} 
            \mathcal{O}((Q_1+Q_2)T\epsilon^{-3}+\frac{\log\epsilon^{-1}}{\log\log\epsilon^{-1}}),
        \end{equation*} 
        and the number of overall additional quantum gates is 
        \begin{equation*}
            \mathcal{O}(9\log(\epsilon^{-1})+\log(\epsilon^{-1})\polylog(\log(\epsilon^{-1}))+\log(\epsilon^{-1})\log(\log(\epsilon^{-1}))).
        \end{equation*}
    \end{itemize}
\end{proof}

\begin{remark}
    In the classical implementation of the Hamiltonian-preserving schemes, one can use any time discretization for the time derivative. Here we use the forward Euler time discretization as an example to analyze the classical complexity which contains the number of products and quotients. After discretization, the ODE $\boldsymbol{f}'(t)=A\boldsymbol{f}(t)+\boldsymbol{b}(t)$ becomes $\boldsymbol{f}^{n+1}=(\Delta tA+I)\boldsymbol{f}^n+\Delta t\boldsymbol{b}$ for $n=0,\ldots,{N_t}-1$. To satisfy the hyperbolic CFL condition
    \begin{equation*}
         \Delta t\max_{i,j,k,l}\left[\frac{|\xi_{k}|}{\Delta x}+\frac{|\eta_{l}|}{\Delta y}+\frac{\left|V_{i+\frac{1}{2},j}^{-}-V_{i-\frac{1}{2},j}^{+}\right|}{\Delta x\Delta\xi}+\frac{\left|V_{i,j+\frac{1}{2}}^{-}-V_{i,j-\frac{1}{2}}^{+}\right|}{\Delta y\Delta\eta}\right]\leqslant1,
    \end{equation*}
    a time step $\Delta t=\mathcal{O}(\Delta x,\Delta y,\Delta \xi,\Delta \eta)$ is needed \cite{jinHamiltonianPreservingSchemesLiouville2005}, so we take $N_t=\mathcal{O}(TN_x)$. Then the classical complexity is 
    \begin{equation}
        \mathcal{O}(s_r(\Delta tA+I)N_xN_y N_\xi N_\eta N_t)=\begin{cases}
            \mathcal{O}(7T\epsilon^{-5}) & \text{when the solution has discontinuities only at interfaces},\\
            \mathcal{O}(7T\epsilon^{-10})& \text{when the solution has discontinuities not only at interfaces}.
        \end{cases}
    \end{equation}
    There is a {\it polynomial} advantage in $\epsilon$ for quantum algorithms.
\end{remark}

\section{Quantum simulation of Liouville equation with interface condition \eqref{equ:ddimLiouvilleClassicalinterfacecondition} in $n$ spatial dimensions}\label{sec:interfacend}

In this section, we extend the construction in Sections \ref{sec:interface1d} and \ref{sec:interface2d} to the $n$-dimensional Liouville equation
\begin{equation}
    f_t+\sum_{\alpha=1}^n v_\alpha \partial_{x_\alpha} f
    -\sum_{\alpha=1}^n \partial_{x_\alpha}V\,\partial_{v_\alpha}f=0,
    \qquad t>0,\quad \boldsymbol{x},\boldsymbol{v}\in \mathbb R^n .
    \label{equ:nD_Liouville}
\end{equation}
We assume that the discontinuities of $V$ are aligned with the Cartesian grids. More precisely, each
discontinuous interface is contained in a grid face normal to one of the coordinate directions. Across a
face normal to the $x_\alpha$-direction, the tangential velocity components remain unchanged, while the
normal velocity component is determined by the constant Hamiltonian condition
\begin{equation}
    \frac12 (v_\alpha^+)^2+V^+
    =
    \frac12 (v_\alpha^-)^2+V^-,
    \qquad
    v_\beta^+=v_\beta^- ,\quad \beta\neq \alpha .
    \label{equ:nD_HP_condition}
\end{equation}
Thus the interface condition for the density function is
\begin{equation}
    f(t,\boldsymbol{x}^+,v_\alpha^+,v_\perp)
    =
    f(t,\boldsymbol{x}^-,v_\alpha^-,v_\perp),
    \label{equ:nD_interface_condition}
\end{equation}
where $v_\perp=(v_1,\ldots,v_{\alpha-1},v_{\alpha+1},\ldots,v_n)$ denotes the tangential velocity
components.

This \(n\)-dimensional extension is another new component compared with the geometrical optics work \cite{jinQuantumSimulationLiouville2026}, which treated the one and two dimensional constructions. In the present setting, each coordinate direction may contain grid aligned potential interfaces, and the normal velocity update in each direction is governed by the corresponding energy-threshold rule. The analysis below therefore keeps track of the sparsity contributed by each spatial direction and combines them to obtain the \(n\)-dimensional sparse-access query complexity.

\subsection{The original algorithm}

For each $\alpha=1,\ldots,n$, we use a uniform mesh in the $x_\alpha$-direction and in the
$v_\alpha$-direction. The grid points are denoted by
\[
    x_{\alpha,i_\alpha+\frac12},\qquad i_\alpha=0,\ldots,N_{x_\alpha},
    \qquad
    v_{\alpha,j_\alpha+\frac12},\qquad j_\alpha=0,\ldots,N_{v_\alpha}.
\]
The cell centers are
\[
    x_{\alpha,i_\alpha}
    =
    \frac12\left(x_{\alpha,i_\alpha+\frac12}
    +x_{\alpha,i_\alpha-\frac12}\right),
    \qquad
    v_{\alpha,j_\alpha}
    =
    \frac12\left(v_{\alpha,j_\alpha+\frac12}
    +v_{\alpha,j_\alpha-\frac12}\right).
\]
The mesh sizes are denoted by
\[
    \Delta x_\alpha=x_{\alpha,i_\alpha+\frac12}-x_{\alpha,i_\alpha-\frac12},
    \qquad
    \Delta v_\alpha=v_{\alpha,j_\alpha+\frac12}-v_{\alpha,j_\alpha-\frac12}.
\]
Let
\[
    \boldsymbol i=(i_1,\ldots,i_n),\qquad
    \boldsymbol j=(j_1,\ldots,j_n),
\]
and let $\boldsymbol e_\alpha$ be the standard basis vector in the $\alpha$-th direction. We use the
notation
\[
    \mathcal I_x=\prod_{\alpha=1}^n \{1,\ldots,N_{x_\alpha}\},
    \qquad
    \mathcal I_v=\prod_{\alpha=1}^n \{1,\ldots,N_{v_\alpha}\}.
\]
The cell average is defined by
\begin{equation*}
    f_{\boldsymbol i,\boldsymbol j}
    =
    \frac{1}{
    \prod_{\alpha=1}^n \Delta x_\alpha \Delta v_\alpha}
    \int_{C_{\boldsymbol i}^x}
    \int_{C_{\boldsymbol j}^v}
    f(\boldsymbol{x},\boldsymbol{v},t)\,d\boldsymbol{v}\,d\boldsymbol{x} ,
    \label{equ:nD_cell_average}
\end{equation*}
where  
\[
    C_{\boldsymbol i}^x
    =
    \prod_{\alpha=1}^n
    \left[x_{\alpha,i_\alpha-\frac12},x_{\alpha,i_\alpha+\frac12}\right],
    \qquad
    C_{\boldsymbol j}^v
    =
    \prod_{\alpha=1}^n
    \left[v_{\alpha,j_\alpha-\frac12},v_{\alpha,j_\alpha+\frac12}\right].
\]

At a cell interface normal to the $x_\alpha$-direction, we denote the limiting values of $V$
on the two sides of the interface by
\[
    V^{-}_{\boldsymbol i+\frac12\boldsymbol e_\alpha},
    \qquad
    V^{+}_{\boldsymbol i+\frac12\boldsymbol e_\alpha}.
\]
If $V$ is continuous at this face, then
\[
    V^{-}_{\boldsymbol i+\frac12\boldsymbol e_\alpha}
    =
    V^{+}_{\boldsymbol i+\frac12\boldsymbol e_\alpha}.
\]
As in the one- and two-dimensional cases, $V$ is approximated by a piecewise multilinear function on
each cell.

The semi-discrete Hamiltonian-preserving scheme is
\begin{equation}
    \left(f_{\boldsymbol i,\boldsymbol j}\right)_t
    +
    \sum_{\alpha=1}^n
    \frac{v_{\alpha,j_\alpha}}{\Delta x_\alpha}
    \left(
        f^{-}_{\boldsymbol i+\frac12\boldsymbol e_\alpha,\boldsymbol j}
        -
        f^{+}_{\boldsymbol i-\frac12\boldsymbol e_\alpha,\boldsymbol j}
    \right)
    +
    \sum_{\alpha=1}^n
    d^\alpha_{\boldsymbol i}
    \left(
        f_{\boldsymbol i,\boldsymbol j+\frac12\boldsymbol e_\alpha}
        -
        f_{\boldsymbol i,\boldsymbol j-\frac12\boldsymbol e_\alpha}
    \right)
    =0,
\label{equ:nD_semidiscrete}
\end{equation}
where
\begin{equation*}
    d^\alpha_{\boldsymbol i}
    =
    -
    \frac{
    V^{-}_{\boldsymbol i+\frac12\boldsymbol e_\alpha}
    -
    V^{+}_{\boldsymbol i-\frac12\boldsymbol e_\alpha}
    }{
    \Delta x_\alpha \Delta v_\alpha
    } .
    \label{equ:nD_d_alpha}
\end{equation*}
The fluxes
\[
    f_{\boldsymbol i,\boldsymbol j+\frac12\boldsymbol e_\alpha},
    \qquad
    f_{\boldsymbol i,\boldsymbol j-\frac12\boldsymbol e_\alpha}
\]
are defined by the usual upwind discretization in the $v_\alpha$-direction. The split fluxes
\[
    f^{-}_{\boldsymbol i+\frac12\boldsymbol e_\alpha,\boldsymbol j},
    \qquad
    f^{+}_{\boldsymbol i-\frac12\boldsymbol e_\alpha,\boldsymbol j}
\]
are obtained by applying the one-dimensional transmission/reflection rule of the Algorithm \ref{alg:fluxClassical1} in the normal direction
$x_\alpha$, with all tangential indices fixed.

We now describe these split fluxes more explicitly. Define the same step functions \eqref{equ:stepfun}
and the hat function in the $v_\alpha$-direction
\begin{equation}
    h_\alpha(z)=\max\left(1-\frac{|z|}{\Delta v_\alpha},0\right).
    \label{equ:nD_hat_function}
\end{equation}
For a multi-index $\boldsymbol j$, let
\[
    \boldsymbol j^{(\alpha:k)}
    =
    (j_1,\ldots,j_{\alpha-1},k,j_{\alpha+1},\ldots,j_n),
\]
and, since the truncated interval in the $v_\alpha$-direction is symmetric about zero, define the
reflected velocity index
\[
    \bar j_\alpha=N_{v_\alpha}-j_\alpha+1,
    \qquad
    \bar{\boldsymbol j}^{\,\alpha}
    =
    \boldsymbol j^{(\alpha:\bar j_\alpha)} .
\]

At the face $\boldsymbol i+\frac12\boldsymbol e_\alpha$, set
\[
    \Delta V_{\boldsymbol i+\frac12\boldsymbol e_\alpha}
    =
    V^{+}_{\boldsymbol i+\frac12\boldsymbol e_\alpha}
    -
    V^{-}_{\boldsymbol i+\frac12\boldsymbol e_\alpha}.
\]
For $v_{\alpha,j_\alpha}\neq 0$, define
\begin{equation*}
    \Theta_{\boldsymbol i+\frac12\boldsymbol e_\alpha,j_\alpha}
    =
    \begin{cases}
    v_{\alpha,j_\alpha}^2
    +2\Delta V_{\boldsymbol i+\frac12\boldsymbol e_\alpha},
    & v_{\alpha,j_\alpha}>0,\\[2mm]
    v_{\alpha,j_\alpha}^2
    -2\Delta V_{\boldsymbol i+\frac12\boldsymbol e_\alpha},
    & v_{\alpha,j_\alpha}<0.
    \end{cases}
    \label{equ:nD_theta}
\end{equation*}
Then the transmission and reflection coefficients are
\begin{equation}
    a^{T}_{\boldsymbol i+\frac12\boldsymbol e_\alpha,j_\alpha}
    =
    g^T\left(
    \Theta_{\boldsymbol i+\frac12\boldsymbol e_\alpha,j_\alpha}
    \right),
    \qquad
    a^{R}_{\boldsymbol i+\frac12\boldsymbol e_\alpha,j_\alpha}
    =
    g^R\left(
    \Theta_{\boldsymbol i+\frac12\boldsymbol e_\alpha,j_\alpha}
    \right).
    \label{equ:nD_trans_refl_coeff}
\end{equation}
Thus
\[
    a^{T}_{\boldsymbol i+\frac12\boldsymbol e_\alpha,j_\alpha}
    +
    a^{R}_{\boldsymbol i+\frac12\boldsymbol e_\alpha,j_\alpha}
    =1.
\]
When $\Theta_{\boldsymbol i+\frac12\boldsymbol e_\alpha,j_\alpha}>0,$
we define the interpolation coefficient by
\begin{equation}
    \beta^\alpha_{\boldsymbol i+\frac12\boldsymbol e_\alpha,j_\alpha,k}
    =
    h_\alpha\left(
    \operatorname{sgn}(v_{\alpha,j_\alpha})
    \sqrt{
    \Theta_{\boldsymbol i+\frac12\boldsymbol e_\alpha,j_\alpha}
    }
    -v_{\alpha,k}
    \right),
    \qquad k=1,\ldots,N_{v_\alpha}.
    \label{equ:nD_beta}
\end{equation}
When $\Theta_{\boldsymbol i+\frac12\boldsymbol e_\alpha,j_\alpha}\leq 0,$
the transmission coefficient is zero, and the corresponding transmission term does not appear in the
numerical flux.

With these notations, for $v_{\alpha,j_\alpha}>0$,
\begin{equation*}
\begin{aligned}
&\frac{v_{\alpha,j_\alpha}}{\Delta x_\alpha}
    \left(
        f^{-}_{\boldsymbol i+\frac12\boldsymbol e_\alpha,\boldsymbol j}
        -
        f^{+}_{\boldsymbol i-\frac12\boldsymbol e_\alpha,\boldsymbol j}
    \right)
\\
&=
    \frac{v_{\alpha,j_\alpha}}{\Delta x_\alpha}
    \left[
        f_{\boldsymbol i,\boldsymbol j}
        -
        a^{T}_{\boldsymbol i-\frac12\boldsymbol e_\alpha,j_\alpha}
        \sum_{k=1}^{N_{v_\alpha}}
        \beta^\alpha_{\boldsymbol i-\frac12\boldsymbol e_\alpha,j_\alpha,k}
        f_{\boldsymbol i-\boldsymbol e_\alpha,\boldsymbol j^{(\alpha:k)}}
        -
        a^{R}_{\boldsymbol i-\frac12\boldsymbol e_\alpha,j_\alpha}
        f_{\boldsymbol i,\bar{\boldsymbol j}^{\,\alpha}}
    \right].
\end{aligned}
\label{equ:nD_x_flux_positive}
\end{equation*}
For $v_{\alpha,j_\alpha}<0$,
\begin{equation*}
\begin{aligned}
&\frac{v_{\alpha,j_\alpha}}{\Delta x_\alpha}
    \left(
        f^{-}_{\boldsymbol i+\frac12\boldsymbol e_\alpha,\boldsymbol j}
        -
        f^{+}_{\boldsymbol i-\frac12\boldsymbol e_\alpha,\boldsymbol j}
    \right)
\\
&=
    \frac{v_{\alpha,j_\alpha}}{\Delta x_\alpha}
    \left[
        a^{T}_{\boldsymbol i+\frac12\boldsymbol e_\alpha,j_\alpha}
        \sum_{k=1}^{N_{v_\alpha}}
        \beta^\alpha_{\boldsymbol i+\frac12\boldsymbol e_\alpha,j_\alpha,k}
        f_{\boldsymbol i+\boldsymbol e_\alpha,\boldsymbol j^{(\alpha:k)}}
        +
        a^{R}_{\boldsymbol i+\frac12\boldsymbol e_\alpha,j_\alpha}
        f_{\boldsymbol i,\bar{\boldsymbol j}^{\,\alpha}}
        -
        f_{\boldsymbol i,\boldsymbol j}
    \right].
\end{aligned}
\label{equ:nD_x_flux_negative}
\end{equation*}
If $V$ is continuous at the corresponding face, then the above formula reduces to the usual upwind
flux.

\subsection{The ODE form}

We now convert the semi-discrete scheme \eqref{equ:nD_semidiscrete} into the ODE form. Assume that
all grid numbers $N_{x_\alpha}$ and $N_{v_\alpha}$ are even and that the truncated intervals in all
$x_\alpha$- and $v_\alpha$-directions are symmetric about zero. We also assume that discontinuities do
not appear at the outer boundaries $x_{\alpha,\frac12}$ or $x_{\alpha,N_{x_\alpha}+\frac12},\alpha=1,\ldots,n.$

Define
\begin{equation*}
    \boldsymbol f(t)
    =
    \sum_{\boldsymbol i\in\mathcal I_x}
    \sum_{\boldsymbol j\in\mathcal I_v}
    f_{\boldsymbol i,\boldsymbol j}(t)
    \left(
    \bigotimes_{\alpha=1}^n |i_\alpha-1\rangle
    \right)
    \left(
    \bigotimes_{\alpha=1}^n |j_\alpha-1\rangle
    \right).
    \label{equ:nD_vector_f}
\end{equation*}
This definition specifies both the order of the equations and the order of the variables. For brevity,
we write
\[
    |\boldsymbol i,\boldsymbol j\rangle
    =
    \left(
    \bigotimes_{\alpha=1}^n |i_\alpha-1\rangle
    \right)
    \left(
    \bigotimes_{\alpha=1}^n |j_\alpha-1\rangle
    \right).
\]

For each $\alpha=1,\ldots,n$, let $A_\alpha^x$ be the matrix corresponding to the $x_\alpha$-transport
term, where the split numerical fluxes at cell interfaces are defined by the
Hamiltonian-preserving scheme. It is given by
\begin{equation}
\begin{aligned}
A_\alpha^x
=&
\sum_{\substack{\boldsymbol i\in\mathcal I_x,\ \boldsymbol j\in\mathcal I_v\\
v_{\alpha,j_\alpha}>0}}
\frac{v_{\alpha,j_\alpha}}{\Delta x_\alpha}
|\boldsymbol i,\boldsymbol j\rangle
\langle \boldsymbol i,\boldsymbol j|
\\
&-
\sum_{\substack{\boldsymbol i\in\mathcal I_x,\ i_\alpha\geq 2\\
\boldsymbol j\in\mathcal I_v,\ v_{\alpha,j_\alpha}>0}}
\sum_{k=1}^{N_{v_\alpha}}
\frac{v_{\alpha,j_\alpha}}{\Delta x_\alpha}
a^{T}_{\boldsymbol i-\frac12\boldsymbol e_\alpha,j_\alpha}
\beta^\alpha_{\boldsymbol i-\frac12\boldsymbol e_\alpha,j_\alpha,k}
|\boldsymbol i,\boldsymbol j\rangle
\langle \boldsymbol i-\boldsymbol e_\alpha,\boldsymbol j^{(\alpha:k)}|
\\
&-
\sum_{\substack{\boldsymbol i\in\mathcal I_x,\ i_\alpha\geq 2\\
\boldsymbol j\in\mathcal I_v,\ v_{\alpha,j_\alpha}>0}}
\frac{v_{\alpha,j_\alpha}}{\Delta x_\alpha}
a^{R}_{\boldsymbol i-\frac12\boldsymbol e_\alpha,j_\alpha}
|\boldsymbol i,\boldsymbol j\rangle
\langle \boldsymbol i,\bar{\boldsymbol j}^{\,\alpha}|
\\
&+
\sum_{\substack{\boldsymbol i\in\mathcal I_x,\ i_\alpha\leq N_{x_\alpha}-1\\
\boldsymbol j\in\mathcal I_v,\ v_{\alpha,j_\alpha}<0}}
\sum_{k=1}^{N_{v_\alpha}}
\frac{v_{\alpha,j_\alpha}}{\Delta x_\alpha}
a^{T}_{\boldsymbol i+\frac12\boldsymbol e_\alpha,j_\alpha}
\beta^\alpha_{\boldsymbol i+\frac12\boldsymbol e_\alpha,j_\alpha,k}
|\boldsymbol i,\boldsymbol j\rangle
\langle \boldsymbol i+\boldsymbol e_\alpha,\boldsymbol j^{(\alpha:k)}|
\\
&+
\sum_{\substack{\boldsymbol i\in\mathcal I_x,\ i_\alpha\leq N_{x_\alpha}-1\\
\boldsymbol j\in\mathcal I_v,\ v_{\alpha,j_\alpha}<0}}
\frac{v_{\alpha,j_\alpha}}{\Delta x_\alpha}
a^{R}_{\boldsymbol i+\frac12\boldsymbol e_\alpha,j_\alpha}
|\boldsymbol i,\boldsymbol j\rangle
\langle \boldsymbol i,\bar{\boldsymbol j}^{\,\alpha}|
\\
&-
\sum_{\substack{\boldsymbol i\in\mathcal I_x,\ \boldsymbol j\in\mathcal I_v\\
v_{\alpha,j_\alpha}<0}}
\frac{v_{\alpha,j_\alpha}}{\Delta x_\alpha}
|\boldsymbol i,\boldsymbol j\rangle
\langle \boldsymbol i,\boldsymbol j| .
\end{aligned}
\label{equ:nD_Ax_alpha}
\end{equation}
The corresponding boundary vector is
\begin{equation}
\begin{aligned}
b_\alpha^x(t)
=&
-
\sum_{\substack{\boldsymbol i\in\mathcal I_x,\ i_\alpha=1\\
\boldsymbol j\in\mathcal I_v,\ v_{\alpha,j_\alpha}>0}}
\frac{v_{\alpha,j_\alpha}}{\Delta x_\alpha}
f_{\boldsymbol i-\boldsymbol e_\alpha,\boldsymbol j}(t)
|\boldsymbol i,\boldsymbol j\rangle
\\
&+
\sum_{\substack{\boldsymbol i\in\mathcal I_x,\ i_\alpha=N_{x_\alpha}\\
\boldsymbol j\in\mathcal I_v,\ v_{\alpha,j_\alpha}<0}}
\frac{v_{\alpha,j_\alpha}}{\Delta x_\alpha}
f_{\boldsymbol i+\boldsymbol e_\alpha,\boldsymbol j}(t)
|\boldsymbol i,\boldsymbol j\rangle .
\end{aligned}
\label{equ:nD_bx_alpha}
\end{equation}

For the $v_\alpha$-direction term, using the upwind approximation gives
\begin{equation*}
\begin{aligned}
    d^\alpha_{\boldsymbol i}
    \left(
        f^{v_\alpha}_{\boldsymbol i,\boldsymbol j+\frac12\boldsymbol e_\alpha}
        -
        f^{v_\alpha}_{\boldsymbol i,\boldsymbol j-\frac12\boldsymbol e_\alpha}
    \right)
    =
    -
    \frac{|d^\alpha_{\boldsymbol i}|+d^\alpha_{\boldsymbol i}}{2}
    f_{\boldsymbol i,\boldsymbol j-\boldsymbol e_\alpha}
    +
    |d^\alpha_{\boldsymbol i}|
    f_{\boldsymbol i,\boldsymbol j}
    -
    \frac{|d^\alpha_{\boldsymbol i}|-d^\alpha_{\boldsymbol i}}{2}
    f_{\boldsymbol i,\boldsymbol j+\boldsymbol e_\alpha}.
\end{aligned}
\label{equ:nD_v_flux}
\end{equation*}
Thus the corresponding matrix $A_\alpha^v$ is
\begin{equation}
\begin{aligned}
A_\alpha^v
=&
-
\sum_{\substack{\boldsymbol i\in\mathcal I_x,\ \boldsymbol j\in\mathcal I_v\\
j_\alpha\geq 2}}
\frac{|d^\alpha_{\boldsymbol i}|+d^\alpha_{\boldsymbol i}}{2}
|\boldsymbol i,\boldsymbol j\rangle
\langle \boldsymbol i,\boldsymbol j-\boldsymbol e_\alpha|
\\
&+
\sum_{\boldsymbol i\in\mathcal I_x}
\sum_{\boldsymbol j\in\mathcal I_v}
|d^\alpha_{\boldsymbol i}|
|\boldsymbol i,\boldsymbol j\rangle
\langle \boldsymbol i,\boldsymbol j|
\\
&-
\sum_{\substack{\boldsymbol i\in\mathcal I_x,\ \boldsymbol j\in\mathcal I_v\\
j_\alpha\leq N_{v_\alpha}-1}}
\frac{|d^\alpha_{\boldsymbol i}|-d^\alpha_{\boldsymbol i}}{2}
|\boldsymbol i,\boldsymbol j\rangle
\langle \boldsymbol i,\boldsymbol j+\boldsymbol e_\alpha| .
\end{aligned}
\label{equ:nD_Av_alpha}
\end{equation}
The boundary vector in the $v_\alpha$-direction is
\begin{equation}
\begin{aligned}
b_\alpha^v(t)
=&
-
\sum_{\substack{\boldsymbol i\in\mathcal I_x,\ \boldsymbol j\in\mathcal I_v\\
j_\alpha=1}}
\frac{|d^\alpha_{\boldsymbol i}|+d^\alpha_{\boldsymbol i}}{2}
f_{\boldsymbol i,\boldsymbol j-\boldsymbol e_\alpha}(t)
|\boldsymbol i,\boldsymbol j\rangle
\\
&-
\sum_{\substack{\boldsymbol i\in\mathcal I_x,\ \boldsymbol j\in\mathcal I_v\\
j_\alpha=N_{v_\alpha}}}
\frac{|d^\alpha_{\boldsymbol i}|-d^\alpha_{\boldsymbol i}}{2}
f_{\boldsymbol i,\boldsymbol j+\boldsymbol e_\alpha}(t)
|\boldsymbol i,\boldsymbol j\rangle .
\end{aligned}
\label{equ:nD_bv_alpha}
\end{equation}

Combining the $x_\alpha$-transport terms and the $v_\alpha$-transport terms, we obtain
\begin{equation}
    \boldsymbol f'(t)=A\boldsymbol f(t)+b(t),
    \qquad
    A=-\sum_{\alpha=1}^n \left(A_\alpha^x+A_\alpha^v\right),
    \qquad
    b(t)=-\sum_{\alpha=1}^n \left(b_\alpha^x(t)+b_\alpha^v(t)\right).
    \label{equ:nD_ODE}
\end{equation}

We next estimate the sparsity of the matrices. Let
\[
    \mathcal I_\alpha
    =
    \left\{
    \boldsymbol x:
    V \text{ is discontinuous on a face normal to the } x_\alpha\text{-direction}
    \right\}.
\]
Define
\begin{equation}
    Q_\alpha
    :=
    2\left\lceil
    \sqrt{
    \frac54+
    \frac{2}{\Delta v_\alpha}
    \sqrt{
        \left(\frac{\Delta v_\alpha}{2}\right)^2
        +
        2\max_{\boldsymbol x\in\mathcal I_\alpha}
        |V^-(\boldsymbol x)-V^+(\boldsymbol x)|
    }}
    -\frac12
    \right\rceil
    +2 .
    \label{equ:nD_Q_alpha}
\end{equation}
By the same argument as in the one- and two-dimensional cases, the hat function gives at most two
nonzero interpolation entries in each row, and the column sparsity is controlled by the largest jump of the potential. Hence
\begin{equation*}
    s_r(A_\alpha^x)\leq 3,
    \qquad
    s_c(A_\alpha^x)\leq Q_\alpha .
    \label{equ:nD_sparsity_Ax}
\end{equation*}
Moreover, the sum of the upwind matrices in the velocity variables satisfies
\begin{equation*}
    s\left(\sum_{\alpha=1}^n A_\alpha^v\right)\leq 2n+1 .
    \label{equ:nD_sparsity_Av}
\end{equation*}
Consequently, the sparsity of the full semi-discrete matrix is bounded by
\begin{equation}
    s(A+A^\dagger)
    =
    \mathcal O\left(
    \sum_{\alpha=1}^n Q_\alpha+n
    \right).
    \label{equ:nD_sparsity_A}
\end{equation}

\subsection{Schr\"odingerization}

We first rewrite \eqref{equ:nD_ODE} into a homogeneous system \eqref{equ:interface1homogeneous}. Let
    $N_{\rm ph}=\prod_{\alpha=1}^n N_{x_\alpha}N_{v_\alpha}.$
Then $\boldsymbol{0}$ is the $N_{\rm ph}\times N_{\rm ph}$ zero matrix and $\boldsymbol 1$ is the
$N_{\rm ph}$-dimensional vector with all components equal to one.
The whole procedure is the same as \eqref{equ:procedure}.

The following theorem gives the query complexity of the Schr\"odingerization based Hamiltonian-preserving scheme in $n$ dimensions.

\begin{theorem}
In the $n$-dimensional case, given the initial quantum state $|\boldsymbol u(0)\rangle$, assume that
\[
    N_{x_\alpha}=2^{n_{x_\alpha}},
    \qquad
    N_{v_\alpha}=2^{n_{v_\alpha}},
    \qquad
    N_p=2^{n_p},
    \qquad \alpha=1,\ldots,n,
\]
$R-L=\mathcal O(1)$ and that the inhomogeneous term $\boldsymbol b$ is independent of time. Suppose that the errors of the
Hamiltonian-preserving scheme, the spectral method in $p$, and the Hamiltonian simulation are all
$\epsilon$. Then the state $|\boldsymbol u(t)\rangle$ can be simulated up to time $T$ with success probability at
least $1-2\epsilon$ with
\begin{itemize}
    \item \[
                \mathcal O\left(
                \left(\sum_{\alpha=1}^n Q_\alpha+n\right)nT\epsilon^{-2}
                +
                \frac{\log(\epsilon^{-1})}{\log\log(\epsilon^{-1})}
                \right)
            \]
            queries and a factor
            \[
                \mathcal O\left(
                (2n+1)\log(\epsilon^{-1})
                +
                \log(\epsilon^{-1})\operatorname{polylog}(\log(\epsilon^{-1}))
                \right)
            \]
            additional quantum gates, when the solution has discontinuities only at interfaces.
    \item
            \[
                \mathcal O\left(
                \left(\sum_{\alpha=1}^n Q_\alpha+n\right)nT\epsilon^{-3}
                +
                \frac{\log(\epsilon^{-1})}{\log\log(\epsilon^{-1})}
                \right)
            \]
            queries and a factor
            \[
                \mathcal O\left(
                (4n+1)\log(\epsilon^{-1})
                +
                \log(\epsilon^{-1})\operatorname{polylog}(\log(\epsilon^{-1}))
                \right)
            \]
            additional quantum gates, when the solution has discontinuities not only at interfaces.
\end{itemize}
\end{theorem}

\begin{proof}
The proof follows from Lemma \ref{lem:complexity} and from the sparsity estimate above. The matrix elements are
specified to $\log(\epsilon^{-1})$ bits of precision. Since the Hamiltonian after Schr\"odingerization is
    $H=H_1\otimes D_\mu-H_2\otimes I,$
we have
\[
    s( H)
    =
    \mathcal O\left(
    \sum_{\alpha=1}^n Q_\alpha+n
    \right).
\]
In the auxiliary $p$-direction, the initial data are only continuous but not differentiable. Hence, for
a first-order extension in the $p$ variable, we take
    $\Delta p=\mathcal O(\epsilon).$
Moreover,
\[
    \|H\|_{\max}
    \leq
    \mathcal O\left(
    \frac{\|H_1\|_{\max}}{\Delta p}
    +
    \|H_2\|_{\max}
    \right).
\]

\begin{itemize}
\item If the solution has discontinuities only at interfaces, then the Hamiltonian-preserving scheme is
first-order accurate. Thus
\[
    N_{x_\alpha}=N_{v_\alpha}=\mathcal O(\epsilon^{-1}),
    \qquad \alpha=1,\ldots,n.
\]
Therefore
\[
    \|H\|_{\max}
    =
    \mathcal O\left(
    \frac{\sum_{\alpha=1}^n (N_{x_\alpha}+N_{v_\alpha})}{\Delta p}
    \right)
    =
    \mathcal O(n\epsilon^{-2}).
\]
Lemma \ref{lem:complexity} gives the query complexity
\[
    \mathcal O\left(
    \left(\sum_{\alpha=1}^n Q_\alpha+n\right)nT\epsilon^{-2}
    +
    \frac{\log(\epsilon^{-1})}{\log\log(\epsilon^{-1})}
    \right).
\]
The number of qubits in the physical and velocity variables is
\[
    \sum_{\alpha=1}^n
    \left(n_{x_\alpha}+n_{v_\alpha}\right)
    =
    2n\log(\epsilon^{-1}),
\]
and the $p$-register contributes one more logarithmic factor. Hence the additional gate factor is
\[
    \mathcal O\left(
    (2n+1)\log(\epsilon^{-1})
    +
    \log(\epsilon^{-1})\operatorname{polylog}(\log(\epsilon^{-1}))
    \right),
\]
up to the cost of the sparse-access oracle.

\item If the solution has discontinuities not only at interfaces, then the Hamiltonian-preserving scheme is
only half-order accurate. Thus
\[
    N_{x_\alpha}=N_{v_\alpha}=\mathcal O(\epsilon^{-2}),
    \qquad \alpha=1,\ldots,n.
\]
Consequently,
\[
    \|H\|_{\max}=\mathcal O(n\epsilon^{-3}),
\]
and Lemma \ref{lem:complexity} gives
\[
    \mathcal O\left(
    \left(\sum_{\alpha=1}^n Q_\alpha+n\right)nT\epsilon^{-3}
    +
    \frac{\log(\epsilon^{-1})}{\log\log(\epsilon^{-1})}
    \right)
\]
queries. In this case,
\[
    \sum_{\alpha=1}^n
    \left(n_{x_\alpha}+n_{v_\alpha}\right)
    =
    4n\log(\epsilon^{-1}),
\]
and adding the $p$-register gives the stated additional gate factor.
\end{itemize}
\end{proof}

\begin{remark}
The theorem above keeps the same convention as Theorems \ref{thm:1dim_complexity1} and \ref{thm:2dim_complexity1}, where the length of the
truncated $p$-interval is not included in the leading logarithmic qubit count. If one also accounts for
the growth of the admissible $p$-interval, then, under the estimate
\[
    \lambda^+_{\max}(H_1)+\lambda^-_{\max}(H_1)
    =
    \mathcal O\left(
    \sum_{\alpha=1}^n (N_{x_\alpha}+N_{v_\alpha})
    \right),
\]
one may choose
\[
    N_p
    =
    \mathcal O\left(
    \frac{
    \left(\lambda^+_{\max}(H_1)+\lambda^-_{\max}(H_1)\right)T
    +
    \log(\epsilon^{-1})
    }{\epsilon}
    \right).
\]
Therefore the query complexities remain unchanged, while the logarithmic qubit count becomes
\[
    (2n+2)\log(\epsilon^{-1})+\mathcal O(\log(nT))
\]
in the first-order case, and
\[
    (4n+3)\log(\epsilon^{-1})+\mathcal O(\log(nT))
\]
in the half-order case. If the cost of the quantum Fourier transform in the $p$ variable is listed
separately, one should additionally include $\mathcal O(n_p\log n_p)$ gates.
\end{remark}

\begin{remark}
In a classical implementation, using the forward Euler method as an example, the fully discrete scheme
has the form
\[
    \boldsymbol f^{m+1}=(\Delta t A+I)\boldsymbol f^m+\Delta t\,\boldsymbol b,
    \qquad m=0,\ldots,N_t-1.
\]
The hyperbolic CFL condition requires
\[
    \Delta t
    \max_{\boldsymbol i,\boldsymbol j}
    \sum_{\alpha=1}^n
    \left[
    \frac{|v_{\alpha,j_\alpha}|}{\Delta x_\alpha}
    +
    \frac{
    \left|
    V^{-}_{\boldsymbol i+\frac12\boldsymbol e_\alpha}
    -
    V^{+}_{\boldsymbol i-\frac12\boldsymbol e_\alpha}
    \right|
    }{
    \Delta x_\alpha \Delta v_\alpha
    }
    \right]
    \leq 1 .
\]
Thus $\Delta t=\mathcal O\left(\frac{1}{n}
    \min_{\alpha}\{\Delta x_\alpha,\Delta v_\alpha\}\right)$ and $N_t=\mathcal O(T/\Delta t)$ in the first-order
case. Since the row sparsity of $\Delta t A+I$ is $\mathcal O(n)$, the classical complexity is
\[
    \mathcal O\left(
    n\prod_{\alpha=1}^n N_{x_\alpha}N_{v_\alpha}N_t
    \right).
\]
In the first-order case,
\[
    N_{x_\alpha}=N_{v_\alpha}=\mathcal O(\epsilon^{-1}),
    \qquad
    N_t=\mathcal O(nT\epsilon^{-1}),
\]
and therefore the classical complexity is
\[
    \mathcal O(n^2T\epsilon^{-(2n+1)}).
\]
In the half-order case,
\[
    N_{x_\alpha}=N_{v_\alpha}=\mathcal O(\epsilon^{-2}),
    \qquad
    N_t=\mathcal O(nT\epsilon^{-2}),
\]
so the classical complexity is
\[
    \mathcal O(n^2T\epsilon^{-(4n+2)}).
\]
Therefore, in the sparse-access oracle model and up to the cost of implementing the oracles, the
Schr\"odingerization based quantum algorithm gives a {\it polynomial}  advantage in the error $\epsilon$ and
{\it avoids the exponential dependence on the dimension $n$}  suffered by the corresponding classical
grid based Hamiltonian-preserving scheme.
\end{remark}

\section{Recovery of the physical solution state and quantum readout of macroscopic observables}\label{sec:recovery}

This section provides another component that is not simply part of the sparse Hamiltonian simulation step. Since the Schr\"odingerization procedure produces an enlarged quantum state rather
than the physical density vector directly, one must specify how the physical
information is extracted from this state. This can be done either by postselecting
and normalizing the physical state \(\ket{\boldsymbol f(t)}\), or by estimating suitable overlaps directly
on the enlarged state. We therefore analyze the postselected recovery, the associated
normalization and postselection overhead, and the overlap based readout of
macroscopic observables such as density, momentum, and averaged velocity.

\subsection{Postselected recovery of the physical solution state}
\label{subsec:postselected-recovery-state}

We now clarify how the amplitude encoded physical solution state is obtained from the output of the Schr\"odingerization method. This point is important because the one-point recovery formula is a statement about vectors, while a quantum computer produces normalized quantum states.

Let $N_{\rm ph}$ denote the number of phase-space grid points. In the one-dimensional case,
$N_{\rm ph}=N_xN_\xi$, while in the $n$-dimensional case, $N_{\rm ph}=\prod_{\alpha=1}^n N_{x_\alpha}N_{v_\alpha}.$
After the inhomogeneous system \eqref{equ:nD_ODE} is rewritten in homogeneous form \eqref{equ:interface1homogeneous}, we write
\[
    \boldsymbol u(t)=
    \begin{pmatrix}
        \boldsymbol f(t)\\
        \boldsymbol r(t)
    \end{pmatrix}
    =|0\rangle\otimes \boldsymbol f(t)+|1\rangle\otimes \boldsymbol r(t)
    \in \mathbb{C}^{2N_{\rm ph}} .
\]
Here $\boldsymbol f(t)$ is the physical solution vector and $\boldsymbol r(t)$ is the auxiliary component
introduced by the homogeneous embedding.

By the warped phase transformation, in the recovery region $p>0$, the transformed variable satisfies $\boldsymbol v(t,p)=e^{-p}\boldsymbol u(t)$
at the continuous level. In the fully discretized and simulated problem, this identity is affected by the
truncation of the auxiliary $p$-domain, the discretization in the $p$-variable, and the Hamiltonian
simulation error. Let $p^\diamond=p_{k^\diamond}>0$ be a chosen recovery grid point. We write
\[
        \boldsymbol v(t,p^\diamond)
        =e^{-p^\diamond}\boldsymbol u(t)+\boldsymbol e_{\rm rec}(t),
        \qquad
        \|\boldsymbol e_{\rm rec}(t)\|_2
        \le \varepsilon_{\rm rec}\,\|\boldsymbol u(t)\|_2 .
\]
Here $\varepsilon_{\rm rec}$ collects the truncation error in the auxiliary $p$-domain, the discretization
error in the $p$-variable, and the Hamiltonian simulation error, measured relative to the size of the
original solution $\boldsymbol u(t)$. Since the ideal recovery slice has norm
$e^{-p^\diamond}\|\boldsymbol u(t)\|_2$, the effective relative error on the selected slice is
$\eta_{\rm rec}:=e^{p^\diamond}\varepsilon_{\rm rec}.$
In the following estimates we assume $\eta_{\rm rec}<1$.

After the Hamiltonian simulation of the Schr\"odingerized system and the inverse quantum Fourier
transform in the $p$-register, the normalized quantum state corresponding to the discrete vector
$\boldsymbol w(t)$ is
\begin{equation}
     |\boldsymbol w(t)\rangle
 =
 \frac{1}{C_{\boldsymbol w}(t)}
 \sum_{i=0}^{2N_{\rm ph}-1}\sum_{k=0}^{N_p-1}
 v_i(t,p_k)|i\rangle |k\rangle ,
 \qquad
 C_{\boldsymbol w}(t)=\|\boldsymbol w(t)\|_2 .
 \label{equ:enlarged-state}
\end{equation}
For the exact Hamiltonian simulation, and with the unitary normalization of the quantum Fourier
transform, $C_{\boldsymbol w}(t)=C_{\boldsymbol w}(0)$. In an implementation this equality holds up to the
Hamiltonian simulation error already included in $\varepsilon_{\rm rec}$.

The one-point recovery is implemented by postselecting the $p$-register onto
\[
        \Pi_{p^\diamond}=I_{2N_{\rm ph}}\otimes |k^\diamond\rangle\langle k^\diamond|,
        \qquad p_{k^\diamond}=p^\diamond .
\]
The success probability of this step is
\[
        P_{p^\diamond}(t)
        =
        \|\Pi_{p^\diamond}|\boldsymbol w(t)\rangle\|_2^2
        =
        \frac{\|\boldsymbol v(t,p^\diamond)\|_2^2}{C_{\boldsymbol w}(t)^2}.
\]
Conditioned on this outcome and after discarding the $p$-register, the remaining state is
\[
        |U(t)\rangle
        =
        \frac{|\boldsymbol v(t,p^\diamond)\rangle}{\|\boldsymbol v(t,p^\diamond)\|_2}
        =
        \frac{\boldsymbol u(t)}{\|\boldsymbol u(t)\|_2}
        +\mathcal O(\eta_{\rm rec}).
\]
The error above is understood in the vector norm, up to an irrelevant global phase.

We then project onto the physical $\boldsymbol f$-block. Let
\[
        \Pi_{\boldsymbol f}=|0\rangle\langle 0|\otimes I_{N_{\rm ph}} .
\]
Conditioned on the successful recovery of $p=p^\diamond$, the success probability of this second
projection is
\[
        P_{\boldsymbol f}(t\mid p^\diamond)
        =
        \|\Pi_{\boldsymbol f} |U(t)\rangle\|_2^2
        =
        \frac{\|\boldsymbol f(t)\|_2^2}{\|\boldsymbol u(t)\|_2^2}
        +\mathcal O(\eta_{\rm rec}).
\]
Let $C_{\boldsymbol f}(t)=\|\boldsymbol f(t)\|_2.$
If $\frac{\|\boldsymbol u(t)\|_2}{C_{\boldsymbol f}(t)}$ is bounded and $\frac{\|\boldsymbol u(t)\|_2}{C_{\boldsymbol f}(t)}\eta_{\rm rec}<1$, then conditioned on this
postselection and after discarding the fixed block-label qubit, the prepared physical solution state is
\[
        |F(t)\rangle
        =
        \frac{\boldsymbol f(t)}{\|\boldsymbol f(t)\|_2}
        +\mathcal O\!\left(\frac{\|\boldsymbol u(t)\|_2}{C_{\boldsymbol f}(t)}\eta_{\rm rec}\right).
\]
Thus, in the one-dimensional case,
\[
    |F(t)\rangle
    =
    \frac{1}{C_{\boldsymbol f}(t)}
    \sum_{i=1}^{N_x}\sum_{j=1}^{N_\xi}
    f_{ij}(t)
    |i-1\rangle |j-1\rangle
    +
    \mathcal O\!\left(\frac{\|\boldsymbol u(t)\|_2}{C_{\boldsymbol f}(t)}\eta_{\rm rec}\right).
    \label{eq:amplitude-encoded-physical-f}
\]
The total success probability of preparing the physical solution state from the extended state is
\[
\begin{aligned}
        P_{p^\diamond,\boldsymbol f}(t)
        & =
        \|\left(\Pi_{\boldsymbol f}\otimes |k^\diamond\rangle\langle k^\diamond|\right)
        |\boldsymbol w(t)\rangle\|_2^2 \\
        & =
        \frac{\|\Pi_{\boldsymbol f}\boldsymbol v(t,p^\diamond)\|_2^2}{C_{\boldsymbol w}(t)^2} \\
        & =
        \frac{e^{-2p^\diamond}C_{\boldsymbol f}(t)^2}{C_{\boldsymbol w}(t)^2}
        \left(1+\mathcal O\!\left(\frac{\|\boldsymbol u(t)\|_2}{C_{\boldsymbol f}(t)}\eta_{\rm rec}\right)\right).
\end{aligned}
\]
Consequently, if $C_{\boldsymbol w}(0)$ is known from state preparation, then
\begin{equation}\label{equ:cf}
        C_{\boldsymbol f}(t)
        =
        e^{p^\diamond}C_{\boldsymbol w}(0)\sqrt{P_{p^\diamond,\boldsymbol f}(t)}
        \left(1+\mathcal O\!\left(\frac{\|\boldsymbol u(t)\|_2}{C_{\boldsymbol f}(t)}\eta_{\rm rec}\right)\right).    
\end{equation}
This normalization factor is needed for unnormalized moments, such as the density. It cancels in ratios
such as the averaged velocity.

\begin{remark}\label{rmk:postselected-AA}
The preparation of $|F(t)\rangle$ described above is a postselected state preparation. It is not obtained deterministically by the one-point recovery formula alone. The one-point recovery corresponds to projecting the $p$-register onto the recovery slice $p=p^\diamond$, and the physical solution is obtained by an additional projection onto the $\boldsymbol f$-block of the homogeneous embedding. A naive implementation based on repeated postselection incurs an overhead of order \(P_{p^\diamond,\boldsymbol f}(t)^{-1}\). In principle, if coherent access to the full state preparation circuit and its inverse is available, amplitude amplification can reduce this dependence to order \(P_{p^\diamond,\boldsymbol f}(t)^{-1/2}\). However, this requires repeated coherent applications of the Hamiltonian simulation circuit for the Schr\"odingerized system and its inverse, and may be costly in practice when the simulation circuit is long.
\end{remark}

\subsection{Readout of density, momentum, and averaged velocity}
\label{subsec:quantum-observable-readout}

After the Schr\"odingerization and the postselected recovery, suppose that the physical solution component has been prepared as the normalized amplitude encoded state
\begin{equation}
    |F(t)\rangle
    =
    \frac{1}{C_{\boldsymbol f}(t)}
    \sum_{i=1}^{N_x}\sum_{j=1}^{N_\xi}
    f_{ij}(t)\,
    |i-1\rangle |j-1\rangle,
    \qquad
    C_{\boldsymbol f}(t)
    =
    \left(\sum_{i=1}^{N_x}\sum_{j=1}^{N_\xi}|f_{ij}(t)|^2\right)^{1/2}.
    \label{equ:amplitude-encoded-f}
\end{equation}
If the state obtained from the recovery procedure has an error $\varepsilon_F=\mathcal O(\frac{\|\boldsymbol u(t)\|_2}{C_{\boldsymbol f}(t)}\eta_{\rm rec})$, then all overlap estimates below acquire an additional additive error of order $\varepsilon_F$. To simplify notation, the formulas below are written for the ideal state in \eqref{equ:amplitude-encoded-f}.

Here $f_{ij}(t)$ denotes the cell average of the numerical solution. We first describe the one-dimensional case. The density and momentum at the spatial cell $x_i$ are approximated by
\begin{equation*}
    \rho_i(t)
    =
    \Delta \xi \sum_{j=1}^{N_\xi} f_{ij}(t),
    \qquad
    m_i(t)
    =
    \Delta \xi \sum_{j=1}^{N_\xi} \xi_j f_{ij}(t),
    \label{equ:discrete-density-momentum}
\end{equation*}
and the averaged velocity is
\begin{equation*}
    u_i(t)=\frac{m_i(t)}{\rho_i(t)},
    \qquad \rho_i(t)\neq 0.
    \label{equ:discrete-averaged-velocity}
\end{equation*}

Define the two normalized moment states
\begin{equation}
    |\chi_i^{(0)}\rangle
    =
    |i-1\rangle
    \otimes
    \frac{1}{\sqrt{N_\xi}}
    \sum_{j=1}^{N_\xi}|j-1\rangle,
    \label{equ:zeroth-moment-state}
\end{equation}
and
\begin{equation}
    |\chi_i^{(1)}\rangle
    =
    |i-1\rangle
    \otimes
    \frac{1}{\sqrt{\Xi_2}}
    \sum_{j=1}^{N_\xi}\xi_j |j-1\rangle,
    \qquad
    \Xi_2=\sum_{j=1}^{N_\xi}\xi_j^2 .
    \label{equ:first-moment-state}
\end{equation}
When the recovered solution is real-valued, up to numerical errors, the density and momentum can be written as linear functionals of the solution vector, which can be estimated through overlaps with the amplitude encoded state:
\begin{equation}
    \rho_i(t)
    =
    C_{\boldsymbol f}(t)\,\Delta \xi\,\sqrt{N_\xi}\,
    \operatorname{Re}\langle \chi_i^{(0)}|F(t)\rangle,
    \label{equ:density-overlap}
\end{equation}
and
\begin{equation}
    m_i(t)
    =
    C_{\boldsymbol f}(t)\,\Delta \xi\,\sqrt{\Xi_2}\,
    \operatorname{Re}\langle \chi_i^{(1)}|F(t)\rangle .
    \label{equ:momentum-overlap}
\end{equation}
The real part is taken here because the physical solution $f_{ij}(t)$ should be real and Schr\"odingerization or numerical errors may introduce a very small imaginary part, so the $\operatorname{Re}$ is used to extract the physical solution.
Therefore the averaged velocity can be obtained from the ratio
\begin{equation}
    u_i(t)
    =
    \frac{\sqrt{\Xi_2}\,
    \operatorname{Re}\langle \chi_i^{(1)}|F(t)\rangle}
    {\sqrt{N_\xi}\,
    \operatorname{Re}\langle \chi_i^{(0)}|F(t)\rangle},
    \qquad \rho_i(t)\neq 0.
    \label{equ:velocity-overlap}
\end{equation}
The normalization factor $C_{\boldsymbol f}(t)$ is needed for the unnormalized value of the density, but it cancels in the averaged velocity. This factor may be tracked from the normalization constants and success probabilities in the state-preparation and recovery procedures as in \ref{equ:cf}. 

The overlaps in \eqref{equ:density-overlap}--\eqref{equ:momentum-overlap} can be estimated by standard Hadamard tests. More explicitly, for $\ell=0,1$, prepare
\begin{equation}
    |\Psi_i^{(\ell)}(t)\rangle
    =
    \frac{1}{\sqrt{2}}
    \left(
    |0\rangle|\chi_i^{(\ell)}\rangle
    +
    |1\rangle|F(t)\rangle
    \right).
    \label{equ:observable-test-state}
\end{equation}
Given $X\otimes I$ and $|\Psi_i^{(\ell)}(t)\rangle$, one can use (real) Hadamard test to estimate the real part of
\begin{equation}
    \bra{\Psi_i^{(\ell)}(t)}X\otimes I\ket{\Psi_i^{(\ell)}(t)}
    =
    \operatorname{Re}\langle \chi_i^{(\ell)}|F(t)\rangle.
    \label{equ:pauli-observable-overlap}
\end{equation}
Thus the macroscopic observables can be estimated without reconstructing the full phase-space vector $\boldsymbol f(t)$.
Moreover, amplitude estimation can also be applied to estimate the normalized overlaps. For the state \eqref{equ:observable-test-state} applying a Hadamard gate to the ancilla qubit and measuring the probability of the ancilla qubit to be in state $\ket{0}$ give $$p_0^{(\ell)}=\Pr(\text{ ancilla qubit }=0)=\frac{1+\mathrm{Re}\langle\chi_i^{(\ell)}|F(t)\rangle}2.$$
Hence $$\mathrm{Re}\langle\chi_i^{(\ell)}|F(t)\rangle=2p_0^{(\ell)}-1.$$
If amplitude estimation is used, each overlap can be estimated to additive precision $\varepsilon_{\rm obs}$ using $O\!\left(\varepsilon_{\rm obs}^{-1}\right)$ uses of the preparation unitaries for $|F(t)\rangle$ and $|\chi_i^{(\ell)}\rangle$, and their inverses.
Direct sampling of the Hadamard test instead gives the usual $O(\varepsilon_{\rm obs}^{-2})$ sampling cost.
\begin{remark}
    For the density, when the recovered distribution $f_{ij}(t)$ is known to be nonnegative, one may also estimate 
    \(|\langle \chi_i^{(0)}|F(t)\rangle|^2\) by applying the inverse preparation unitary $U_\chi$ of 
    \(|\chi_i^{(0)}\rangle\) on $|F(t)\rangle$ and measuring the all-zero outcome. Then the probability of getting $\ket{0\cdots 0}$ is $p_i=|\langle0|U_\chi^\dagger|F(t)\rangle|^2=|\langle\chi_i^{(0)}|F(t)\rangle|^2$. In this special case the overlap $\operatorname{Re}\langle \chi_i^{(0)}|F(t)\rangle$ is nonnegative and can be recovered by taking the square root. For momentum, however, the overlap may have either sign, so an interference based overlap estimation method, such as a Hadamard test, is needed to obtain its real part.
\end{remark}
\begin{remark}
    Equivalently, the macroscopic observables can be read directly from the enlarged state \eqref{equ:enlarged-state} before explicitly preparing the postselected physical state. Define
    \[
    |\zeta_i^{(\ell)}\rangle
    =
    |0\rangle\otimes |\chi_i^{(\ell)} \rangle\otimes|k^\diamond\rangle,
    \qquad \ell=0,1 .
    \]
    Then, in the recovery region,
    \[
    \rho_i(t)
    =
    e^{p^\diamond}C_{\boldsymbol w}(t)\Delta\xi\sqrt{N_\xi}
    \operatorname{Re}\langle \zeta_i^{(0)}|\boldsymbol w(t)\rangle,
    \]
    and
    \[
    m_i(t)
    =
    e^{p^\diamond}C_{\boldsymbol w}(t)\Delta\xi\sqrt{\Xi_2}
    \operatorname{Re}\langle \zeta_i^{(1)}|\boldsymbol w(t)\rangle .
    \]
    This formulation avoids explicitly preparing the normalized physical state
    \(|F(t)\rangle\). However, the relevant overlaps are smaller by the factor
    \(\sqrt{P_{p^\diamond,\boldsymbol f}(t)}\):
    $$\langle\zeta_i^{(\ell)}|\boldsymbol  w(t)\rangle
    =\frac{e^{-p^\diamond}C_{\boldsymbol f}(t)}{C_{\boldsymbol{w}}(t)}\langle\chi_i^{(\ell)}|F(t)\rangle
    \approx\sqrt{P_{p^\diamond,\boldsymbol f}(t)}\langle\chi_i^{(\ell)}|F(t)\rangle,$$
    so the postselection overhead reappears as a more
    stringent overlap precision requirement. For example, achieving the same physical accuracy requires the enlarged state overlap to be estimated with an accuracy smaller by a factor
    \(\sqrt{P_{p^\diamond,\boldsymbol f}(t)}\). Hence direct sampling incurs an overhead
    of order \(1/P_{p^\diamond,\boldsymbol f}(t)\), while amplitude estimation incurs an
    overhead of order \(1/\sqrt{P_{p^\diamond,\boldsymbol f}(t)}\).
    Consequently, in view of the amplitude amplification discussion in Remark \ref{rmk:postselected-AA}, the two formulations have comparable dependence on
    \(P_{p^\diamond,\boldsymbol f}(t)\) in their readout costs.
\end{remark}

It is important to distinguish the precision of the normalized overlaps from the precision of the physical observables. An additive overlap error $\varepsilon_{\rm obs}$ gives an additive error in the physical observable:
\[
    |\delta \rho_i|
    \lesssim
    C_{\boldsymbol f}(t)\Delta\xi\sqrt{N_\xi}\,
    \bigl(\varepsilon_{\rm obs}+\varepsilon_F\bigr),
    \qquad
    |\delta m_i|
    \lesssim
    C_{\boldsymbol f}(t)\Delta\xi\sqrt{\Xi_2}\,
    \bigl(\varepsilon_{\rm obs}+\varepsilon_F\bigr).
\]
Therefore, for a fixed prescribed absolute accuracy in $\rho_i$ or $m_i$, the required overlap precision must be chosen after accounting for the prefactors above. Similarly, the ratio in \eqref{equ:velocity-overlap} is stable only when the denominator, equivalently $\rho_i(t)$, is bounded away from zero. For a fixed number of spatial locations and for a prescribed overlap error, the sampling part of the readout does not scale with the phase-space grid size; however, a fixed error in the physical observable may introduce additional dependence through these normalization factors. If one wants to output all values $\{\rho_i(t),u_i(t)\}_{i=1}^{N_x}$ as classical data, then an $\Omega(N_x)$ cost is unavoidable because of the output size.

For the $n$-dimensional problem, let
\[
    N_v=\prod_{\alpha=1}^n N_{v_\alpha},
    \qquad
    \Delta v=\prod_{\alpha=1}^n \Delta v_\alpha,
\]
and write the recovered solution state as
\begin{equation}
    |F(t)\rangle
    =
    \frac{1}{C_{\boldsymbol f}(t)}
    \sum_{\boldsymbol i\in \mathcal I_x}\sum_{\boldsymbol j\in \mathcal I_v}
    f_{\boldsymbol i,\boldsymbol j}(t)|\boldsymbol i,\boldsymbol j\rangle,
    \qquad
    C_{\boldsymbol f}(t)=
    \left(\sum_{\boldsymbol i\in \mathcal I_x}\sum_{\boldsymbol j\in \mathcal I_v}|f_{\boldsymbol i,\boldsymbol j}(t)|^2\right)^{1/2} .
    \label{equ:nd-amplitude-encoded-f}
\end{equation}
For a fixed spatial multi-index $\boldsymbol i$, the density and the $\alpha$-th momentum component are approximated by
\[
    \rho_{\boldsymbol i}(t)= \Delta v\sum_{\boldsymbol j\in \mathcal I_v} f_{\boldsymbol i,\boldsymbol j}(t),
    \qquad
    m_{\alpha,\boldsymbol i}(t)= \Delta v\sum_{\boldsymbol j\in \mathcal I_v} v_{\alpha,j_\alpha}f_{\boldsymbol i,\boldsymbol j}(t).
\]
Define
\begin{equation}
    |\chi_{\boldsymbol i}^{(0)}\rangle
    =
    |\boldsymbol i\rangle
    \otimes
    \frac{1}{\sqrt{N_v}}
    \sum_{\boldsymbol j\in \mathcal I_v}|\boldsymbol j\rangle ,
    \label{equ:nd-density-state}
\end{equation}
and, for the $\alpha$-th component of the averaged velocity,
\begin{equation}
    |\chi_{\boldsymbol i,\alpha}^{(1)}\rangle
    =
    |\boldsymbol i\rangle
    \otimes
    \frac{1}{\sqrt{\Xi_{2,\alpha}}}
    \sum_{\boldsymbol j\in \mathcal I_v} v_{\alpha,j_\alpha}|\boldsymbol j\rangle,
    \qquad
    \Xi_{2,\alpha}
    =
    \sum_{\boldsymbol j\in \mathcal I_v} v_{\alpha,j_\alpha}^2
    =
    \left(\prod_{\beta\neq \alpha}N_{v_\beta}\right)
    \sum_{j_\alpha=1}^{N_{v_\alpha}}v_{\alpha,j_\alpha}^2 .
    \label{equ:nd-velocity-state}
\end{equation}
Then
\begin{equation}
    \rho_{\boldsymbol i}(t)
    =
    C_{\boldsymbol f}(t)\,\Delta v\,\sqrt{N_v}\,
    \operatorname{Re}\langle \chi_{\boldsymbol i}^{(0)}|F(t)\rangle,
    \label{equ:nd-density-overlap}
\end{equation}
and
\begin{equation}
    m_{\alpha,\boldsymbol i}(t)
    =
    C_{\boldsymbol f}(t)\,\Delta v\,\sqrt{\Xi_{2,\alpha}}\,
    \operatorname{Re}\langle \chi_{\boldsymbol i,\alpha}^{(1)}|F(t)\rangle.
    \label{equ:nd-momentum-overlap}
\end{equation}
The $\alpha$-th component of the averaged velocity is therefore
\begin{equation}
    u_{\alpha,\boldsymbol i}(t)
    =
    \frac{m_{\alpha,\boldsymbol i}(t)}{\rho_{\boldsymbol i}(t)}
    =
    \frac{\sqrt{\Xi_{2,\alpha}}\,
    \operatorname{Re}\langle \chi_{\boldsymbol i,\alpha}^{(1)}|F(t)\rangle}
    {\sqrt{N_v}\,
    \operatorname{Re}\langle \chi_{\boldsymbol i}^{(0)}|F(t)\rangle},
    \qquad \rho_{\boldsymbol i}(t)\neq 0.
    \label{equ:nd-averaged-velocity-overlap}
\end{equation}

The same distinction between overlap error and physical observable error applies in the $n$-dimensional case. If the overlaps in \eqref{equ:nd-density-overlap} and \eqref{equ:nd-momentum-overlap} are estimated with additive error $\varepsilon_{\rm obs}$, and the recovered state has additive error $\varepsilon_F$, then
\[
|\delta \rho_{\boldsymbol i}|
\lesssim
C_{\boldsymbol f}(t)\Delta v\sqrt{N_v}
(\varepsilon_{\rm obs}+\varepsilon_F),
\qquad
|\delta m_{\alpha,\boldsymbol i}|
\lesssim
C_{\boldsymbol f}(t)\Delta v\sqrt{\Xi_{2,\alpha}}
(\varepsilon_{\rm obs}+\varepsilon_F).
\]
Hence a prescribed absolute accuracy for the physical density or momentum component requires choosing the overlap precision after accounting for the dimension-dependent prefactors $\Delta v\sqrt{N_v}$ and $\Delta v\sqrt{\Xi_{2,\alpha}}$. The averaged velocity component \eqref{equ:nd-averaged-velocity-overlap} is stable only when $\rho_{\boldsymbol i}(t)$ is bounded away from zero. For a fixed number of spatial multi-indices and velocity components, the sampling cost of each overlap estimate does not scale with the total phase-space grid size, apart from state preparation, normalization, and postselection costs. If all values $\{\rho_{\boldsymbol i}(t),u_{\alpha,\boldsymbol i}(t)\}_{\boldsymbol i,\alpha}$ are required as classical output, then an output-size cost of at least $\Omega(n\prod_{\alpha=1}^n N_{x_\alpha})$ is unavoidable. Therefore the quantum advantage is most relevant when one needs the encoded solution state or a limited number of observables, rather than all grid values of the macroscopic observable fields as classical output.

\begin{remark}
A direct computational basis measurement of the \(x\)-register of \(|F(t)\rangle\) gives, in one spatial dimension,
\(
\sum_{j=1}^{N_\xi}\frac{|f_{ij}(t)|^2}{C_{\boldsymbol f}(t)^2},
\)
and analogously
\(
\sum_{\boldsymbol j\in \mathcal I_v}
\frac{|f_{\boldsymbol i,\boldsymbol j}(t)|^2}
{C_{\boldsymbol f}(t)^2}
\)
in multiple dimensions. This marginal probability is not the density
\(
\rho_i(t)=\Delta \xi\sum_j f_{ij}(t).
\)
Therefore, under the amplitude encoding used in this work, the density, momentum, and averaged velocity are naturally treated as linear overlap observables, unless a different encoding, such as a square root density encoding, is used. Such an encoding, however, is not produced by the present Schr\"odingerization procedure.

\end{remark}

\section{Numerical examples}\label{sec:Numerical}
In this section, we present numerical examples to validate the proposed Schr\"odingerization based Hamiltonian-preserving formulation. All numerical tests are carried out on a classical computer. Therefore, the Hamiltonian simulation $e^{-iHt}$ is emulated by standard time discretizations, such as the backward Euler method or the Crank-Nicolson method. These experiments are not intended to demonstrate the asymptotic quantum advantage, which would require fault-tolerant quantum hardware. Instead, they are designed as a classical validation of the proposed formulation.

A direct classical emulation requires storing the auxiliary $p$-dimension together with the full phase-space discretization. In one space dimension, the memory scales as $N_xN_\xi N_p$, while in two spatial dimensions it scales as $N_xN_yN_\xi N_\eta N_p$. Moreover, the admissible interval in the auxiliary $p$-dimension grows when the phase-space grids are refined, since the spectral scale of $H_1$ increases with the mesh resolution. Consequently, for a fixed accuracy in the $p$-discretization, the required number of auxiliary grid points $N_p$ also increases. As a result, full asymptotic refinement studies of the Schrodingerized system are beyond the capability of the classical hardware used here, whose memory is limited to 512 GB. The numerical experiments below therefore focus on structural validation: agreement with the classical Hamiltonian-preserving scheme, correct transmission and reflection behavior at interfaces, and the computation of physically relevant observables such as density and averaged velocity. 

\begin{example}(\cite{jinHamiltonianPreservingSchemesLiouville2005})\label{ex:ex1}
    A 1D problem with exact $L^\infty$-solution. Consider the 1D Liouville equation \eqref{equ:1dimLiouvilleClassical}
    with a discontinuous potential given by
    \begin{equation}
    V(x)=\begin{cases}
        0.2,&x<0,\\
        0,&x>0.
    \end{cases}
    \end{equation}
    The initial data are given by
    \begin{equation}
    f(x,\xi,0)=\begin{cases}
        1,&x\le0,\xi>0,\sqrt{x^2+\xi^2}<1,\\
        1,&x\ge0,\xi<0,\sqrt{x^2+\xi^2}<1,\\
        0,&\mathrm{otherwise}.
    \end{cases}
    \end{equation}
\end{example}
The exact solution for $f$ at $t=1$ is given by
\begin{equation}
    f(x,\xi,1)=\begin{cases}
        1,&x\geq0,\quad\xi<\sqrt{0.4},\quad\xi>x;\\
        1,&x\geq0,\quad\xi<0,\quad x<1,\quad\xi>\frac{x-\sqrt{2-x^2}}{2};\\
        1,&x\leq0,\quad\xi<x,\quad\xi>-\sqrt{0.6},\quad x<(1-\frac{\sqrt{0.6-\xi^{2}}}{\sqrt{\xi^{2}+0.4}})\xi;\\
        1,&x\leq0,\quad\xi>0,\quad x>-1,\quad\xi<\frac{x+\sqrt{2-x^{2}}}{2};\\
        1,&x\geq0,\quad\xi>\sqrt{0.4},\quad\xi>x,\quad\xi<\sqrt{1.4},\quad x>(1-\frac{\sqrt{1.4-\xi^{2}}}{\sqrt{\xi^{2}-0.4}})\xi;\\
        0,&\text{otherwise,}
    \end{cases}
\end{equation}
as shown in the first column of Figure \ref{fig:I2ex1solution}.

The computational domain is chosen as $[x, \xi] \in [-1.5, 1.5]\times[-1.5, 1.5]$. The cell number is set by $N_x=N_\xi=2^7$. In the extended space $p$, since $\lambda_{\mathrm{max}}^{+}(H_1)=1.5962$ and $\lambda_{\mathrm{max}}^{-}(H_1)=126.93$, we take the truncated interval in $p$ as $[-131.93,6.5962]$ with $N_p=2^{14}$ and $\Delta p=0.00845$. We choose the recovery point $p^\diamond=1.8615\ge  \lambda^+_{\max}(H_1)T=1.5962$ and time step as $\Delta t =0.02$. 

Figure \ref{fig:I2ex1solution} shows the exact, classical and quantum numerical solution $f$. 

\begin{figure}[!htb]
    \centering
    \includegraphics[width=1.\linewidth]{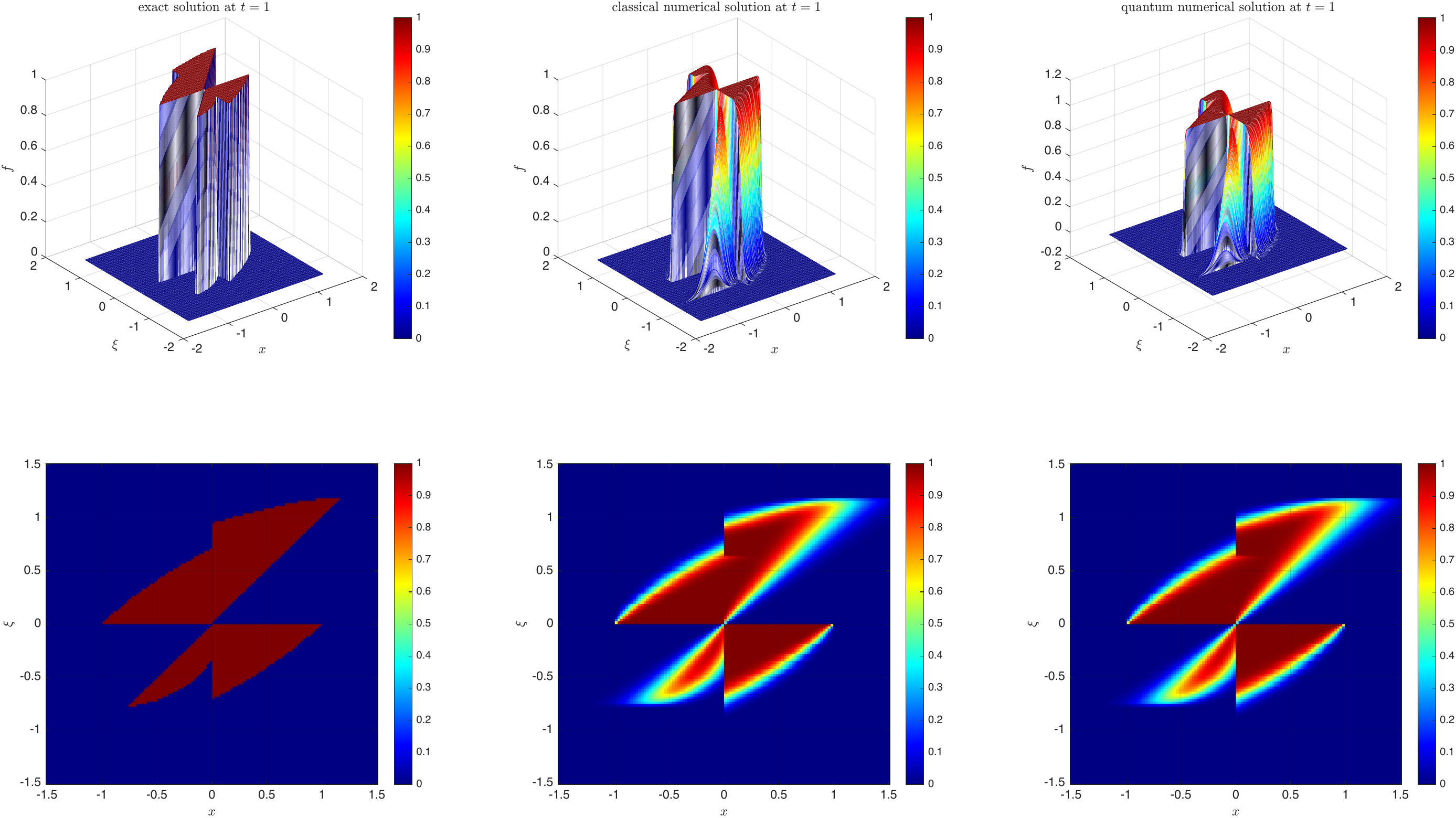}
    \caption{Example \ref{ex:ex1}, the density distribution function $f(x,\xi,t)$ at $t=1$. First row: 3D plot; second row: contour plot. First column: the exact solution; second column: the classical numerical solution; third column: the Schr\"odingerization solution.}
    \label{fig:I2ex1solution}
\end{figure}


\begin{example}(\cite{jinHamiltonianPreservingSchemesLiouville2005})\label{ex:ex2}
    Computing the physical observables of a 1D problem with measure-valued solution.  As mentioned in the Introduction, such problems arise in the computation of the semiclassical limit of the Schr\"odinger equation. Consider the same problem as in example \ref{ex:ex1}, with the initial data
    \begin{equation}
    f(x,\xi,0)=\delta(\xi-w(x)),
    \end{equation}
    where
    \begin{equation}
    w(x)=\begin{cases}0.9,&\quad x\leq-2;\\0.9-\frac{0.9}{4}(x+2)^2,&\quad-2<x\leq0;\\-0.9+\frac{0.9}{4}(x-2)^2;&\quad0<x<2.\\-0.9,&\quad x\geq2,
    \end{cases}
    \end{equation}
    which is plotted as blue dashed lines in the first column of Figure \ref{fig:ex2solution}. Such problems arise in the semiclassical limit of the linear Schr\"odinger equation.
\end{example}
In this example we are interested in the approximation of the moments, such as the density
\[\rho(x,t)=\int f(x,\xi,t)d\xi,\]
and the averaged velocity
\[u(x,t)=\frac{\int f(x,\xi,t)\xi d\xi}{\int f(x,\xi,t)d\xi}.\]

These quantities are computed by decomposition techniques described in \cite{jinHamiltonianPreservingSchemesLiouville2005}. We first solve the level set function $\psi$ and modified density function $\phi$ which satisfy the Liouville equation \eqref{equ:1dimLiouvilleClassical} with initial data $\xi-w(x)$ and $1$, respectively. Then the desired physical observables $\rho$ and $u$ are computed from the numerical singular integrals 
\begin{gather}
    \rho(x,t)=\int f(x,\xi,t)\mathrm{d}\xi=\int\phi(x,\xi,t)\delta(\psi)\mathrm{d}\xi,\\
    u(x,t)=\frac{1}{\rho(x,t)}\int f(x,\xi,t)\xi\mathrm{d}\xi=\int\phi(x,\xi,t)\xi\delta(\psi)\mathrm{d}\xi/\rho(x,t),
\end{gather}
which are computed by a discrete delta function 
\begin{equation}\label{equ:discretedeltacos}
    \delta_{\omega}(x)=\begin{cases}
        \frac{1}{2\omega}(1+\cos({\frac{|\pi x|}{\omega}})),&|{\frac{x}{\omega}}|\leq1;\\
        0,&|{\frac{x}{\omega}}|>1.
    \end{cases}
\end{equation}
Thus one only involves approximating the delta-function at the output time, which gives much higher numerical resolution than evolving the delta-function throughout time evolution.  The $\omega$ in \eqref{equ:discretedeltacos} is half of the support size of the discrete delta function. In our computation we choose \[\omega(\boldsymbol x,\boldsymbol v,t)=\max(|\psi_{\boldsymbol v}|,1)h,\] where $|\psi_{\boldsymbol v}|$ denotes the Jacobian determinant of $\Psi=(\psi_j)$ with respect to $\boldsymbol{v}$: \[|\partial\Psi/\partial(v_{1},\cdots,v_{d})|,\] and is approximated by the central differencing. 

The exact velocity profile and the corresponding density at $t = 1.8$ are given in the Appendix of \cite{jinHamiltonianPreservingSchemesLiouville2005}. Figure \ref{fig:ex2solution} shows the exact multivalued velocity.

The computational domain is chosen as $[x, \xi] \in [-2, 2]\times[-1.6, 1.6]$. The cell number is set by $N_x=N_\xi=2^7$. In the extended space $p$, since $\lambda_{\mathrm{max}}^{+}(H_1)=20.588$ and $\lambda_{\mathrm{max}}^{-}(H_1)=120.32$, we take the truncated interval in $p$ as $[-221.58,42.059]$ with $N_p=2^{14}$ and $\Delta p=0.01609$. We choose the recovery point $p^\diamond=6.7554\le \lambda^+_{\max}(H_1)T=37.059$ and time step as $\Delta t =0.03$. 
Although this value is below the conservative sufficient bound $\lambda^+_{\max}(H_1)T$, it lies in a numerically stable recovery region, as discussed in Remark \ref{rem:recovery-point}.

Figure \ref{fig:ex2solution} shows the exact, classical and quantum numerical solution $f$. Figure \ref{fig:ex2rho} shows the classical and quantum numerical density $\rho$ and averaged velocity $u$ along with the exact solutions in the physical space.

\begin{figure}[!htb]
    \centering
    \includegraphics[width=1.\linewidth]{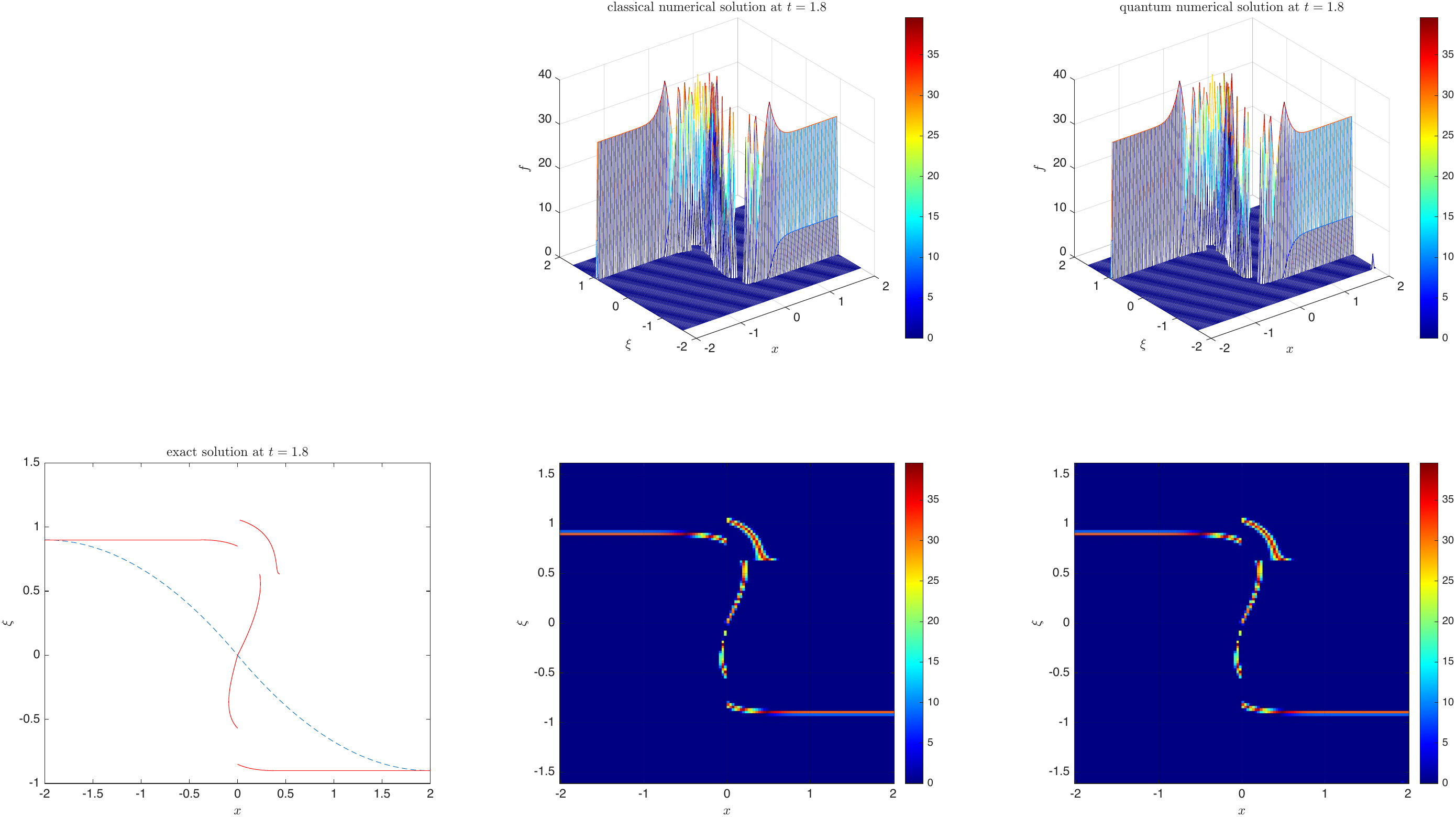}
    \caption{Example \ref{ex:ex2}, the density distribution function $f(x,\xi,t)$ at $t=1.8$. First row: 3D plot; second row: contour plot. First column: the exact solution; second column: the classical numerical solution; third column: the Schr\"odingerization solution.}
    \label{fig:ex2solution}
\end{figure}

\begin{figure}[!htb]
    \centering
    \includegraphics[width=1.\linewidth]{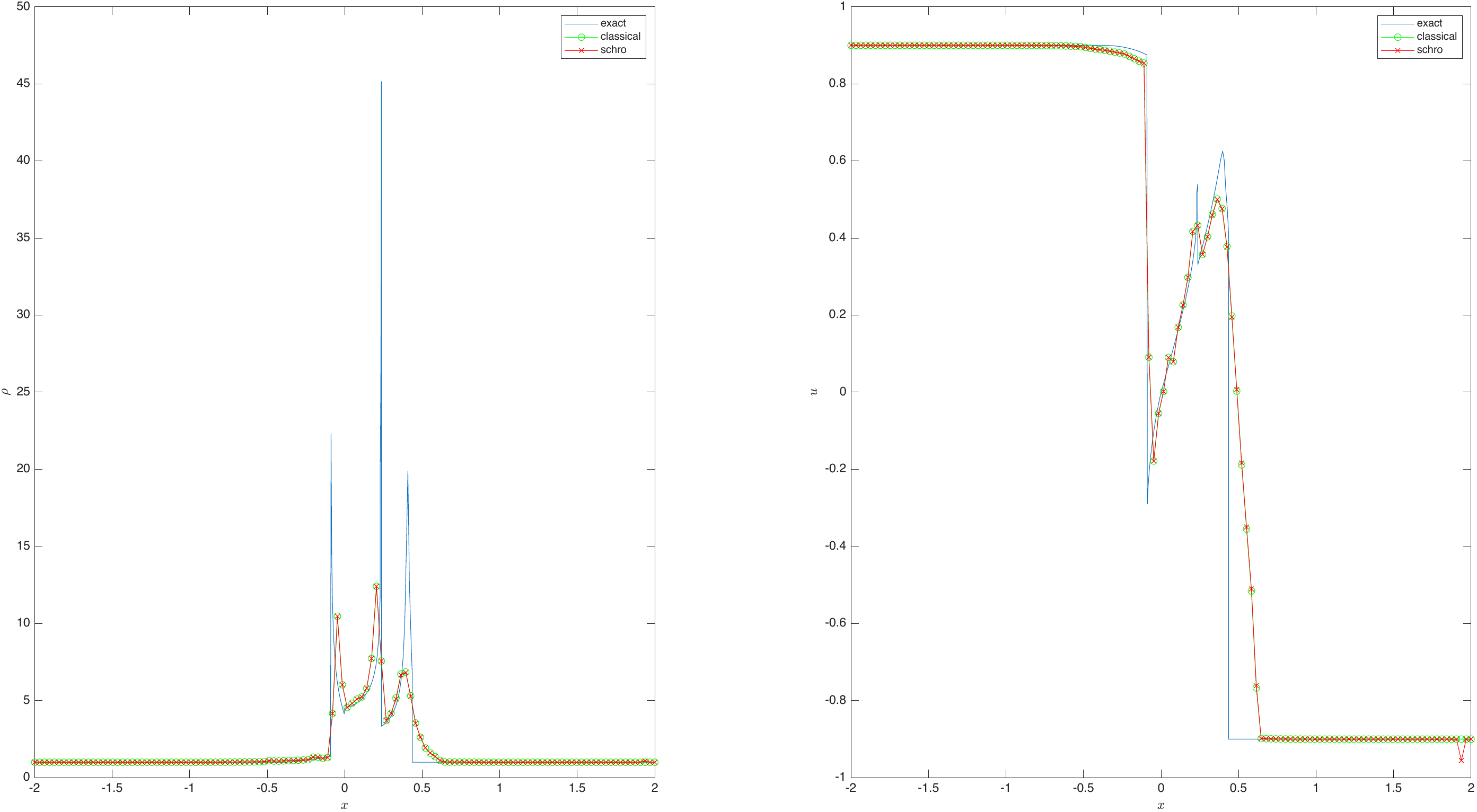}
    \caption{Example \ref{ex:ex2}, the density $\rho$ and averaged velocity $u$ at $t = 1.8$. Solid blue line: the exact solution; green ``$o$": the classical numerical solution; red solid line with ``$\times$": the Schr\"odingerization solution. Left: the density $\rho$; Right: the averaged velocity $u$.}
    \label{fig:ex2rho}
\end{figure}

\begin{example}(\cite{jinHamiltonianPreservingSchemesLiouville2005})\label{ex:ex3}
     Computing the physical observables of a 2D problem with measure-valued solution.  Consider the 2D Liouville equation \eqref{equ:2dimLiouvilleClassical}, with a discontinuous potential given by 
     \begin{equation*}
        V(x,y)=\begin{cases}0.1,&\quad x>0,y>0,\\0,&\quad\mathrm{else},\end{cases}
    \end{equation*}
    and the delta-function initial data
    \begin{equation*}
    f(x,y,\xi,\eta,0)=\rho(x,y,0)\delta(\xi-p(x,y))\delta(\eta-q(x,y)),
    \end{equation*}
    with
    \begin{equation*}
        \begin{aligned}
            \rho(x,y,0)&=\begin{cases}0,&\quad x>-0.1,y>-0.1;\\1,&\quad\mathrm{else},\end{cases}\\
            p(x,y)=q(x,y)&=\begin{cases}0.4,&\quad x>0,y>0,\\0.6,&\quad\mathrm{else}.\end{cases}
        \end{aligned}
    \end{equation*} 
\end{example}

In this example we are interested in the computation of numerical density which is the zeroth moment of this delta-type solution
\begin{equation*}
    \rho(x,y,t)=\iint f(x,y,\xi,\eta,t)d\xi d\eta.
\end{equation*}

The exact density at $t = 0.4$ is
\begin{equation*}
    \rho(x,y,0.4)=\begin{cases}1,&\quad x<0\text{ or }y<0;\\1.5,&\quad0\leq x\leq14/150,y\geq\frac{3x}{2};\\1.5,&\quad0\leq y\leq14/150,y\leq\frac{2x}{3};\\0,&\quad\mathrm{otherwise},\end{cases}
\end{equation*}
as shown in the upper left part in Figure \ref{fig:ex3rho}.

The computational domain is chosen to be $[x,y,\xi,\eta] \in [-0.2,0.2] \times [-0.2,0.2] \times [0.3,0.9] \times [0.3,0.9]$. The cell number is set by $N_x=N_y=N_\xi=N_\eta=2^3$. In the extended space $p$, since $\lambda_{\mathrm{max}}^{+}(H_1)=0.4763$ and $\lambda_{\mathrm{max}}^{-}(H_1)=61.882$, we take the truncated interval in $p$ as $[-29.753,5.1905]$ with $N_p=2^{12}$ and $\Delta p=0.008531$. However, we choose the recovery point $p^\diamond=0.1913>  \lambda^+_{\max}(H_1)T=0.1905$ and time step as $\Delta t =0.008$. 

Figure \ref{fig:ex3rho} shows the exact, classical and quantum numerical solution density $\rho$.

\begin{figure}[!htb]
    \centering
    \includegraphics[width=1.\linewidth]{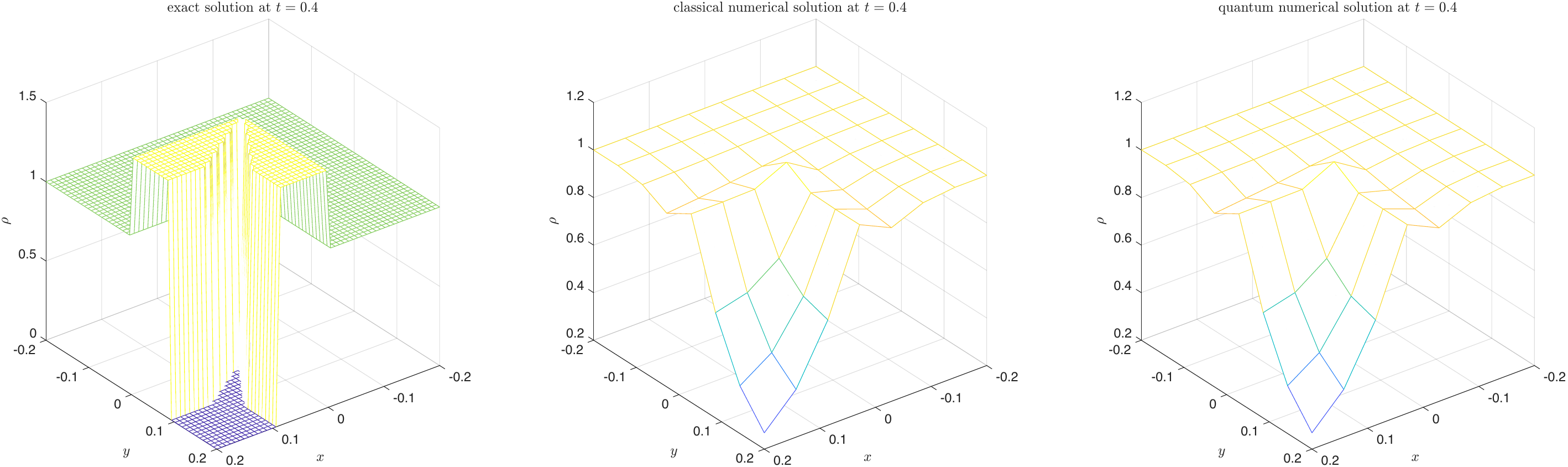}
    \caption{Example \ref{ex:ex3}, the density $\rho$ at $t = 0.4$. First column: the exact solution; second column: the classical numerical solution; third column: the Schr\"odingerization solution.}
    \label{fig:ex3rho}
\end{figure}

In all of the above examples, we showed that the solutions approximated by the Schr\"odingerization method match quite well with those by classical method, demonstrating the consistency of the Schr\"odingerization based formulation. 

This two-dimensional example illustrates the main limitation of classical emulation. With $N_x=N_y=N_\xi=N_\eta=2^3$ and $N_p=2^{12}$, the computation already approaches the available memory limit. If the grid size is doubled in each phase-space dimension, the auxiliary $p$-grid must also be refined, which increases the memory requirement by approximately a factor of $64$. Therefore, the purpose of this test is not to provide a high resolution classical simulation, but to verify that the Schr\"odingerization based formulation reproduces the qualitative behavior of the classical Hamiltonian-preserving scheme and the expected density profile.

\section{Conclusion} \label{sec:conclusion}
In this paper, we developed a Schr\"odingerization based quantum simulation framework for Liouville equations in classical mechanics with discontinuous potential. The central difficulty is the interface dynamics at potential barriers: the normal velocity must satisfy the classical Hamiltonian conservation law, leading to an energy-threshold transmission/reflection rule. By combining Hamiltonian-preserving schemes with Schr\"odingerization, we converted the semi-discrete Liouville equations with interface conditions into sparse Hamiltonian systems.

For one and two dimensional problems, we constructed explicit matrix forms for the Hamiltonian-preserving numerical fluxes and showed how the transmission and reflection mechanisms can be encoded through step functions and compactly supported hat functions. We then extended the construction to \(n\) spatial dimensions under the assumption that the discontinuous interfaces are aligned with the Cartesian grids. In the sparse-access oracle model, the resulting Hamiltonians have sparsity controlled by the interface sparsity parameters \(Q_\alpha\) and the dimension \(n\). The query complexity estimates show {\it polynomial} dependence on the target accuracy and {\it avoid the exponential dependence on the phase-space dimension}  that appears in the corresponding classical grid based Hamiltonian-preserving schemes, up to the cost of implementing the sparse-access oracles.  

We also clarified how the physical solution is recovered from the extended quantum state. We further described how macroscopic observables, including density, momentum, and averaged velocity, can be obtained from overlap estimation without full reconstruction of the phase-space solution. The numerical experiments, carried out through classical simulation of the Schr\"odingerized dynamics, confirm that the proposed formulation reproduces the transmission/reflection behavior of the Hamiltonian-preserving scheme and yields consistent physical observables.

Several issues remain for future work. The present multidimensional construction assumes that the discontinuous interfaces are aligned with the computational grids. Extending the method to general interface geometries and curved barriers is an important next step. 
It would also be useful to analyze the full end-to-end error, including postselection, normalization, and observable readout, to develop lower cost readout strategies for macroscopic observables, and to study the implementation cost of the sparse-access oracles on fault-tolerant quantum architectures. Together with the geometrical optics case studied in \cite{jinQuantumSimulationLiouville2026}, the present work suggests that Schr\"odingerization combined with Hamiltonian-preserving interface discretizations provides a general framework for quantum simulation of kinetic equations with interface conditions.  
This perspective also motivates extensions to related kinetic models with interface conditions, such as radiative transfer equations.

\section*{CRediT authorship contribution statement}
{\bf Shi Jin}: Writing – review \& editing, Writing - original draft, Supervision, Resources, Methodology, Funding acquisition, Conceptualization; 
{\bf Shuyi Zhang}: Writing – review \& editing, Writing - original draft, Visualization, Validation, Software, Methodology, Investigation, Formal analysis, Conceptualization.

\section*{Data availability}
Data will be made available on request.

\section*{Declaration of competing interest}
The authors declare that they have no known competing financial interests or personal relationships that could have appeared to influence the work reported in this paper.

\section*{Acknowledgements}
The research results of this article are sponsored by the Kunshan Municipal Government research funding.  Shuyi Zhang thanks Yue Yu at Xiangtan University for helpful and stimulating discussions.


\appendix
\section{Estimation of sparsity in column}\label{app:sparsity}

\subsection{Estimation of $s_c(B_i^l)$ and $s_c(B_i^u)$}\label{app:B_i}

It suffices to prove the case when $V^-_{i+\frac{1}{2}}>V^+_{i+\frac{1}{2}}$, the case when $V^-_{i+\frac{1}{2}}<V^+_{i+\frac{1}{2}}$ follows similarly. 

When $0<\xi_j\le\sqrt{2\left(V_{i+\frac{1}{2}}^--V_{i+\frac{1}{2}}^+\right)}$, particles with velocity $\xi_j$ can only arise from reflected particles with velocity $-\xi_j$. Thus $\beta_{i+\frac12,j,k}=0$ for all $j,k$ satisfying $0<\xi_j\le\sqrt{2\left(V_{i+\frac{1}{2}}^--V_{i+\frac{1}{2}}^+\right)}$. When $\xi_j>\sqrt{2\left(V_{i+\frac{1}{2}}^--V_{i+\frac{1}{2}}^+\right)}$, particles with velocity $\xi_j$ can only originate from transmitted particles with velocity $\xi_j^-:=\sqrt{\xi_j^2+2\left(V_{i+\frac12}^+-V_{i+\frac12}^-\right)}$ whose density distribution is interpolated by the values at two adjacent grid points. If $\xi_j^-$ lies in $[\xi_{k},\xi_{k+1})$, then $\xi_{j+1}^-$ must lie in $[\xi_{k+l},\xi_{k+1+l})$ for $l\ge1$ because 
\[
    \xi_{j+1}^--\xi_{j}^-=
    \frac{\xi_{j+1}+\xi_j}{\sqrt{\xi_{j+1}^2+2\left(V_{i+\frac12}^+-V_{i+\frac12}^-\right)}+\sqrt{\xi_j^2+2\left(V_{i+\frac12}^+-V_{i+\frac12}^-\right)}}\Delta\xi>\Delta\xi.
\]
In other words, at most one $\xi_j^-$ lies in every interval $[\xi_{k},\xi_{k+1})$. 
This proves $s_c(B_i^l)\le2$.

Then we estimate $s_c(B_i^u)$. Define $f(\xi)=-\sqrt{\xi^2+2\left(V_{i+\frac12}^--V_{i+\frac12}^+\right)}$ for $\xi<0$, then $f'(\xi)=-\frac{\xi}{\sqrt{\xi^2+2\left(V_{i+\frac12}^--V_{i+\frac12}^+\right)}}\in(0,1)$ which is monotonically decreasing in $(-\infty,0)$. Thus we can infer that the points $\xi_j^+=-\sqrt{\xi_j^2+2\left(V_{i+\frac12}^--V_{i+\frac12}^+\right)}$ become denser as $\xi_j$ approaches $0$. Suppose there is a $\xi_{j'}^+$ such that $\xi_{\frac{N_\xi}{2}}^+-\xi_{j'}^+<\Delta \xi$, then
\begin{align*}
    \sqrt{(\frac{N_\xi+1}{2}-j')^2\Delta\xi^2+2\left(V_{i+\frac12}^--V_{i+\frac12}^+\right)}&<\Delta\xi+\sqrt{(\frac{\Delta\xi}{2})^2+2\left(V_{i+\frac12}^--V_{i+\frac12}^+\right)}\\
    (\frac{N_\xi+1}{2}-j')^2&<\frac{5}{4}+\frac{2}{\Delta\xi}\sqrt{(\frac{\Delta\xi}{2})^2+2\left(V_{i+\frac12}^--V_{i+\frac12}^+\right)}\\
    \frac{N_\xi}{2}-j'&<\sqrt{\frac{5}{4}+\frac{2}{\Delta\xi}\sqrt{(\frac{\Delta\xi}{2})^2+2\left(V_{i+\frac12}^--V_{i+\frac12}^+\right)}}-\frac{1}{2}.
\end{align*}
In other words, at most $\frac{N_\xi}{2}-j'+1$ $\xi_j^+$'s lies in some interval $[\xi_{k},\xi_{k+1})$, where $j=j^{\prime},\ldots,\frac{N_\xi}{2}$. 
This proves $s_c(B_i^u)\le2\left\lceil\sqrt{\frac54+\frac{2}{\Delta\xi}\sqrt{\left(\frac{\Delta\xi}{2}\right)^2+2\left(V_{i+\frac{1}{2}}^--V_{i+\frac{1}{2}}^+\right)}}-\frac12\right\rceil$.

\subsection{Estimation of $s_c(A_1)$}\label{app:A_1}

\begin{multline*}
    s_c(A_1)\le \max_{x\in \mathcal{I}}\{s_c(B_i^u+B_{i+1}^l)+s_c(C_i)\}=\max_{x\in \mathcal{I}}\{\max\{s_c(B_i^u),s_c(B_{i+1}^l)\}+s_c(C_i)\}\le\\
    2\left\lceil\sqrt{\frac54+\frac{2}{\Delta\xi}\sqrt{\left(\frac{\Delta\xi}{2}\right)^2+2\max_{x\in\mathcal{I}}\left\{|V^-(x)-V^+(x)|\right\}}}-\frac12\right\rceil+2,
\end{multline*}
where the equality holds because \(B_i^u\) and \(B_{i+1}^l\) act on different halves of the velocity grid, so their nonzero column supports do not overlap.

\bibliographystyle{plain}
\bibliography{ref}

\end{document}